\documentclass[12pt,journal]{article}

\usepackage{arxiv}

\usepackage[utf8]{inputenc} 
\usepackage[T1]{fontenc}    
\usepackage{hyperref}       
\usepackage{url}            
\usepackage{booktabs}       
\usepackage{amsfonts}       
\usepackage{nicefrac}       
\usepackage{microtype}      
\usepackage{lipsum}		    
\usepackage{graphicx}
\usepackage{doi}

\usepackage{tikz,pgfplots}

\usepackage{mathtools,amsmath,amsfonts,amssymb,amsthm,subcaption}
\usepackage{xcolor}
\usepackage{enumitem}
\usepackage[ruled]{algorithm2e}
\usepackage{soul}
\def\R{\mathbb{R}}
\def\Rp{\R_+}
\def\C{\mathbb{C}}
\def\N{\mathbb{N}}
\def\opt{\star}

\def\LInfO{L_\mu^{\infty}(\Omega)}
\def\LInf{L_\mu^{\infty}}

\def\Ns{n_s}
\def\Nl{n_\ell}
\def\cs{c_s}
\def\S{\mathcal{S}}

\def\P{\mathcal{P}}
\def\Nb{n_b}

\def\cmax{\gamma_{\text{max}}}

\def\Phip{\Phi_\eps^+}
\def\Phim{\Phi_\eps^-}
\def\vphip{\vphi_\eps^+}
\def\vphim{\vphi_\eps^-}
\def\Ae{A_{\eps}}
\def\Ap{\Ae^+}
\def\Am{\Ae^-}
\def\ERB{\text{ERB}}
\def\ERBS{\text{ERBS}}
\def\fopt{f^{\star}}

\def\eps{\varepsilon}
\def\vphi{\varphi}
\def\Po{P_0}
\def\Pe{P_\eps}
\def\tPe{\tilde{P}_\eps}
\newcommand{\set}[1]{\{#1\}}
\newcommand{\Lset}[1]{\left\{#1\right\}}
\newcommand{\nrm}[1]{\|#1\|}
\newcommand{\wh}[1]{\widehat{#1}}
\newcommand{\ol}[1]{\overline{#1}}
\newcommand{\nrmW}[1]{\nrm{#1}_{W}}
\theoremstyle{plain}\newtheorem{proposition}{Proposition}
\theoremstyle{plain}\newtheorem{lemma}{Lemma}
\theoremstyle{plain}
\theoremstyle{plain}

\allowdisplaybreaks

\title{Towards Maximizing a Perceptual Sweet Spot}

\author{ \href{https://orcid.org/0000-0002-3178-4413}{\includegraphics[scale=0.06]{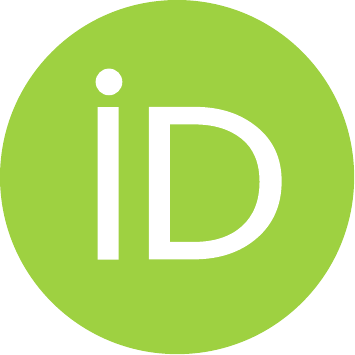}\hspace{1mm}Pedro~Izquierdo~Lehmann} \\
    Institute for Mathematical and Computational Engineering\\
    Pontificia Universidad Cat\'olica de Chile\\
	\texttt{pizquierdo2@uc.cl} \\
	\And
	\href{https://orcid.org/0000-0001-5902-1170}{\includegraphics[scale=0.06]{orcid.pdf}\hspace{1mm}Rodrigo~F.~C\'adiz} \\
	Department of Electrical Engineering and\\
    Music Institute\\
	Pontificia Universidad Cat\'olica de Chile\\
	\texttt{rcadiz@uc.cl} \\
	\And
    \href{https://orcid.org/0000-0002-2533-2509}{\includegraphics[scale=0.06]{orcid.pdf}\hspace{1mm}Carlos~A.~Sing~Long}\thanks{C.~A.~Sing~Long is also with ANID – Millennium Science Initiative Program – Millennium Nucleus Center for the Discovery of Structures in Complex Data, and with ANID – Millennium Science Initiative Program – Millennium Nucleus Center for Cardiovascular Magnetic Resonance.}\\
    Institute for Mathematical and Computational Engineering and\\
    Institute for Biological and Medical Enginering\\
    Pontificia Universidad Cat\'olica de Chile\\
    \texttt{casinglo@uc.cl}\\
}

\renewcommand{\headeright}{}
\renewcommand{\undertitle}{Preprint}
\renewcommand{\shorttitle}{}

\hypersetup{
    pdftitle={Towards Maximizing a Perceptual Sweet Spot},
    pdfsubject={eess.AP, eess.SP},
    pdfauthor={Pedro~Izquierdo~Lehmann, Rodrigo~F.~C\'adiz, Carlos~A.~Sing~Long},
    pdfkeywords={Spatial sound, sound field reconstruction, psychoacoustics, sweet spot, applied functional analysis, non-convex optimization, DC optimization},
}

\begin{document}
\maketitle

\begin{abstract}
    The {\em sweet spot} can be interpreted as the region where acoustic sources create a spatial auditory illusion.
    We study the problem of maximizing this sweet spot when reproducing a desired sound wave using an array of loudspeakers. To achieve this, we introduce a theoretical framework for spatial sound perception that can be used to define a sweet spot, and we develop a method that aims to generate a sound wave that directly maximizes the sweet spot defined by a model within this framework. Our method aims to incorporate perceptual principles from the onset and is flexible: while it imposes little to no constraints on the regions of interest, the arrangement of loudspeakers or their radiation pattern, it allows for audio perception models that include state-of-the-art monaural perceptual models. Proof-of-concept experiments show that our method outperforms state-of-the-art methods when comparing them in terms of their localization and coloration properties.
\end{abstract}


\parskip = 6pt

\section{Introduction}

The field of spatial sound addresses the question: {\em how do we create a desired {\em spatial auditory illusion} over a spatial region of interest with a set of acoustic sources?}~\cite[Chapter 2.3]{nicol2020creating}. The {\em spatial auditory illusion} (SAI) occurs when acoustic sources create a sound scene that produces a desired auditory scene over a region. It is related to {\em sound quality} as described by Wierstof et al~\cite[Chapter 2]{wierstorf2013binaural}: {\em ``The quality of a system as perceived by a listener is considered to be the result of assessing perceived features with regard to the desired features (of the auditory scene).''} Following Blauert, the {\em sound scene} represents the objective nature of a sound wave propagating in the physical world, whereas the {\em auditory scene} represents the imprint of the sound scene in our subjectivity, that is, the result of the auditory system perceiving and organizing sound into meaning~\cite{blauert1997spatial}. Over the last century, several methods have been proposed to answer this. Their performance can be compared in terms of the size of the region where the SAI is achieved. In this work, we call this region the {\em sweet spot}. 

The term {\em sweet spot} is already used in panning and surround systems to describe an ideal listening position that is equally distant to all loudspeakers~\cite{spors2013} and around which there is a limited area where a desired wavefront is correctly recreated~\cite[Chapter 3.1]{nicol2020creating},~\cite{leakey1959some}. It is also used to mean the {\em ``area in which the spatial perception of the auditory scene works without considerable impairments''}~\cite[Chapter 1.2]{wierstorf2014perceptual}. Furthermore, in~\cite{frank2017exploring} the {\em sweet area} is defined as the {\em ``area within which the reproduced sound scene is perceived as plausible,''} where {\em plausible} means the preservation of front localization and envelopment of reverberation. Our use of the term is in the spirit of the second and third ideas. 

One of the earlier and most widespread spatial sound approaches is stereophony
and its generalizations, sourround systems~\cite{rumsey2018} and Vector Base Amplitude Panning (VBAP) \cite{pulkki1997virtual}. These methods, also called {\em panning techniques}, adjust the level and time-delay of the audio signals for each speaker utilizing a \textit{panning law} to simulate steering the perceived direction of the sound source. The possibility of this steering has been explained by the perceptual idea of {\em summing localization} and by the {\em association model}~\cite[Chapter 6.1]{nicol2020creating}. Moreover, due to some psycho-acoustic features of the auditory system, such as the binaural decoloration mechanism, they work sufficiently well in some applications, even with few speakers~\cite{spors2013}; they do not suffer from coloration~\cite{pulkki2001coloration}. However, they can only simulate sound sources that lay on the segments that join the speakers. Furthermore, its quality degrades rapidly as the listener moves away from the center of the target region~\cite{spors2013}.

A popular strategy to recreate an auditory scene is to directly approximate the sound wave that created it. In the literature, this strategy is called {\em sound field synthesis}, {\em sound field reproduction} or {\em sound field reconstruction}. Following Huygens' principle~\cite{huygens}, any sound scene can be approximated accurately with a sufficiently dense arrangement of loudspeakers. However, in practice there is only a limited number of them. Three classes of commonly used methods for sound field reconstruction are {\em mode matching methods}, {\em pressure matching methods} and {\em wave field synthesis}.

Mode Matching Methods (MMM) match the coefficients in the expansion of the target and generated sound waves in spatial spherical harmonics~\cite{daniel2000representation}. Some well-known MMMs are Ambisonics~\cite{gerzon1973periphony}, Higher-Order Ambisonics (HOA), and Near-Field Compensated Ambisonics (NFC-HOA)~\cite{daniel2003further}. All of them minimize the \(\ell^2\)-norm of the difference between the leading coefficients. Ambisonics assumes the loudspeakers emit plane waves and uses only the leading coefficient, whereas HOA uses a larger but fixed number of coefficients. In contrast, NFC-HOA assumes the loudspeakers are monopoles. Ambisonics, HOA and NFC-HOA are designed for circular or spherical regions of interest. When approximating a plane wave, they create a central spherical region with a radius that is inversely proportional to the frequency of the source over which the sound scene is reconstructed almost identically~\cite{Ward2001}.
Some variations of these methods consider a weighted mode matching problem to prioritize certain spatial regions~\cite{Ueno2019}, a mixed pressure-velocity mode matching problem~\cite{Zuo2020}, and an intensity mode matching problem~\cite{zuo20213d}.

Instead of using expansions in spatial spherical harmonics, Pressure Matching Methods (PMM) minimize the spatio-temporal \(L^2\)-error between the the target and generated sound waves~\cite{kirkeby1993reproduction}. The magnitude of the audio signals are often penalized by their \(L^p\)-norm to mitigate the effects of ill-conditioning~\cite{Gauthier2005, Gauthier2017, Jia2018, Feng2018} and to enforce sparse representations~\cite{Radmanesh2013, Lilis2010}. Typically the loudspeakers are modeled as monopoles. In most cases, the solution can only be found numerically, and the discretization of the region of interest plays an important role~\cite{ajdler2006plenacoustic, kolundzija2009sound}. Variations considering a mixed pressure-velocity pressure matching problem have been considered~\cite{buerger2015multizone, buerger2018broadband, Shin2016}.

Finally, Wave Field Synthesis (WFS) leverages the single-layer boundary integral representation of a sound wave over a region of interest~\cite{spors2013}. Traditional WFS~\cite{berkhout1993acoustic} uses a Rayleigh integral representation to derive a solution when the speakers are modeled as dipoles lying on a line. This was later extended to monopoles~\cite{stuart1996application} for 2.5D reproduction. Its reformulation, Revisited WFS~\cite{spors2008theory}, uses a Kirchhoff-Helmholtz integral representation along with a Neumann boundary condition to obtain a solution for an arbitrary distribution of monopoles. It has been shown that the localization properties of the auditory scene are correctly simulated by WFS and do not depend on the position of the listener
over the region of interest~\cite{wierstorf2017assessing}. However, this technique suffers from coloration effects due to spatial aliasing artifacts~\cite{wierstorf2014coloration}.

There is extensive literature analyzing these methods and comparing their performance~\cite{daniel2003further, spors2008comparison, fazi2009analogies, Franck2017, Firtha2018}, they become equivalent in the limit of a continuum of loudspeakers, differing only when using a finite number~\cite{fazi2007theoretical} of them. Although they are amenable to mathematical analysis and have computationally efficient implementations,
their construction has no natural perceptual justification to produce a large sweet spot. As a consequence, the artifacts introduced by these methods, due to approximation errors, may produce noticeable, and possibly avoidable, perceptual artifacts.

An alternative to better reproduce the auditory scene is to explicitly account for psycho-acoustic and perceptual principles in the reconstruction methods~\cite{wierstorf2014perceptual}. The first steps in this direction were taken in~\cite{johnston2000perceptual} by proposing a simple model that aims to preserve the spatial properties of the desired auditory scene. A method to reproduce an active intensity field that is largely uniform in space was then proposed in~\cite{desena2013}. It is based on an optimization problem yielding audio signals where at most two loudspeakers are active simultaneously. However, it makes the restrictive assumption that the target sound wave is a plane wave, and that the loudspeakers emit plane waves. In~\cite{Ziemer2017} the {\em radiation method} and the {\em precedence fade} are proposed. The former is equivalent to applying a PMM over a selection of frequencies that are most relevant psycho-acoustically, whereas the latter is a method to overcome the localization problems associated to the {\em precedence effect}~\cite{litovsky1999precedence}. Finally, in~\cite{Lee2020} a PMM is extended to account for psycho-acoustic effects by considering the \(L^2\)-norm of the differences in pressure convolved in time by a suitable filter.

We believe that there is a gap between methods that aim to directly approximate a sound wave to reproduce a desired auditory scene, and methods that leverage perceptual models to reproduce the same auditory scene. Defining the sweet spot requires a model, either theoretical or empirical, of audio perception. In this work, we introduce theoretical framework for for spatial audio perception, and we develop a method to maximize the sweet spot defined by a model within this framework. Our method is amenable to mathematical analysis, it has an efficient computational implementation, and is guided by perceptual principles. Our numerical results show that our method outperforms some state-of-the-art methods for sound field reconstruction.

The paper is organized as follows. In Section~\ref{sec:mathematicalModel} we introduce the physical assumptions we make, and a theoretical framework for spatial audio perception, deferring to Appendix~\ref{sec:analysisOfSWEET} the technical details. Then, in Section~\ref{sec:sweetReLU} we introduce an intuitive and readily implementable instance of our method to maximize the sweet spot defined by a model within this framework. In Section~\ref{sec:psychAcousticFramework} we discuss the psycho-acoustic concepts that, to our knowledge, can be incorporated in the theoretical framework. In Section~\ref{sec:implementation} we present an implementation of our method. In Section~\ref{sec:experiments} we perform proof-of-concept numerical experiments analyzing the performance of our method, comparing its results with WFS, NFC-HOA and PMM. Finally, in Section~\ref{sec:discussion} we discuss our results, the limitations of our method, and some future lines of research.

\section{Spatial Sound Mathematical Framework}
\label{sec:mathematicalModel}

\subsection{Acoustic framework}
\label{sec:acousticFramework}

Consider \(\Ns\) loudspeakers located at positions \(x_1,\ldots, x_{\Ns}\in \R^3\). When the medium is homogeneous and isotropic, and each loudspeaker behaves as an isotropic point source, the physical sound wave \(u\) they generate is represented in frequency as~\cite[Section~2.5.2]{Evans2010}
\begin{equation}
    \label{eq:synthethizedSignalFourierPointSource}
    \wh{u}(f, x) = \sum_{k=1}^{\Ns} \wh{c}_k(f) \frac{e^{-2\pi i c_s^{-1} f \nrm{x - x_k}}}{4\pi\nrm{x - x_k}}
\end{equation}
where \(\cs\) is the speed of sound, \(c_1,\ldots, c_{\Ns}\) are the audio signals driving each loudspeaker, and \(\wh{c}_k\) is the Fourier transform of \(c_k\) in time
\[
    \wh{c}_k(f) := \int c_k(t) e^{-2\pi i f t}\, dt.
\]
To model the spatial radiation pattern of each loudspeaker, along with time-invariant effects such as reverb~\cite{Gauthier2005,betlehem2005theory}, we may use
\begin{equation}
\label{eq:synthethizedSignalFourier}
    \wh{u}(f, x) = \sum_{k=1}^{\Ns} \wh{c}_k(f)G_k(f,x),
\end{equation}
where \(G_k\) is the Green function of the \(k\)-th loudspeaker. In addition to the array, we consider a region of interest \(\Omega\subset \R^3\) such that \(x_k\notin \ol{\Omega}\); thus, it contains no singularity in~\eqref{eq:synthethizedSignalFourierPointSource}. On this region, we may approximate a sound wave \(u_0\) with the array of loudspeakers {\em as best as possible}. If we had a continuum of isotropic point sources on \(\partial\Omega\) then, under suitable conditions, the \textit{simple source formulation}~\cite[Section~8.7]{williams1999fourier} 
shows we can reproduce \(u_0\) {\em exactly}. However, when only a finite number of physical loudspeakers are available, we must find 
\(\wh{c}_1,\ldots,\wh{c}_{\Ns}\) such that
\begin{equation}
\label{eq:approximationProblemInformal}
    \wh{u}_0(f,x) \approx \sum_{k=1}^{\Ns} \wh{c}_k(f) G_k(f, x)
\end{equation}
in an suitable sense for \(x\in \Omega\). When each \(G_k\) is real-analytic on its second argument the approximation cannot be exact on {\em any} open set unless \(u_0\) was actually generated by the speaker array~\cite[Corollary~1.2.5]{Krantz2002}. This suggests that perfect sound field reconstruction is impossible, and that the difference can be small only on average. In spite of this, in some subset of \(\Omega\) the approximation can be {\em perceptually} accurate.

\subsection{Perceptual framework}
\label{sec:sweetSpot}

\def\uleft{u^{\ell}}
\def\uright{u^{r}}
\def\uside{u^s}
\def\vleft{v^{\ell}}
\def\vright{v^{r}}
\def\vside{v^s}

\theoremstyle{definition}\newtheorem{assumption}{Assumption}
\def\supess{\operatorname*{\textrm{ess~sup}}}
\def\esssupp{\operatorname*{\textrm{ess~supp}}}
\def\Wu{W^\mu}

\def\clO{\overline{\Omega}}
\def\Th{\Theta}
\def\th{\theta}
\def\vth{\vec{\th}}
\def\LE{L^2}
\def\LER{\LE(\R)}
\def\uo{u_0}
\def\vu{\bar{u}}
\def\vv{\bar{v}}
\def\vG{\bar{G}}
\def\vuo{\vu_0}
\def\tP{\tilde{P}}
\def\Dmap{D}
\def\Lmap{L}

\def\TD{T_{D}}
\def\TL{T_{L}}

The comparison in~\eqref{eq:approximationProblemInformal} is between two {\em physical} quantities. To incorporate {\em perceptual} effects, we formally introduce a theoretical framework and defer the mathematical details to Appendix~\ref{sec:maxSweetSpot}.

The perception of an individual located at \(x\in \Omega\) and looking in the direction represented by a unit vector \(\th\) in \(\R^3\) (or an angle in \(\R^2\)) depends on the relation between the sound wave at the left and right ears. Hence, we consider {\em pairs} of signals \(\uleft, \uright\) so that \(\uside = \uside(t,x,\theta)\) represents the wave that reaches the ear \(s\) of a listener located at \(x\) and looking in the direction \(\th\) at time \(t\). We let \(\uside_{(x,\th)}\) represent the signal at ear \(s\). From now on, let \(\vu\) be this vector signal, and let \(W\) be the space of all such signals. Instead of~\eqref{eq:synthethizedSignalFourier} we consider
\begin{equation}
\label{eq:synthethizedEarSignalFourier}
    \wh{u}^s(f, x,\th) = \sum_{k=1}^{\Ns} \wh{c}_k(f) H_k^s(f,x,\th),
\end{equation}
where \(H_k^{s}\) is the {\em head-related transfer function}~\cite{blauert1997spatial} associating to a wave emitted by the \(k\)-th loudspeaker the sound wave reaching the ear \(s\); this comprises the behavior of the loudspeakers. From now on, we let \(W_S\) be the set of all pairs of signals generated by model~\eqref{eq:synthethizedEarSignalFourier}.

Therefore, the problem is to approximate the fixed target signal \(\vuo\) associated to the sound wave \(u_0\) by a signal \(\vu\) represented as~\eqref{eq:synthethizedEarSignalFourier} that is {\em perceptually close} to \(\vuo\). To model the {\em perceptual dissimilarity} we introduce a map \(\Dmap\) that associates to a pair \(\vu, \vuo\) the function \(\Dmap_{(\vu,\vuo)} = \Dmap_{(\vu,\vuo)}(x,\th)\) that quantifies the dissimilarity between the signals \(\vu\) and \(\vuo\) perceived by a listener located at \(x\) and looking in the direction \(\th\). We do not make strong assumptions on \(\Dmap\) {\em except} that it is {\em convex}: for \(\vu_1,\vu_2\) we have
\begin{equation}
\label{eq:dissimilarityIsConvex}
    \sup_{\th\in\Th}\, D_{(\lambda\vu_1+(1-\lambda)\vu_2,\vuo)}(x,\th)\\
        \leq \lambda \sup_{\th\in\Th}\,D_{(\vu_1,\vuo)}(x,\th) + (1-\lambda)\sup_{\th\in\Th}\, D_{(\vu_2,\vuo)}(x,\th)
\end{equation}
for any \(\lambda\in [0, 1]\). We assume the dissimilarity is negligible if \(\Dmap_{(\vu, \vuo)}(x,\th) \leq 0\). Depending on the application, this may be interpreted as {\em authenticity}, i.e., \(\vu\) is indiscernible from \(\vuo\), or as {\em plausibility}, i.e., some perceived features of \(\vu\) and \(\vuo\) show a context-reasonable correspondence~\cite[Chapter 2]{wierstorf2013binaural}. In Section~\ref{sec:psychAcousticFramework} we discuss some functional forms for \(D\).

Suppose the listeners can look only along some directions of interest \(\Th\). We define the {\em auditory illusion threshold} as
\begin{equation}
\label{eq:defThresholdMap}
    \TD\vu(x) := \sup_{\th\in\Th}\, \Dmap_{(\vu,\vuo)}(x,\th). 
\end{equation}
A listener located at \(x\) will perceive no noticeable differences between \(\vu\) and \(\vuo\) regardless of the direction she is looking if \(T\vu(x) \leq 0\). Hence, the {\em sweet spot}
\begin{equation}
\label{eq:defSweetSpot}
    \S(\vu) = \set{x\in \Omega:\TD\vu(x) \leq 0}
\end{equation}
is the region within \(\Omega\) where a listener does not perceive significant differences between \(\vu\) and \(\vuo\). We define a {\em loudness discomfort threshold} similarly: we assume there is a function \(\Lmap\) that associates to \(\vu\) the function \(\Lmap_{\vu} = L_{\vu}(x, \th)\) quantifying the {\em loudness discomfort} experienced by a listener at a location \(x\) looking in a direction \(\th\). We also assume \(\Lmap\) satisfies~\eqref{eq:dissimilarityIsConvex} and we define the the {\em loudness discomfort threshold}
\[
    \TL\vu(x) := \sup_{\th\in \Th}\, L_{\vu}(x,\th).
\]
We assume that \(\TL\vu(x) \leq 0\) when no discomfort is experienced. Hence, to avoid choosing a signal \(\vu\) that causes discomfort, we restrict our choices to
\begin{equation}
\label{eq:defPainThreshold}
    \P := \set{\vu\in W: \forall x\in\Omega.\ \TL\vu(x) \leq 0}.
\end{equation}
Consequenty, our goal is to find a signal \(\vu\in W_S\) that maximize the {\em weighted area} of the sweet spot \(\mu(\S(\vu))\) while causing no discomfort by solving
\begin{equation*}
    \label{opt:sweetSpotBasic}
    (\Po) \,\, \left\{\,\,
    \begin{aligned}
        & \underset{\vu\in W_S}{\text{maximize}}
        & & \mu(\S(\vu))\\
        & \text{subject to}
        & & \vu\in \P.
    \end{aligned}\right.
\end{equation*}

\section{The SWEET-ReLU method}
\label{sec:sweetReLU}

We introduce a particular instance of our method to approximately solve \((P_0)\) that is applicable in practice and has an intuitive interpretation. We defer the analysis of our general method to Appendix~\ref{sec:analysisOfSWEET}. Let \(h:\R\to\R\) be the Heaviside function, i.e., \(h(t) = 0\) if \(t < 0\) and \(h(t)= 1\) if \(t \geq 0\). Observe that
\[
    \mu(\S(\vu)) = \mu(\Omega) - \int_\Omega h(\TD\vu(x))\, d\mu(x).
\]
Maximizing \(\mu(\S(\vu))\) implies minimizing the second term. This is challenging due to \(h\) being piecewise constant. For \(\eps >0\) we can approximate it by a piecewise linear function \(\tilde{h}:\R\to\R\) such that \(\tilde{h}(t) = 0\) for \(t < 0\), \(\tilde{h}(t) = t\) for \(t\in [0, 1]\) and \(\tilde{h}(t) = 1\) for \(t > 1\). For \(\eps > 0\) let \(\tilde{h}_\eps(t) = \tilde{h}(t/\eps)\) and let \(\vu_1\in W_S\) be arbitrary. We can minimize the approximation
\begin{align*}
    \int_{\Omega} \tilde{h}_\eps(\TD\vu_1(x))\, d\mu &= \left(\int_{\set{x\in\Omega:\TD\vu_1(x) \leq \eps}} + \int_{\set{x\in\Omega:\TD\vu_1(x) > \eps}}\right) \tilde{h}_\eps(\TD\vu_1(x)) d\mu(x)\\
    &= \frac{1}{\eps}\int_{\set{x\in\Omega: \TD\vu_1(x) \leq \eps}} (\TD\vu_1(x))_+ d\mu+ \mu(\set{x\in\Omega: \TD\vu_1(x) > \eps})
\end{align*}
where \((u)_+ = \max\set{0,u}\). We interpret the first integral as the contribution from the region where we may change \(\vu_1\) to decrease the auditory illusion threshold, whereas we interpret the second as the contribution of the region over which it is {\em already} too large. Thus, we proceed in a {\em greedy} manner: we let \(\Omega_1=\Omega\), and then, at the \(k\)-th iteration, we define
\begin{equation}
\label{eq:sweetReLUSet}
    \Omega_{k+1} = \Omega_{k}\cap \set{x\in\Omega:\, \TD\vu_k(x)\leq \eps}
\end{equation}
and solve
\[
    (P^{\text{SReLU}}_{k+1})\,\,\left\{
    \begin{aligned}
        & \underset{\vu\in W_S}{\text{minimize}}
        & & \int_{\Omega_{k+1}} (\TD\vu(x))_+\, d\mu(x)\\
        & \text{subject to}
        & & \vu\in \P
    \end{aligned}\right.
\]
to obtain \(\vu_{k+1}\). The above problem is convex, whence \(\vu_{k+1}\) can be found using efficient algorithms. The sequence \(\set{\vu_k}_{k\in\N}\) has at least one accumulation point, which approximates the solution to \((P_0)\). Once the sequence has converged, we can repeat the procedure for a smaller value of \(\eps\). As the intuition behind our method suggests choosing a \(\vu_1\) for which \(\TD\vu_1 \equiv 0\), e.g. \(\vu_1=\vuo\), we simply choose  \(\Omega_1=\Omega\) in practice. By running the algorithm for increasingly smaller values of \(\eps\) we obtain an increasingly accurate approximation to a solution to \((P_0)\). We call this instance of our method SWEET-ReLU (Algorithm~\ref{algo:ReLUDCA}). We defer a justification of these facts to Appendix~\ref{sec:analysisOfSWEET} and the deduction of this instance to Appendix~\ref{sec:sweetReLUInstance}.

\begin{algorithm}[t]
    \SetKwInOut{Input}{input}
    \SetKwInOut{Set}{set}
    \Input{A decreasing sequence \(\set{\eps_i}\) of positive numbers, positive integers \(n_{\eps}, n_{\max}\)}
    \Set{\(\vu_1^{\star} \gets \vuo\)}
    \For{\(i = 1, \ldots, n_{\eps} - 1\)}{
        \Set{\(\Omega_1 \gets \set{x:\, \vu_i^{\star}(x) \leq \eps_i}\), \(\vu_1 \gets \vu_i^{\star}\)}
        \For{\(j=1,\ldots, n_{\max} - 1\)}{
            \(\Omega_{j+1} \gets \Omega_{j} \cap \{x\in\Omega:T\vu_{j}(x)\leq\eps_i\}\)\\
            \(\vu_{j+1} \gets \textsc{Solve}(P_{j+1}^{\text{SReLU}})\)
        }
        \Set{\(\vu_{i+1}^{\star} \gets \vu_{n_{\max}}\)}
    }
    \Return{\(\vu_{n_\eps}^{\star}\)}
    \caption{SWEET-ReLU}
    \label{algo:ReLUDCA}
\end{algorithm}

\section{Perceptual and psycho-acoustic theory}
\label{sec:psychAcousticFramework}

\def\IR{I_R}
\def\vs{\bar{s}}

The theoretical framework introduced aims to be flexible enough to account for a variety of perceptual and psycho-acoustic models. We now present some of the perceptual and psycho-acoustic considerations that lead to models within this framework.

The {\em spatial auditory illusion} (SAI) is achieved when the imprint of the reproduced sound scene on a listener resembles a desired auditory scene. The formation of the auditory scene depends not only on the signals reaching the listener's ears but both on the listener itself ~\cite{deutsch1983auditory}. Also, it may even be influenced by external visual, tactile, and proprioreceptive stimuli~\cite{francombe2015elicitation}. Formulating a model accounting for all these effects goes beyond the scope of this work. Instead, we focus on proposing maps \(\Dmap\) and \(\Lmap\) representative of an {\em average listener} or {\em worst-case listener}, and motivated by the concept of {\em auditory event}, which is only related to the perceptual processing of the ear signals. The auditory scene is then regarded as the integration of different and separable auditory events.

The process of extracting auditory events from ear signals has been studied in the field of {\em auditory scene analysis} (ASA)~\cite{bregman1994auditory, deutsch1983auditory}. Although the psycho-acoustic and cognitive mechanisms involved are the focus of current research~\cite[Chapter 2.1]{nicol2020creating}, the quantitative modeling of the process has been carried out by the field of {\em computational auditory scene analysis} (CASA)~\cite{wang2006computational}. The analysis of auditory scenes is a combination of bottom-up, or {\em signal driven}, and top-down, or {\em hypothesis driven}, processes~\cite{blauert1999models}. In the former, the physical properties of the input signal, which are processed at the peripheral auditory system, serve as a basis for the formation of auditory events and are summarized under {\em primitive grouping cues}~\cite{bregman1994auditory}. In the top-down process higher cognitive processes such as prior knowledge turn into play to determine which signal components the listeners attend to and how these components are assembled and recognized, involving processes that are referred to as {\em schema-based cues}~\cite{bregman1994auditory}. Although a complete model of auditory events formation should also consider top-down processes, for simplicity we focus only on primitive grouping cues. These allow us to compare \(\vu_{(x,\th)}\) and \(\vu_{0,(x,\th)}\) by accounting mostly for the immediate psycho-acoustic peripheral processing.

The direct comparison between the auditory scene generated by \(\vu\) and \(\vuo\) leads to the {\em full reference {\em (FR)} model}. In contrast, in the {\em internal references {\em (IR)} model} the auditory scene is compared to abstract representations in the listener's memory~\cite[Chapter 2.6]{raake2020binaural}. Although FR models have been shown to have limitations~\cite{raake2013comprehensive}, e.g., it is not always clear which binaural input signal is related to what the listener really desires to hear, we focus exclusively on them for simplicity. Otherwise schema based cues would be needed.

Spatial sound applications identify the most important {\em features}~\cite[Chapter 2]{wierstorf2013binaural} of the auditory scene in a multidimensional approach~\cite[Chapter 3.2]{raake2020binaural}. In Letowski’s simple model~\cite{letowski1989sound}, the auditory scene is described in terms of {\em ``loudness, pitch, (apparent) duration, spatial character (spaciousness), and timbre.''} The last two are selected as the more important for spatial sound applications. The analysis of the most important features for spatial sound applications has been recently refined in~\cite{francombe2017evaluation, reardon2018evaluation}, 
once again calling attention to timbral and spatial features. More specifically, azimuth localization and coloration are widely used features for the perceptual assessment of multi-channel reproduction systems~\cite{wierstorf2014perceptual, pulkki2001coloration, wierstorf2014coloration}. Motivated by these studies, we focus on models that account only for coloration and azimuth localization. The trade-off is that these models can, at best, account for {\em plausibility} of the SAI, more than {\em authenticity}. Since the detailed biophysics of the phenomena are not necessary to model an accurate input-output relation, we only consider functional (phenomenological) models instead of physiological (biophysical) ones~\cite[Chapter 1.2]{dietz2021computational}.

\subsection{Binaural Azimuth Localization}

{\em Azimuth localization} is the estimation of the direction of arrival of the incoming sound in the horizontal plane. In concordance with Lord Rayleigh's {\em duplex theory}~\cite{raleigh1907our}, the literature shows that interaural time differences (ITDs) are the primary azimuth localization cue at low frequencies~\cite{wightman1992dominant} whereas interaural level differences (ILDs) become relevant at high frequencies as to resolve ambiguities in the decoding of ITDs~\cite{dietz2011auditory} which appear over 1.4 kHz~\cite{brughera2013human}. Consequently, binaural models, i.e., models that use both the right and the left ear signals as inputs, are crucial, and most of these models for azimuth localization focus on the extraction of the ITDs. The cross-correlation between the left and
right ear signals~\cite{cherry1956human} is usually used to model the mechanism for extraction of the ITDs, e.g.~\cite{jeffress1948place,lindemann1986extension,blauert1978some,braasch2002localization, nix2006sound}
and azimuth localization. Other models~\cite{dietz2011auditory, pulkki2009functional} have appeared after the cross-correlation model was challenged by physiological findings~\cite{mcalpine2003sound}. Unfortunately, the extraction of the ITDs using cross-correlation (or as in~\cite{dietz2011auditory,pulkki2009functional}) and the extraction of ILDs, lead to dissimilarity maps \(\Dmap\) that do not satisfy~\eqref{eq:dissimilarityIsConvex}.

\subsection{Binaural Coloration}

{\em Coloration} is commonly defined as timbre distortion~\cite{wierstorf2014coloration},~\cite[Chapter 8.1]{nicol2020creating} where timbre is the property that {\em ``enables the listener to judge that two sounds which have, but do not have to have, the same spaciousness, loudness, pitch, and duration are dissimilar''}~\cite{letowski1989sound}. Although timbre could be quantified in a spectro-temporal space, the metric of the timbral space is not known and could be non-trivial~\cite{wierstorf2014coloration}.

Binaural perceptual effects such as binaural unmasking, spatial release from masking~\cite{culling2021binaural}, and binaural decoloration~\cite{bruggen2001sound}, allow the auditory system to improve the quality of the perceived sound in terms of signal-to-noise ratio (SNR) identifiability and coloration. Even though these effects make defining an accurate binaural metric even more challenging, binaural detection and masking models have been developed in the literature~\cite{breebaart2001binaural}. They can be used to detect binaural timbral differences, and have been adapted to account for localization cues~\cite{park2008model}. Furthermore, a model accounting for binaural perceptual attributes, such as coloration and localization, is proposed in~\cite{pulkki1999analyzing}. More recently, a model for binaural coloration using multi-band loudness model weights to analyse the perceptual relevance of frequency components has been developed~\cite{mckenzie2022predicting}. Unfortunately, these binaural methods once again lead to dissimilarity maps that do not satisfy~\eqref{eq:dissimilarityIsConvex}.

\subsection{Monoaural models}

To our knowledge, the binaural models in the literature do not lead to dissimilarities satisfying~\eqref{eq:dissimilarityIsConvex}. In contrast, under suitable assumptions, monoaural models, i.e., models that need just one ear signal as input, do. Furthermore, they can be applied independently over each ear, to then use a {\em worst-case scenario} methodology~\cite{thiede2000peaq} to extend them to binaural signals. We follow this approach and focus on monoaural models as surrogates to capture coloration effects. Monaural spectral localization models have been developed for localization across the sagittal plane, and also for localization across the azimuthal plane~\cite{zakarauskas1993computational}, but, to our knowledge, they do not satisfy~\eqref{eq:dissimilarityIsConvex}.

Models to detect monaural distortion aim to determine when two audio signals \(s_0 = s_0(t)\) and \(s = s(t)\) are perceived as different, and how this perception degrades as a function of the dissimilarities between \(s_0\) and \(s\). To achieve this, two main ideas are used for the estimation of audible distortions: the masked threshold and the comparison of internal representations~\cite{thiede2000peaq}.

The masked threshold compares the error signal \(\eps=s-s_0\) against \(s_0\) using a perceptual distortion function \(D^{\opt}(\eps, s_0)\). The error is assumed to be inaudible if this value is less than a fixed masking threshold~\cite{par2005, painter2000perceptual}. The comparison of internal representations leverages a model for an {\em internal representation} \(s\mapsto \IR(s)\) resulting from the signal transformations performed in the ear. The internal representations are compared using an {\em internal detector} \((\IR(s), \IR(s_0))\mapsto D^{\opt}(\IR(s), \IR(s_0))\) and the difference between the signals is assumed to be perceptible if this value exceeds a given threshold~\cite{Jepsen2008, thiede2000peaq}. These studies do not provide analytical expressions that satisfy~\eqref{eq:dissimilarityIsConvex} for the representation nor for the internal detector. An approximation yielding such expressions is given in~\cite{Plasberg2007}; another simplified model is developed in~\cite{Taal2012}.

The models developed in~\cite{par2005, Plasberg2007, Taal2012} yield monoaural dissimilarity maps \(\Dmap\) that satisfy~\eqref{eq:dissimilarityIsConvex}. These methods can be represented as
\begin{equation}
\label{eq:defAudibleThreshold}
    \Dmap^{\mathrm{m}}(s,s_0) = B_1(s - s_0)+ \ldots+ B_{\Nb}(s - s_0)
\end{equation}
where \(B_1,\dots B_{\Nb}\) are filters of the form
\begin{equation}
\label{eq:hearingBandGeneral}
    B_k(s-s_0) = \int_{\R}\left|\int_{\R} K_{B_k}(t, t') (s- s_0)(t')\, dt'\right|^2 dt
\end{equation}
for a suitable function \(K_{B_k}\) representing a time-variant or time-invariant filter that may depend on \(s_0\) itself. In~\cite{par2005,Taal2012} the filters \(B_k\) represent the auditory distortion over the \(k\)-th auditory filter of the cochlea, whereas in~\cite{Plasberg2007} the sum reduces to only one locally time invariant filter that accounts for the whole auditory distortion. Although monoaural, these models can be used for binaural signals \(\vs\) by taking the worst distortion between ear signals~\cite{thiede2000peaq}
\begin{equation}
\label{eq:binauralExtension}
     \Dmap^{\mathrm{b}}(\vs,\vs_0)=\max\{\Dmap^{\mathrm{m}}(s^{\ell},s_0^{\ell}),\Dmap^{\mathrm{m}}(s^{r},s_0^r)\}.
\end{equation}

\subsection{Discomfort}
\label{sec:loudness}

To model the loudness discomfort \(\Lmap\) we consider empirical evaluations of discomfort. This is a simplification motivated by computational simplicity and also by a small number of comprehensive studies on the subject. Empirical thresholds for loud discomfort levels for sinusoidal signals over a finite set of frequencies have been defined in the literature, e.g. in~\cite{knobel2006nivel, sherlock2005estimates}. Naturally, for a sinusoidal signal of frequency \(f_k\) these can be expressed with the monaural expression
\begin{equation}
\label{eq:criteriaForDiscomfort}
    Q_k(s)=\int_{\R}|\wh{s}(f)|^2\,\rho_k(f)^2 df,
\end{equation}
where \(\rho_k(f)\equiv\rho_k\in\mathbb{R}_+\) is the multiplicative inverse of the threshold at \(f_k\). Finally, for binaural signals, these models can be applied taking the worst discomfort between ears as in~\eqref{eq:binauralExtension}.

\section{Implementation}
\label{sec:implementation}

We provide a proof-of-concept implementation of SWEET-ReLU for approximating a sound wave generated by a (pseudo) sinusoidal, i.e., where the spectrum is concentrated around a single frequency, monopole emitting at a single frequency \(\fopt\).

\subsection{Implementation of the acoustic framework}

 The original sound wave \(u_0\) is assumed to be emitted by an isotropic point source at \(x_0\in\R^3\). Each loudspeaker is assumed to be an isotropic point source. As we consider (pseudo) sinusoidal audio signals, the sound wave \(u\) generated by the array is given by~\eqref{eq:synthethizedSignalFourierPointSource} with
\[
    \wh{c}_k(f) = a_{k}\, e^{-(f - \fopt)^2 / 2\sigma^2}
\]
for coefficients \(a_{k}\in \C\) and a fixed spectral localization parameter \(\sigma \ll 1\).

\subsection{Implementation of the perceptual framework}

Monaural distortion detectability methods cannot represent the necessary features to correctly define an auditory illusion map as described in Section~\ref{sec:psychAcousticFramework}. Hence, they may not be optimal when modelling binaural perception. For instance, they cannot represent explicitly any type of azimuth localization, nor binaural coloration effects. However, they can represent monaural coloration effects, and some yield dissimilarity maps satisfying~\eqref{eq:dissimilarityIsConvex}. For this reason, we use a monoaural model as a proof-of-concept. We use van de Par's spectral psycho-acoustic model~\cite{par2005}. Although it is suboptimal when modeling temporal masking effects, the signals we consider are stationary, whence temporal masking is almost non-existent.

This monoaural model can be applied to binaural signals by using the worst-case as in~\eqref{eq:binauralExtension}. For \(x\in \Omega\) and \(\th\in\Th\) we apply this model to the left and right ear signals \(\vu(x,\th)\) and \(\vuo(x,\th)\). As van de Par's model can be represented by time-invariant filters in~\eqref{eq:hearingBandGeneral}, for \(s\in\set{\ell,r}\) we have that
\[
    B_j u^s(x,\th) = \int_{\R} |(\wh{u}^s - \wh{u}_0^s)(f,x,\th)|^2 \rho_{B_j}(f,x,\th)\, df
\]
where \(\rho_{B_j}\) depends on \(u_0\) as
\[
    \rho_{B_j}(f, x,\th) = \frac{w_{B_j}(f)}{C_A + \int_{\R} |\wh{u}_0^s(f,x,\th)|^2w_{B_j}(f) df}.
\]
The constant \(C_A > 0\) limits the perception of very weak signals in silence. The weight \(w_{B_j}\) is defined as \(w_{B_j} := |\eta \gamma_{f_j}|^2\) where 
\[
     \log_{10}\eta(f) = C_{\eta,0} - C_{\eta,1} f^{-0.8} - C_{\eta,2} (f - 3.3\times 10^{3})^2 + C_{\eta,3} f^4
\]
models the outer and middle ear as proposed by Terhardt~\cite{terhardt1979calculating} with \(C_{\eta,0} = 4.69\), \(C_{\eta,1} = 18.2 \times 10^{1.4}\), \(C_{\eta,2} = 32.5\times 10^{-7} \) and \(C_{\eta, 3} = 5\times 10^{-16}\), and
\[
    \gamma_{f_j}(f) = \left(1+\left(\frac{945\pi(f-f_j)}{48\ERB(f_j)}\right)^2\right)^{-2}
\]
models the filtering property of the basilar membrane in the inner ear at the center frequency \(f_j\), where the Equivalent Rectangular Bandwidth (ERB) of the auditory filter centered at \(f_j\) is \(\ERB(f_j)= 24.7(1 + 4.37\times 10^{-3} f_j)^{-1}\) as suggested by Glasberg and Moore~\cite{glasberg1990derivation}. The center frequencies \(f_j\) are uniformly spaced on the ERB-rate scale \(\ERBS(f) = 21.4 \log (1 + 4.37\times 10^{-3}f)\). In van de Par's model, the distortion is noticeable when its metric is greater or equal to 1. Hence, the monaural dissimilarity map becomes
\begin{align*}
    D^{\mathrm{m}}_{u^s,u^s_0}(x,\th) &= -1 + C_{\Psi}\sum_{j=1}^{n_b}\frac{\int_{\R}|(\wh{u}^s - \wh{u}_0^s)(f,x,\th)|^2w_{B_j}(f) df}{C_A + \int_{\R}  |\wh{u}_0^s(f,x,\th)|^2w_{B_j}(f)df}\\
    &\approx -1 + C_\Psi'\sum_{j=1}^{n_b}\frac{|(\wh{u}^s - \wh{u}_0^s)(\fopt,x,\th)|^2w_{B_j}(\fopt)}{C_A +  w_{B_j}(\fopt)|\wh{u}_0^s(\fopt,x,\th)|^2}
\end{align*}
where \(C'_{\Psi} = 2^{1/4}\pi^{1/2}\sigma C_{\Psi}\) and we used the approximation for (pseudo) sinusoidal signals
\[
    \int_{\R} \vphi(f)|\wh{u}_0^s(f,x,\th)|^2\,df \approx \sqrt{2^{1/2}\pi}\sigma \vphi(\fopt) |\wh{u}_0^s(\fopt,x,\th)|^2
\]
when \(\sigma \ll 1\). The constants \(C'_\Psi\) and \(C_A\) are defined as suggested in~\cite{par2005}. This accounts for the absolute threshold of hearing and the just-noticeable difference in level for sinusoidal signals, which gives, \(C'_\Psi\approx 1.555\) and \(C_A\approx 4.481\) when considering \(n_b=100\) as the number of center frequencies, and \(f_1=20, f_{n_b}=10^3\) as the first and last center frequency. The worst-case extension of this perceptual dissimilarity to binaural signals is given by
\[
    D_{\vu,\vuo}(x,\th) = \max\set{D^{\mathrm{m}}_{u^\ell,u^\ell_0}(x,\th), D^{\mathrm{m}}_{u^r,u^r_0}(x,\th)}.
\]
To model the loudness discomfort we use the experimental results in~\cite{knobel2006nivel} about the discomfort caused by sinusoidal signals. We interpolate the data in this study with cubic splines with natural boundary~\cite[Section~8.6]{quarteroni2010numerical} to obtain a function \(\eta_P>0\). Therefore, we consider
\begin{align*}
    L_{u^s}^{\mathrm{m}}(x,\th) &= -1+C_\Pi\int_\R|\wh{u}^s(f,x,\th)|^2\rho_{\fopt}(f, x,\th) df\\
    &\approx -1+C_\Pi'|\wh{u}^s(\fopt,x,\th)|^2\eta_P(\fopt)^{-1}
\end{align*}
where \(\rho_{\fopt}(f, x,\th)\equiv \gamma_{\fopt}(f)/\eta_P(f)\). Naturally, \(C'_\Pi=1\), as the empirical thresholds in~\cite{knobel2006nivel} are attained when \(L^{\mathrm{m}}u^s(x)\leq 0\). Then, the worst-case extension to binaural signals is
\[
    L_{\vu}(x,\th) = \max\set{L_{\uleft}^{\mathrm{m}} (x,\th),L_{\uright}^{\mathrm{m}} (x,\th)}.
\]
Although this implementation is intended to study single-frequency signals, it can be generalized to multifrequency signals. The generalization of \(D^\mathrm{m}\) is direct and it comprises additional terms in the sum. The generalization of \(L^\mathrm{m}\) to a multifrequency signal could sum the discomfort associated to each frequency, similarly to the way the auditory filter errors are summed in \(D^\mathrm{m}\), or it could use an integrating function as indicated in Section~\ref{sec:loudness}. These are simple heuristics for a proof-of-concept study, and their effectiveness should be validated experimentally.

\subsection{Discretization}

\def\Nl{n_{l}}

In a typical experiment, listeners are seated in a room, and their locations and orientations within the room are determined in advance. In this case, a simplification consists in defining the sweet spot in terms of the number listeners for which the SAI is achieved. If the listeners can be located at finite number of points \(z_1,\ldots, z_{\Nl}\in\Omega\) then we can model the weighted area of the sweet spot as
\begin{equation}
\label{eq:atomicSweetSpot}
    \mu(\S(\vu))= \sum_{\ell = 1}^{\Nl} \TD\vu(z_\ell),
\end{equation}
which is equivalent to assuming \(\mu\) is an atomic measure. In this case, every component of the proof-of-concept implementation can be evaluated either in closed-form, as is the case of the weighted area or the Green function of the loudspeakers, or can be approximated to very high-accuracy, as is the case of the head-related transfer functions (see Section~\ref{sec:experiments} for details). 

The approximation~\eqref{eq:atomicSweetSpot} can be used when there is a discrete number of locations for listeners in a room. In other applications where it is necessary to control a continuum, e.g., when listeners may move across the room, quadrature rules must be used. In this case, the approach to solve \((P_0)\) follows an approximate-then-discretize approach, and numerical errors may have an effect on the sweet spot computed in practice.

\section{Experiments}
\label{sec:experiments}

\begin{figure}
    \centering
    \begin{subfigure}[t]{0.245\textwidth}
        \centering
        \includegraphics[width=\textwidth]{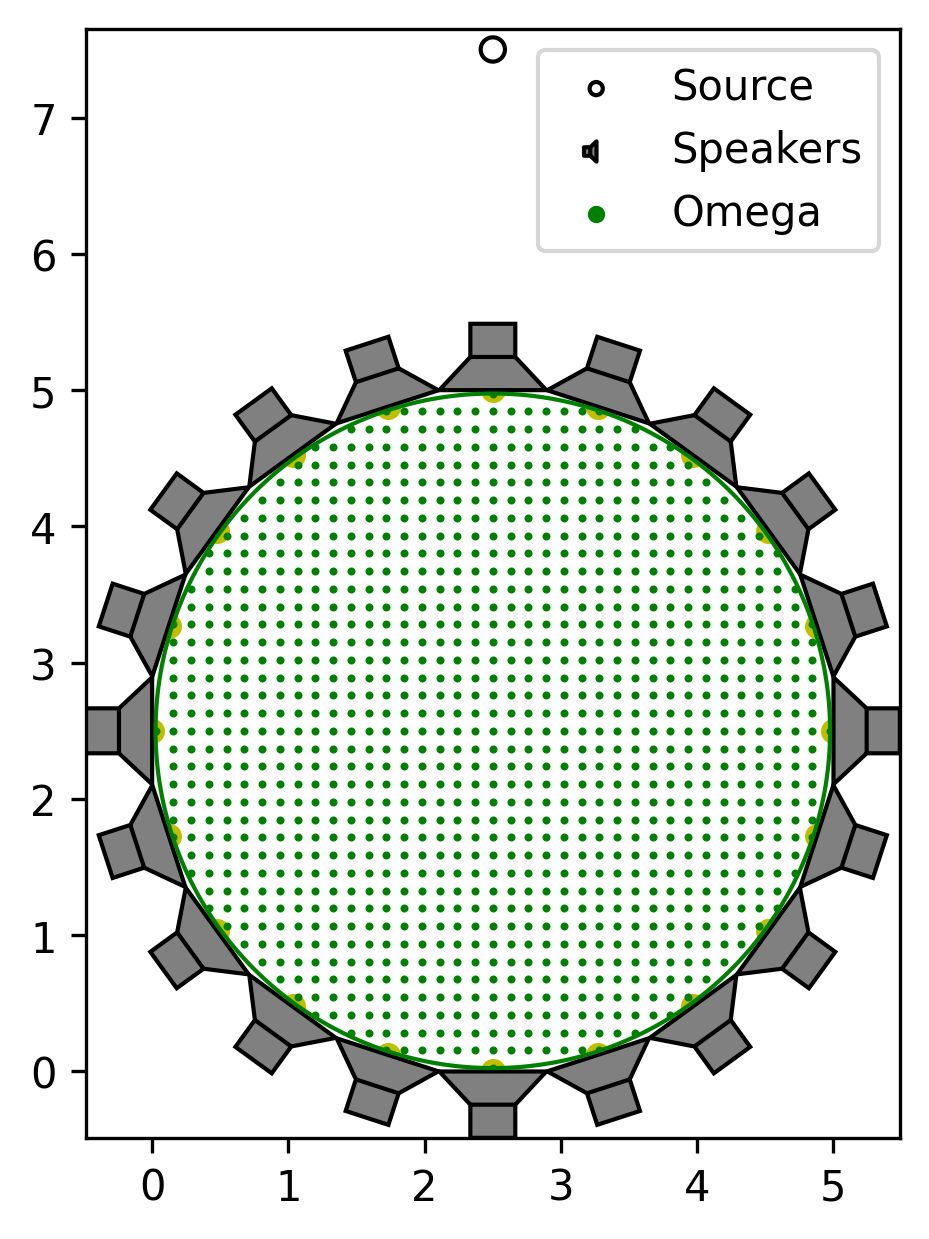}
        \caption{}
        \label{ex:NF-instance}
    \end{subfigure}%
    \begin{subfigure}[t]{0.245\textwidth}
        \centering
        \includegraphics[width=\textwidth]{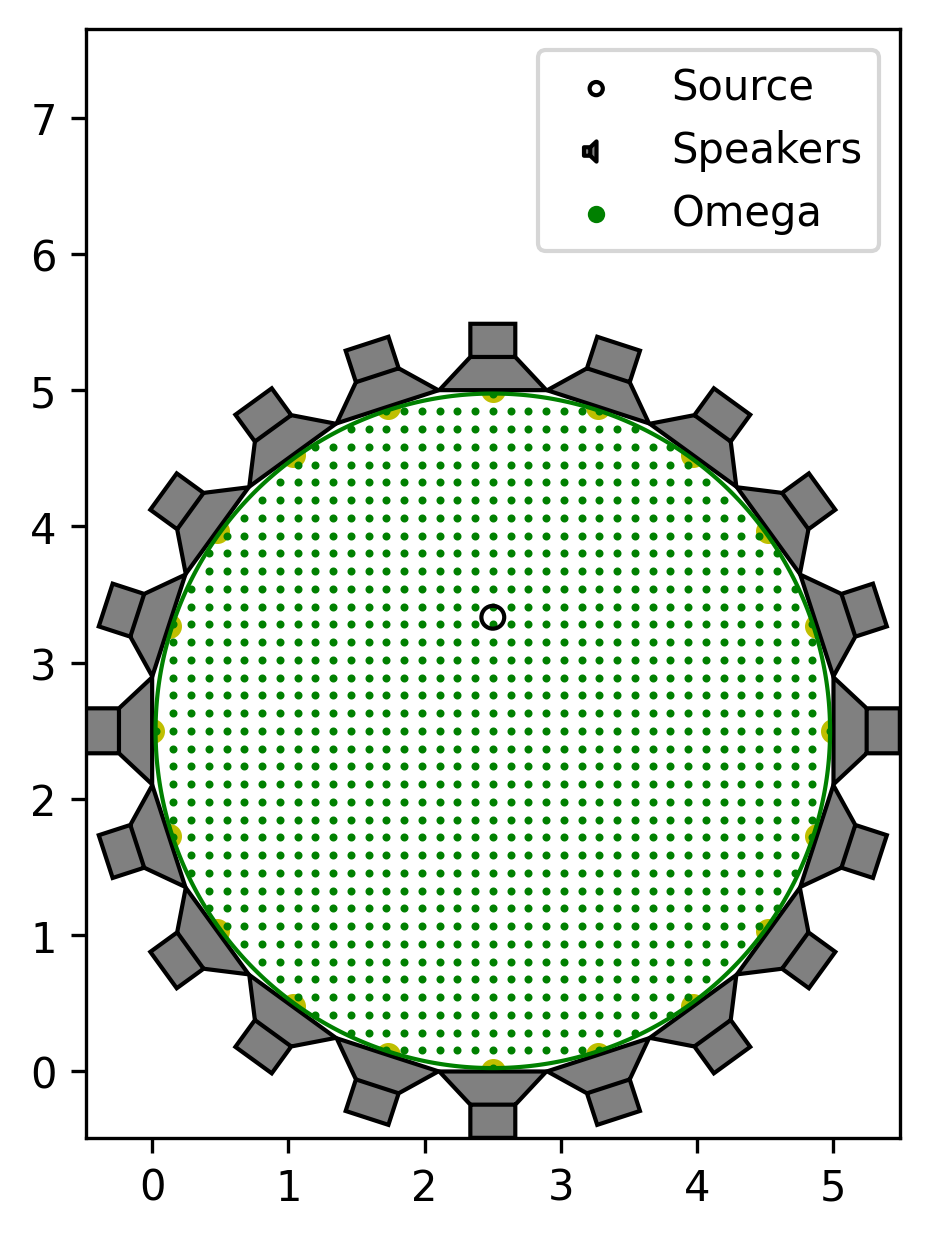}
        \caption{}
        \label{ex:FS-instance}
    \end{subfigure}\\
    \begin{subfigure}[t]{0.245\textwidth}
        \centering
        \includegraphics[width=\textwidth]{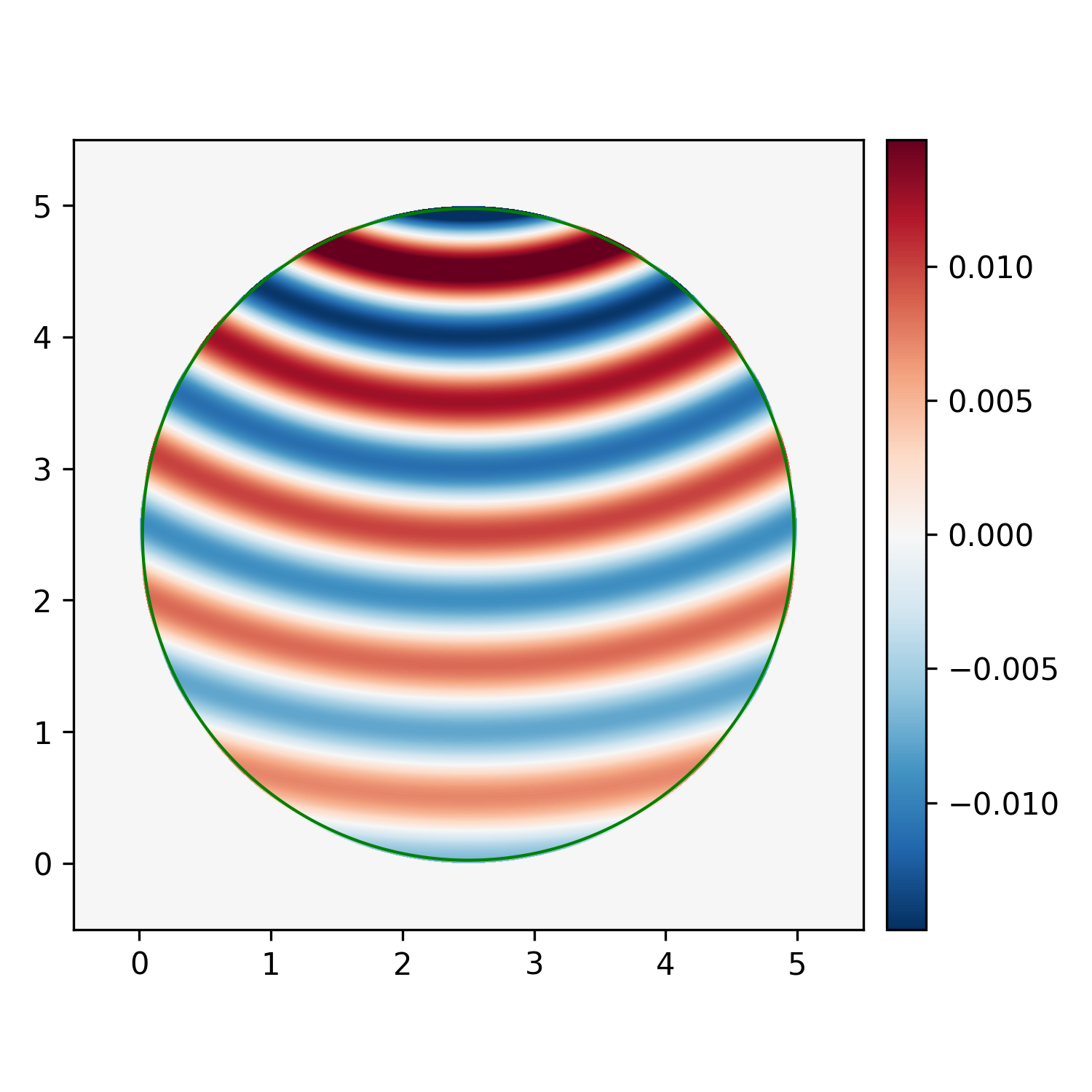}
        \caption{}
        \label{ex:NF-u0}
    \end{subfigure}%
    \begin{subfigure}[t]{0.245\textwidth}
        \centering
        \includegraphics[width=\textwidth]{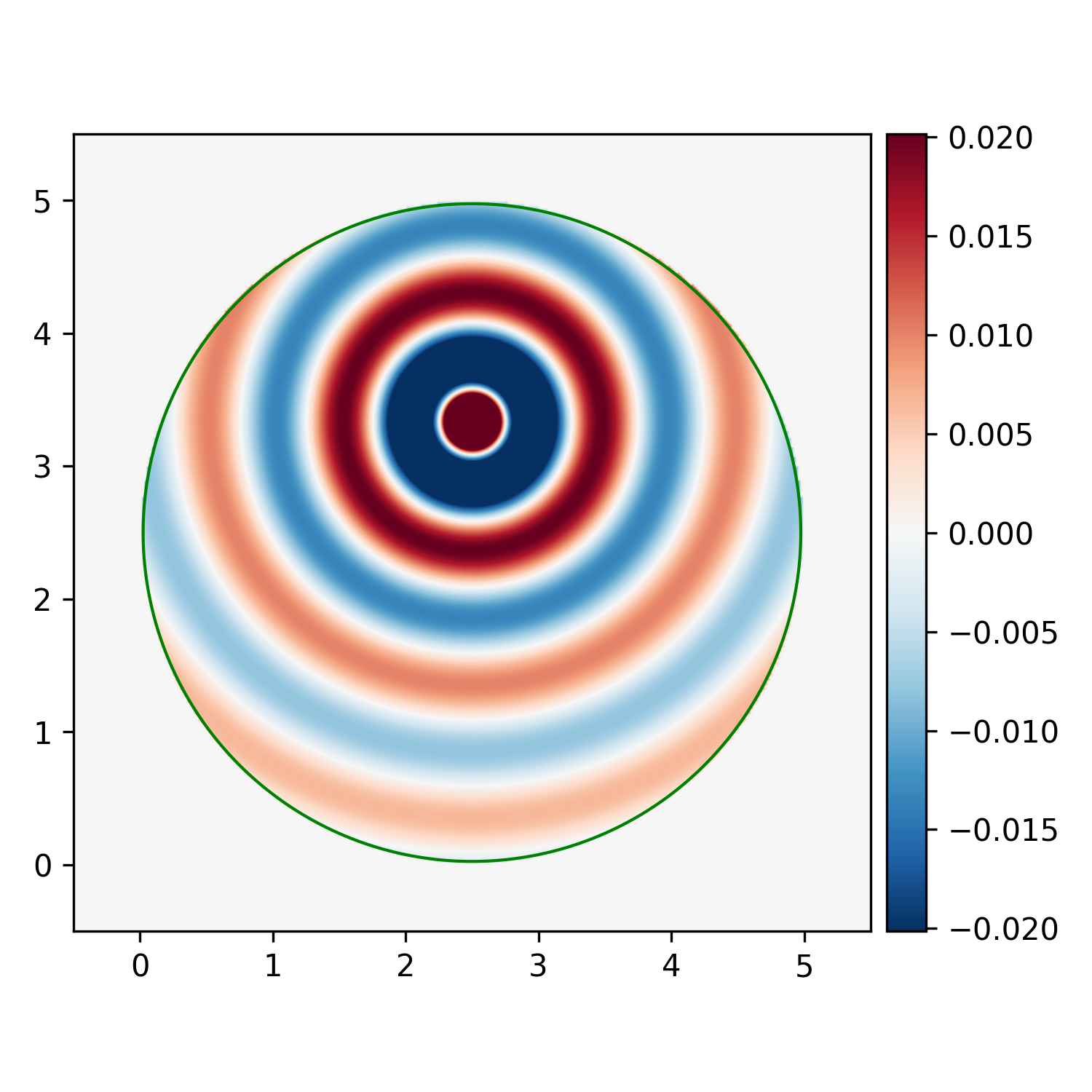}
        \caption{}
        \label{ex:FS-u0}
    \end{subfigure}%
    \caption{{\em Rows:} Experimental setup. Target \(\wh{u}_0(\fopt)\) (real part). {\em Columns:} Near-field instance. Focus-source instance.}
    \label{ex:instances}
\end{figure}

\begin{figure*}[!htbp]
    \centering
    \begin{subfigure}[t]{0.245\textwidth}
        \centering
        \includegraphics[width=\textwidth]{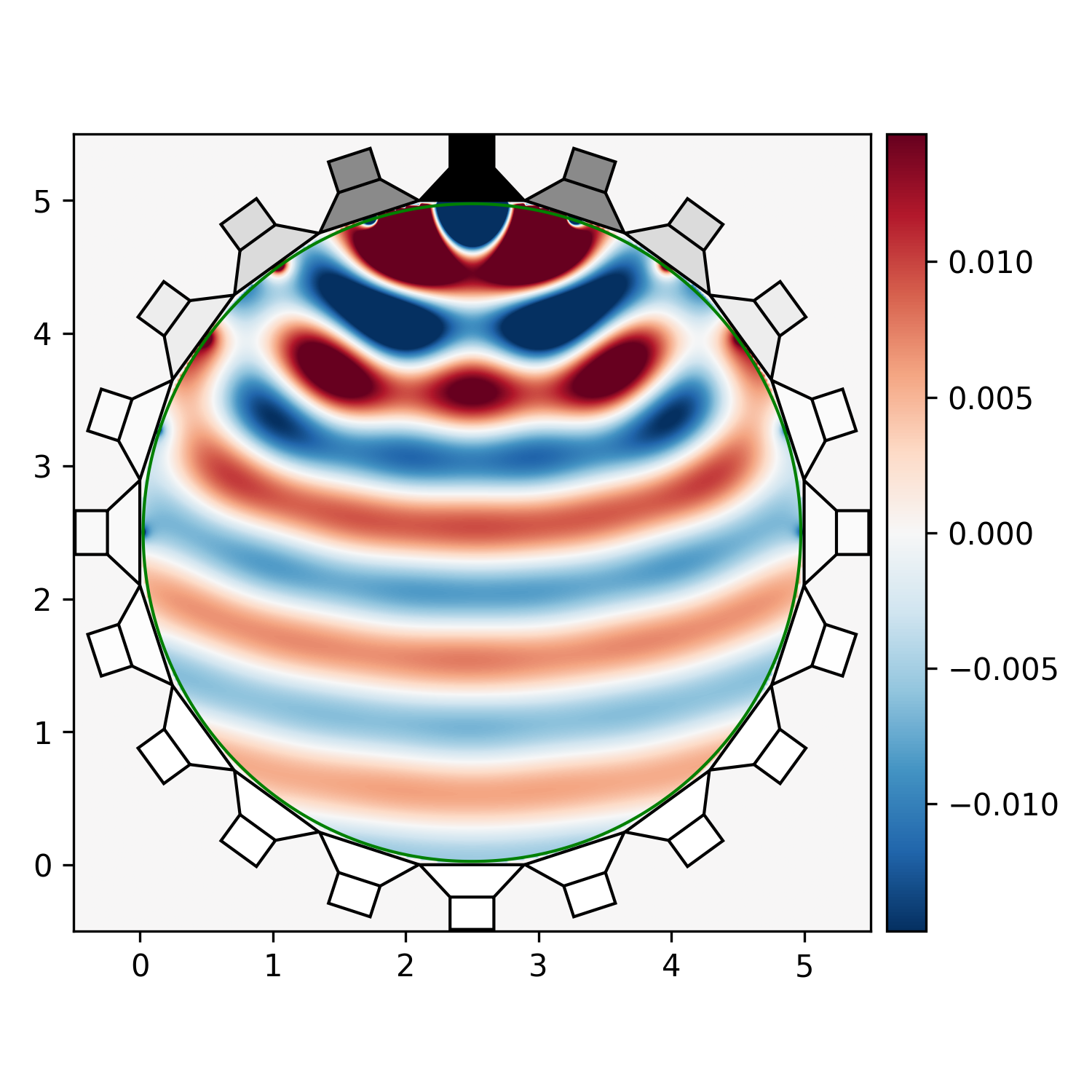}
        \caption{}
        \label{ex:SWEET-NF-u}
    \end{subfigure}%
    \begin{subfigure}[t]{0.245\textwidth}
        \centering
        \includegraphics[width=\textwidth]{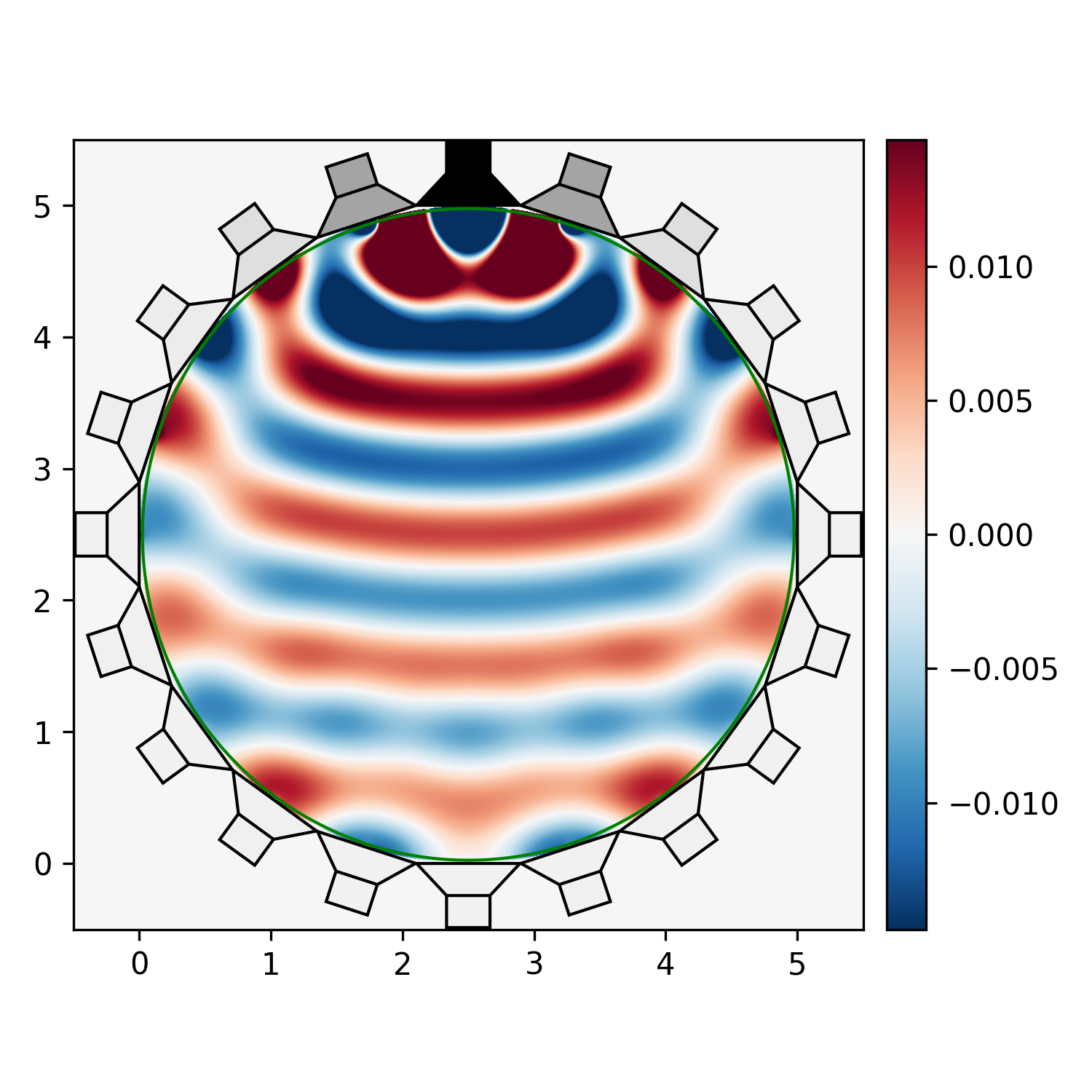}
        \caption{}
        \label{ex:NFCHOA-NF-L}
    \end{subfigure}%
    \begin{subfigure}[t]{0.245\textwidth}
        \centering
        \includegraphics[width=\textwidth]{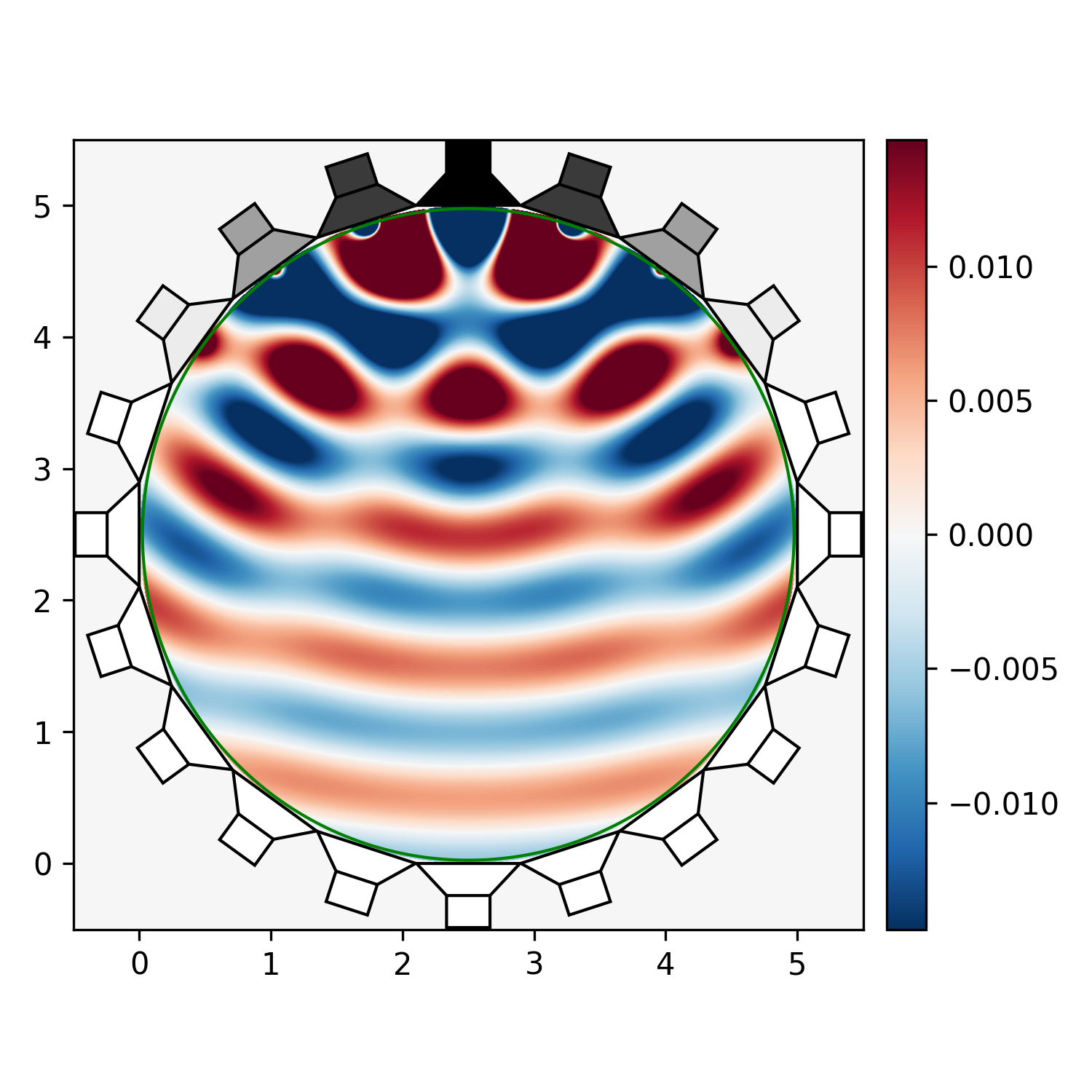}
        \caption{}
        \label{ex:WFS-NF-C}
    \end{subfigure}%
    \begin{subfigure}[t]{0.245\textwidth}
        \centering
        \includegraphics[width=\textwidth]{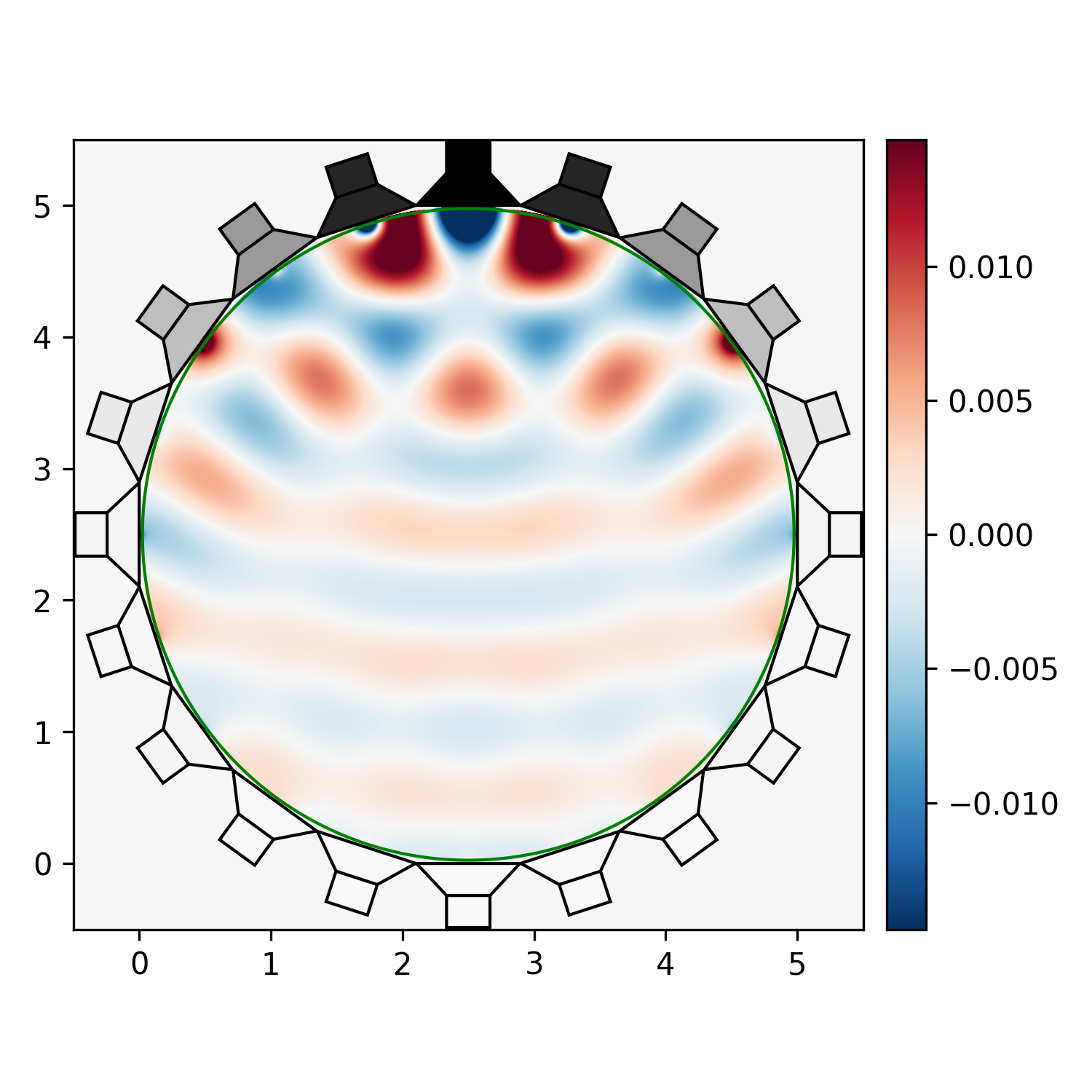}
        \caption{}
        \label{ex:L2-NF-u}
    \end{subfigure}\\
    \begin{subfigure}[t]{0.245\textwidth}
        \centering
        \includegraphics[width=\textwidth]{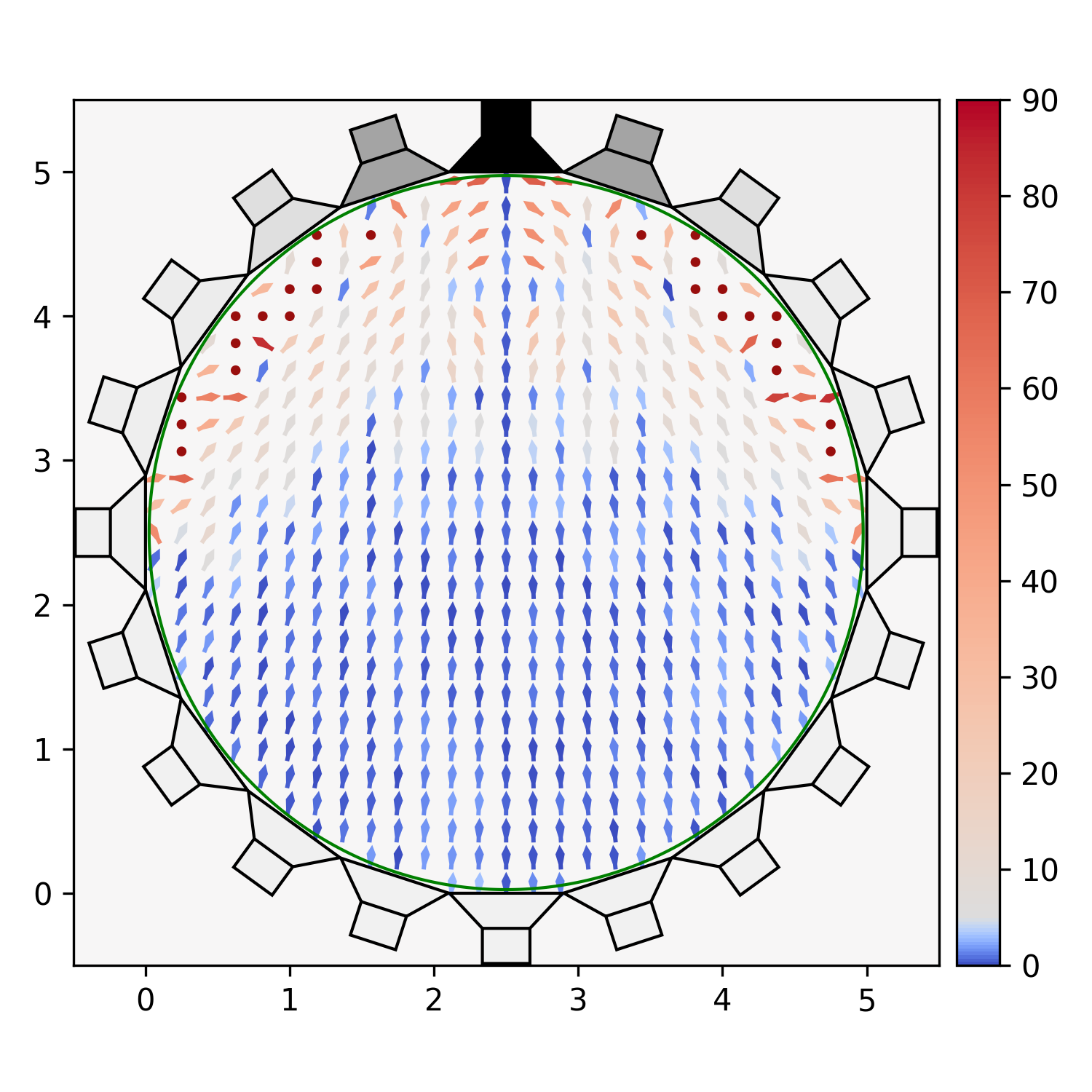}
        \caption{}
        \label{ex:SWEET-NF-L}
    \end{subfigure}
    \begin{subfigure}[t]{0.245\textwidth}
        \centering
        \includegraphics[width=\textwidth]{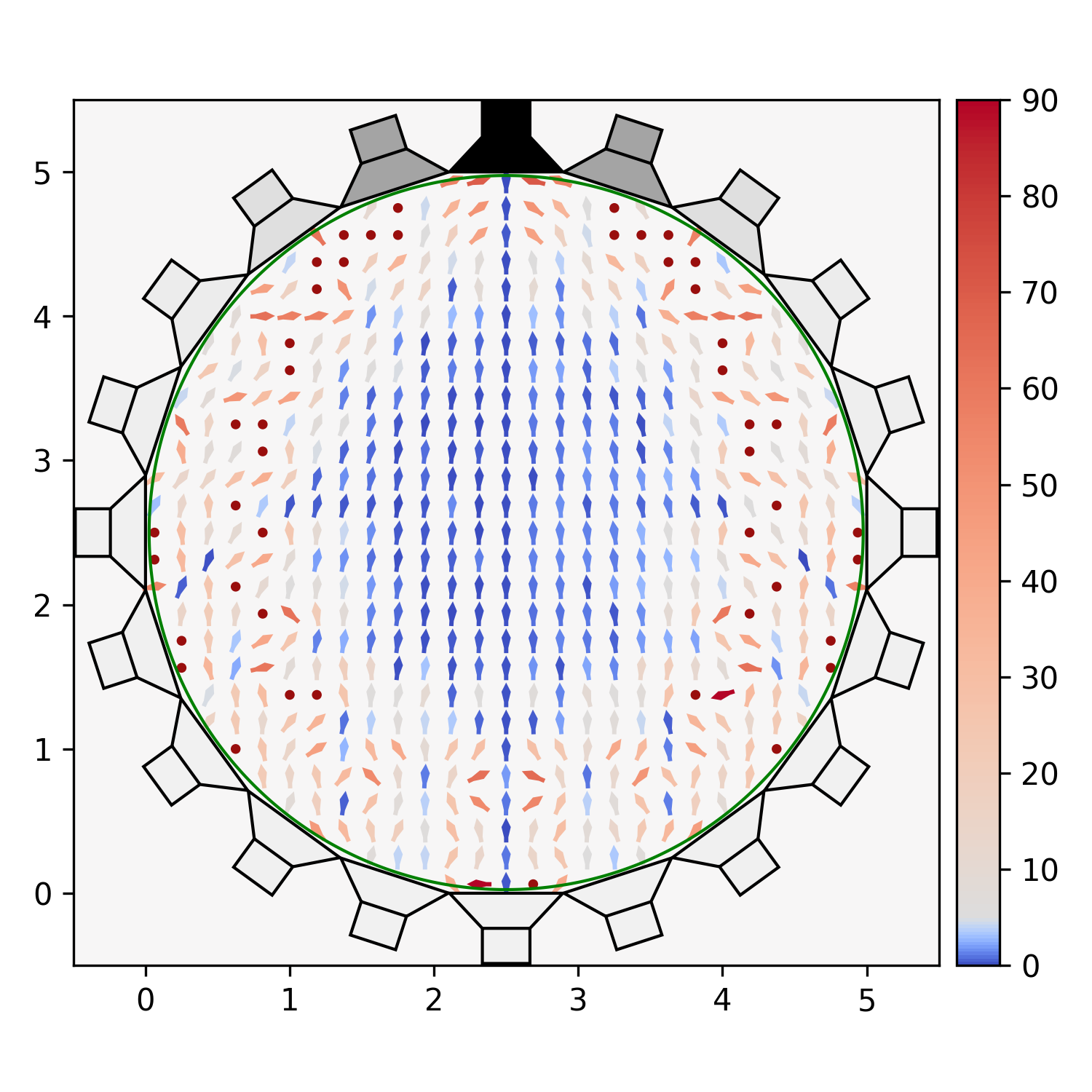}
        \caption{}
        \label{ex:NFCHOA-NF-L}
    \end{subfigure}
    \begin{subfigure}[t]{0.245\textwidth}
        \centering
        \includegraphics[width=\textwidth]{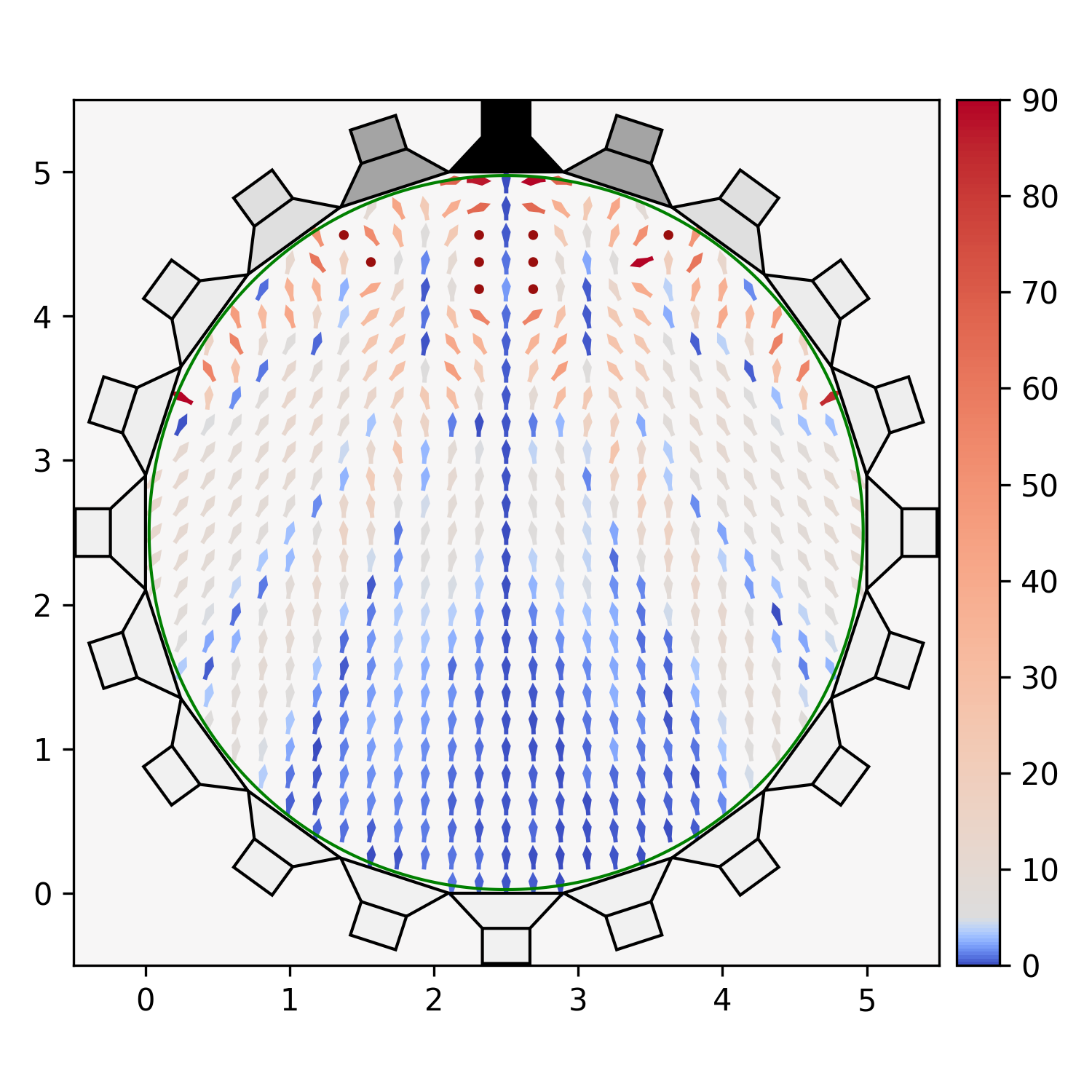}
        \caption{}
        \label{ex:WFS-NF-L}
    \end{subfigure}%
    \begin{subfigure}[t]{0.245\textwidth}
        \centering
        \includegraphics[width=\textwidth]{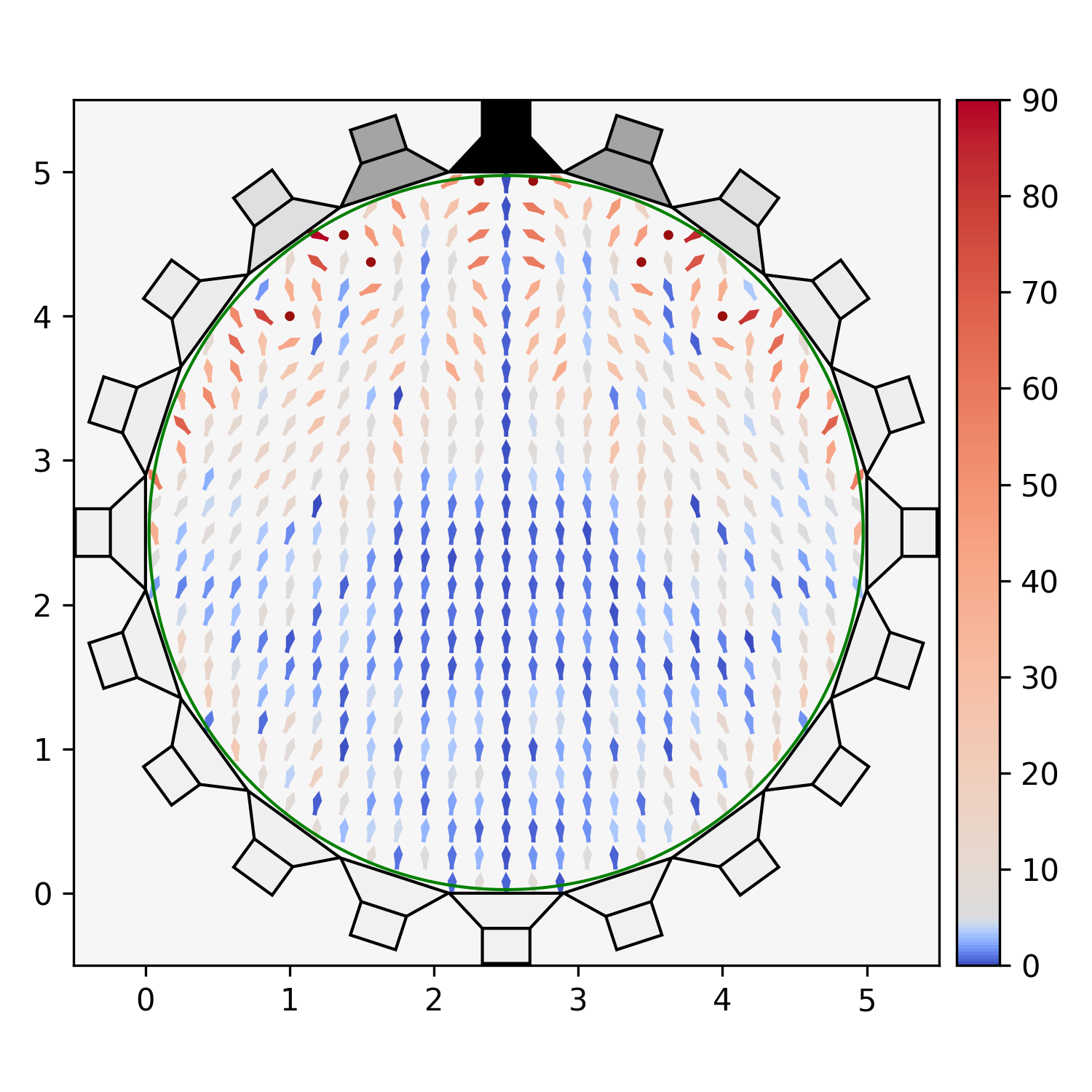}
        \caption{}
        \label{ex:L2-NF-L}
    \end{subfigure}\\
    \begin{subfigure}[t]{0.245\textwidth}
        \centering
        \includegraphics[width=\textwidth]{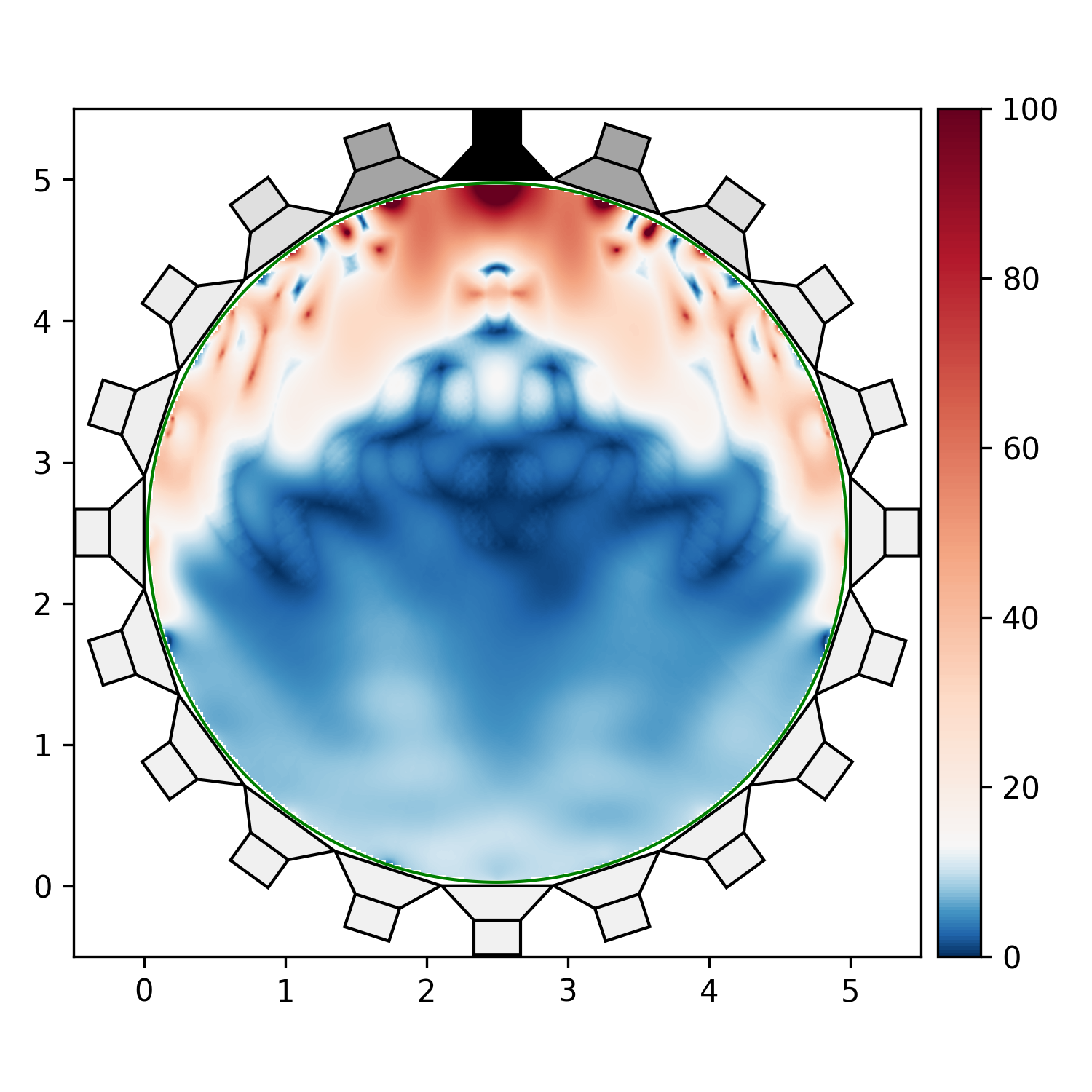}
        \caption{}
        \label{ex:SWEET-NF-C}
    \end{subfigure}%
    \begin{subfigure}[t]{0.245\textwidth}
        \centering
        \includegraphics[width=\textwidth]{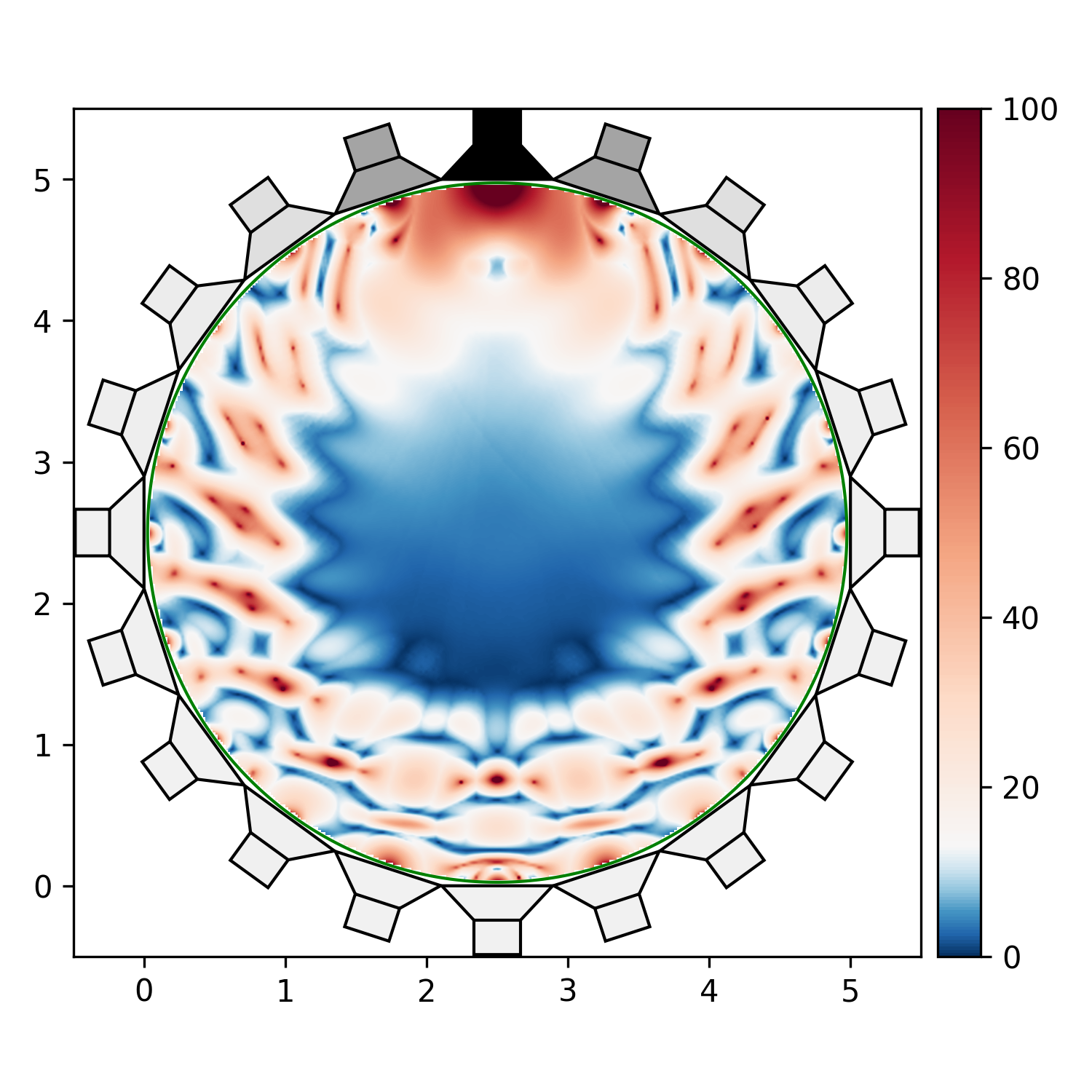}
        \caption{}
        \label{ex:NFCHOA-NF-C}
    \end{subfigure}
    \begin{subfigure}[t]{0.245\textwidth}
        \centering
        \includegraphics[width=\textwidth]{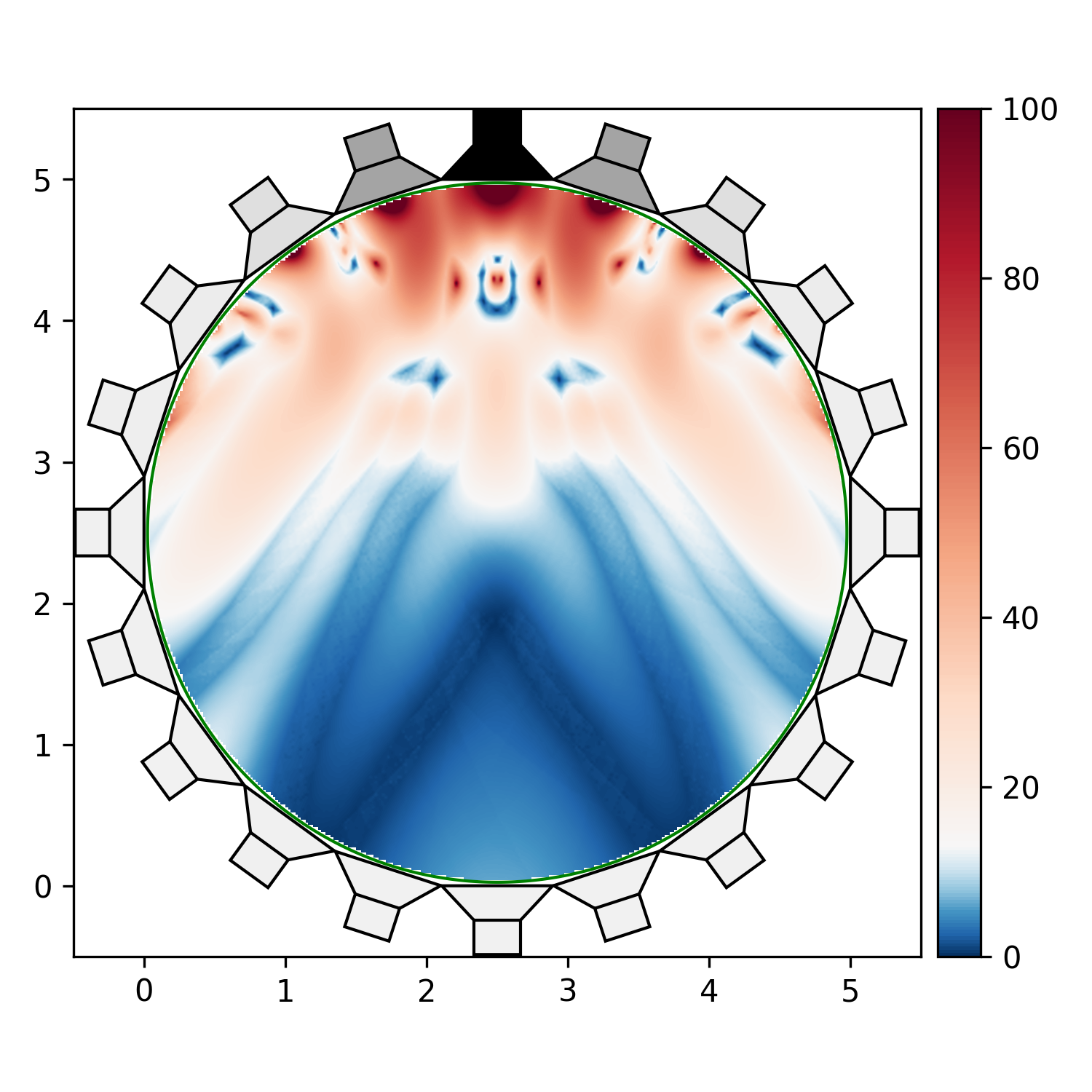}
        \caption{}
        \label{ex:WFS-NF-C}
    \end{subfigure}%
    \begin{subfigure}[t]{0.245\textwidth}
        \centering
        \includegraphics[width=\textwidth]{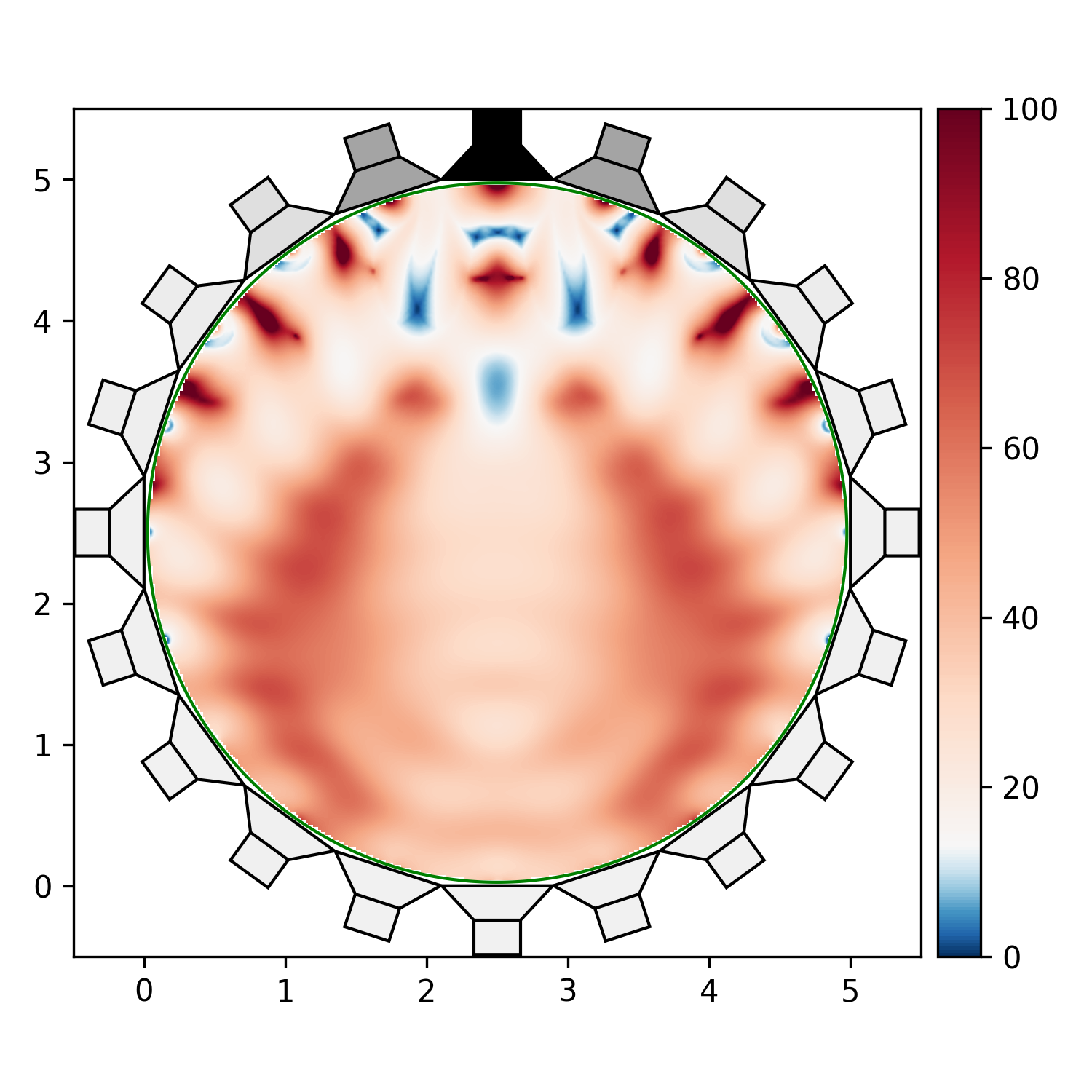}
        \caption{}
        \label{ex:L2-NF-C}
    \end{subfigure}%
    \caption{Near-field instance. {\em Rows:} Near-field \(\wh{u}(\fopt)\) (real part). Dietz's near-field azimuth localization, where the direction of the arrows represent the perceived localization whereas the color represents the deviation in degrees of the perceived localization from the desired one. McKenzie's near-field coloration (sones). {\em Columns:} SWEET-ReLU, NFC-HOA, WFS, and \(L^2\)-PMM.}
    \label{ex:NF-example}
\end{figure*}

\begin{figure*}[!htbp]
    \centering
    \begin{subfigure}[t]{0.245\textwidth}
        \centering
        \includegraphics[width=\textwidth]{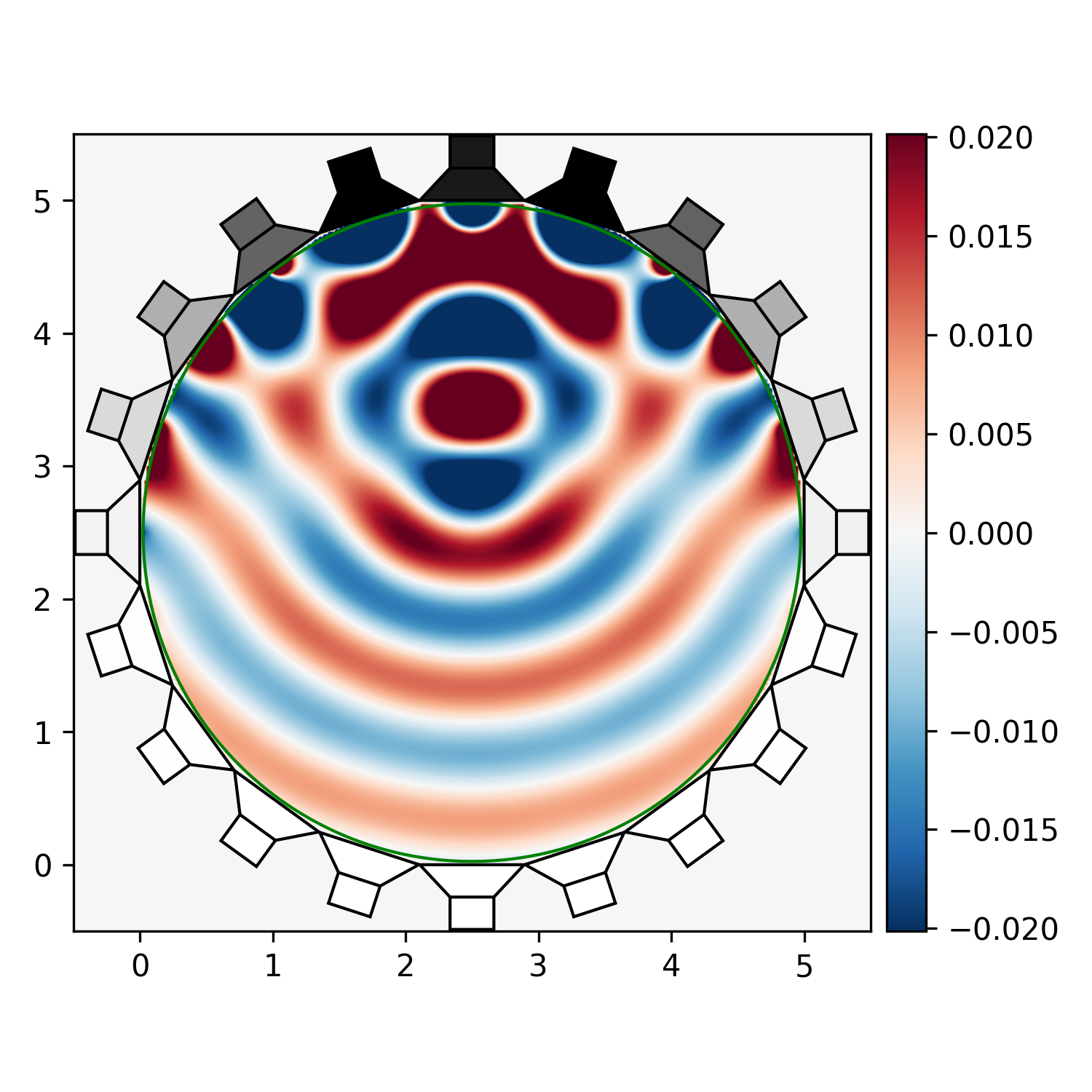}
        \caption{}
        \label{ex:SWEET-FS-u}
    \end{subfigure}%
    \begin{subfigure}[t]{0.245\textwidth}
        \centering
        \includegraphics[width=\textwidth]{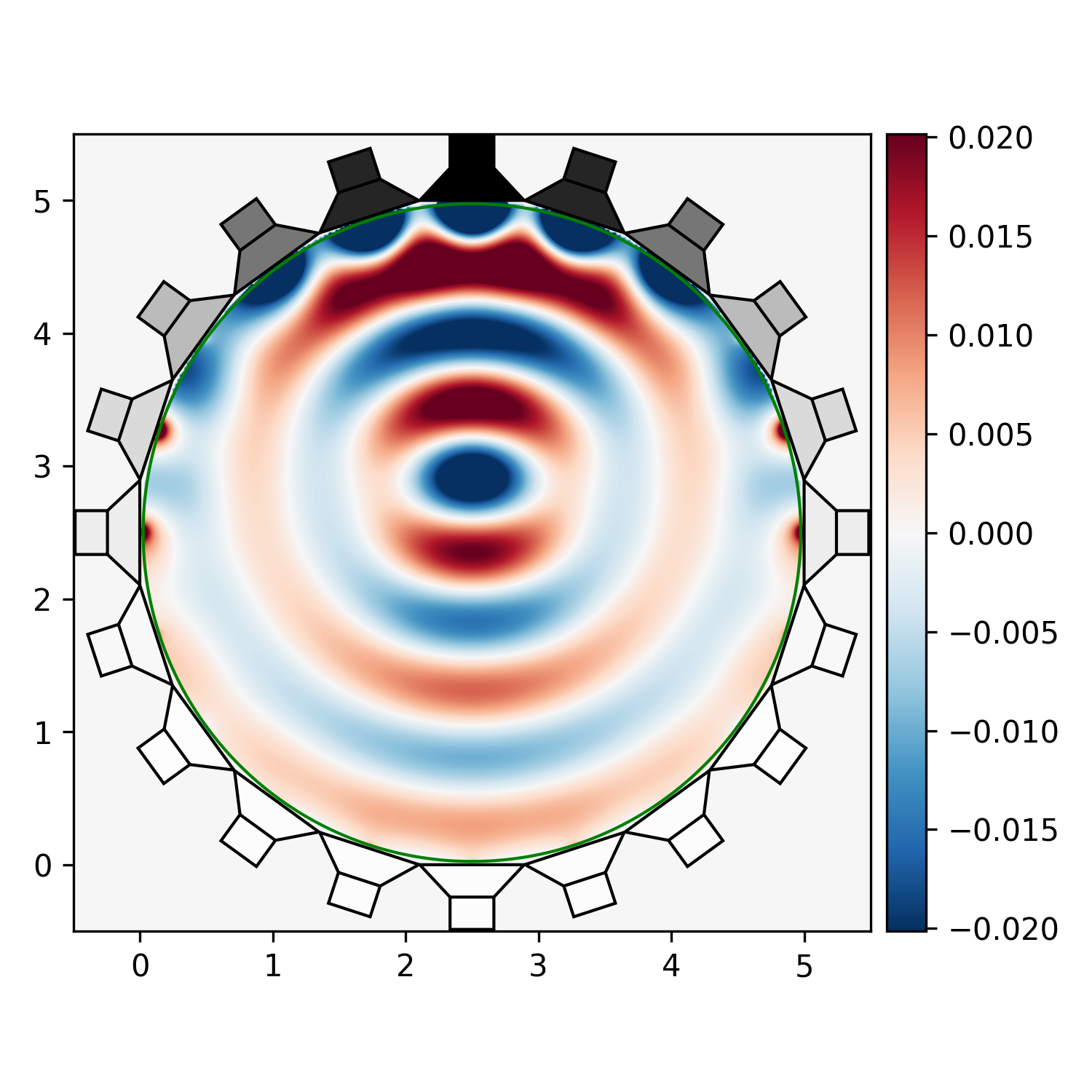}
        \caption{}
        \label{ex:NFCHOA-FS-u}
    \end{subfigure}%
    \begin{subfigure}[t]{0.245\textwidth}
        \centering
        \includegraphics[width=\textwidth]{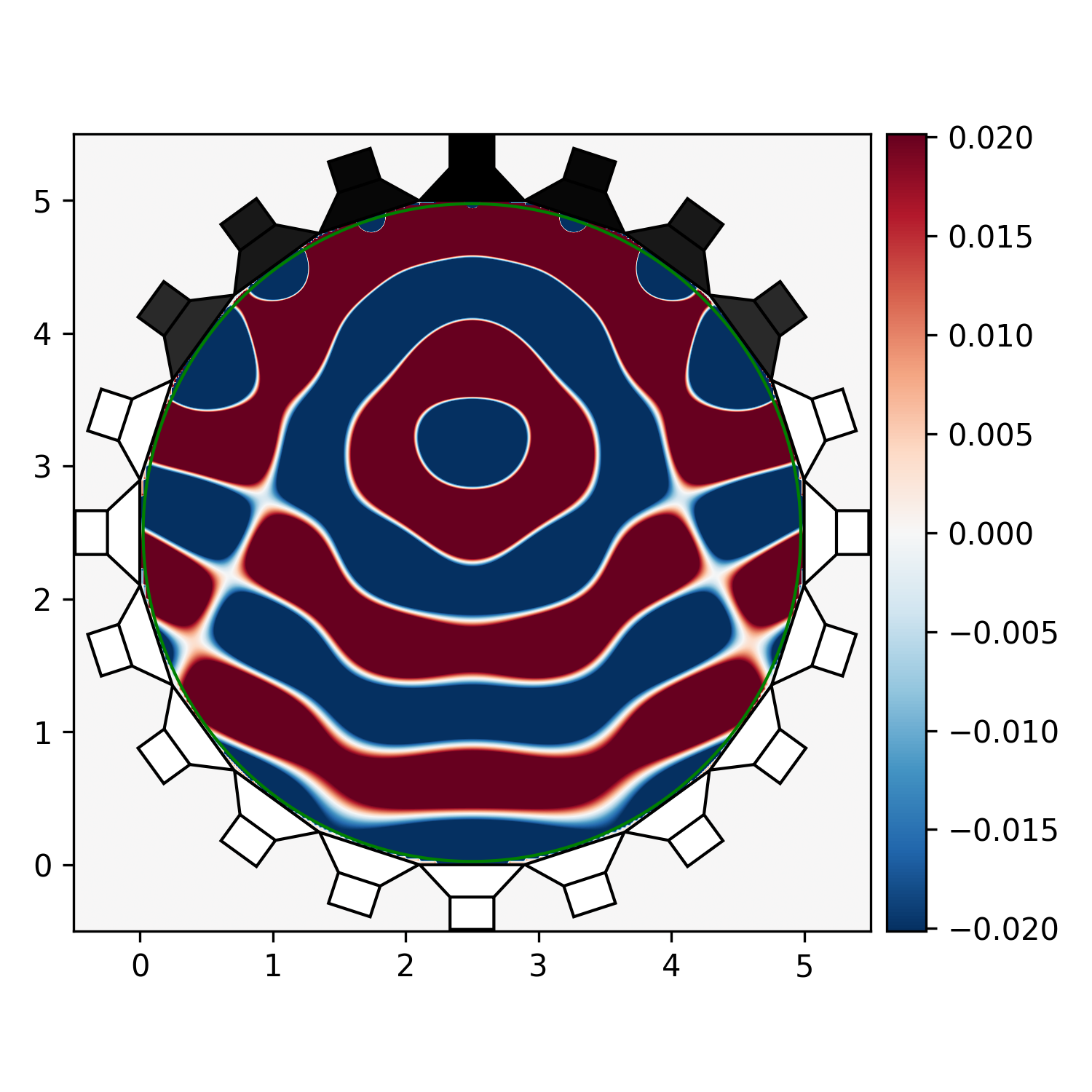}
        \caption{}
        \label{ex:WFS-FS-u}
    \end{subfigure}%
    \begin{subfigure}[t]{0.245\textwidth}
        \centering
        \includegraphics[width=\textwidth]{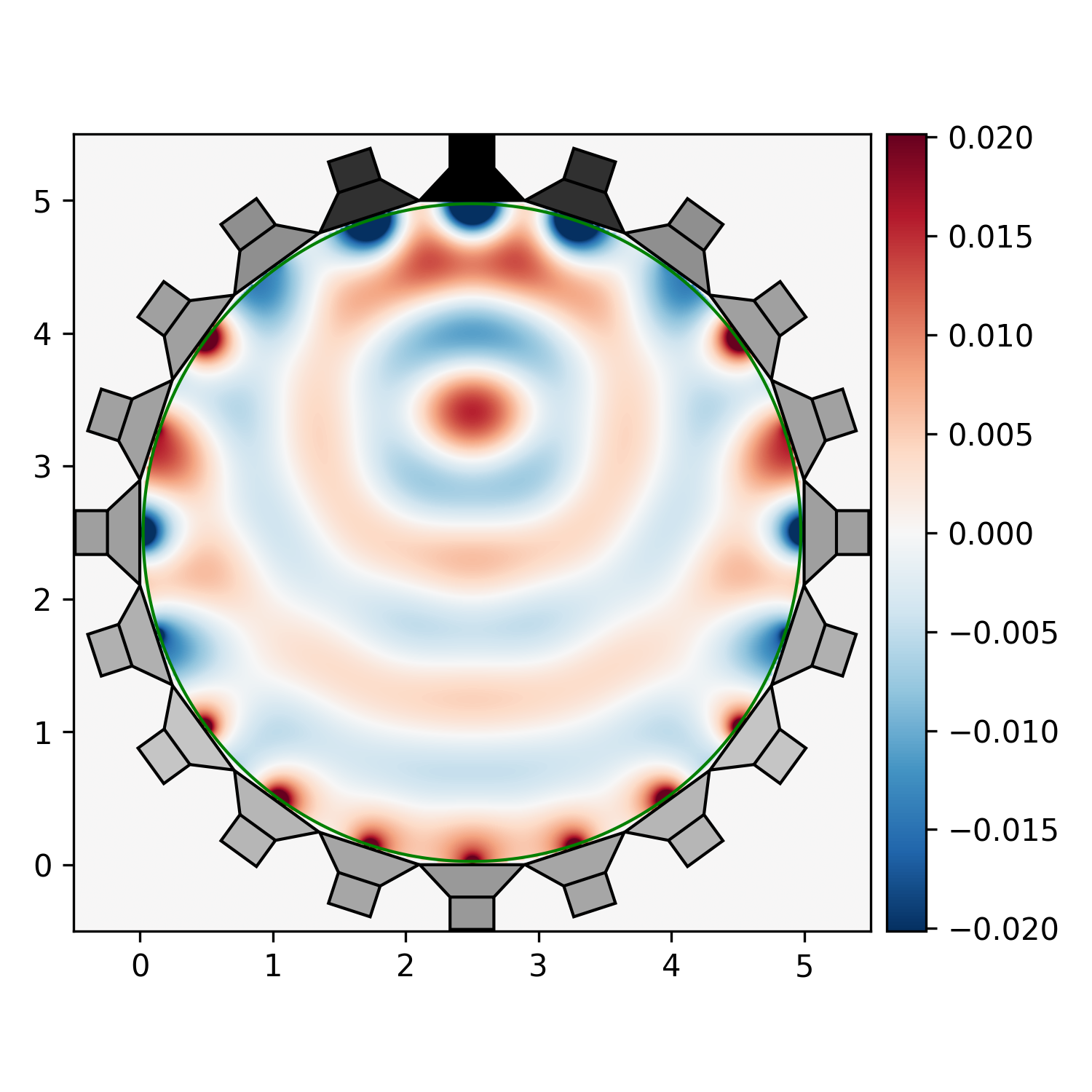}
        \caption{}
        \label{ex:L2-FS-u}
    \end{subfigure}\\
    \begin{subfigure}[t]{0.245\textwidth}
        \centering
        \includegraphics[width=\textwidth]{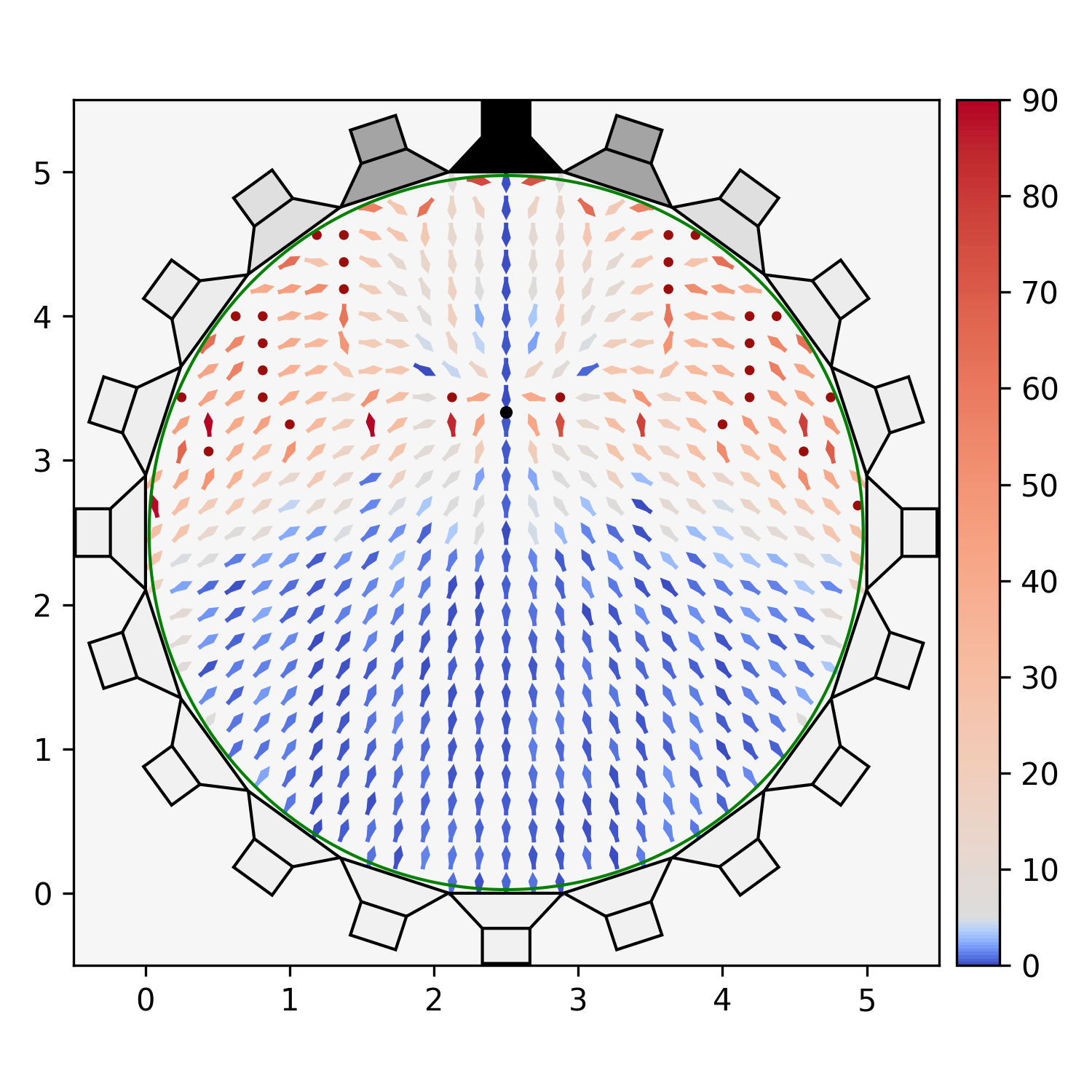}
        \caption{}
        \label{ex:SWEET-FS-L}
    \end{subfigure}
    \begin{subfigure}[t]{0.245\textwidth}
        \centering
        \includegraphics[width=\textwidth]{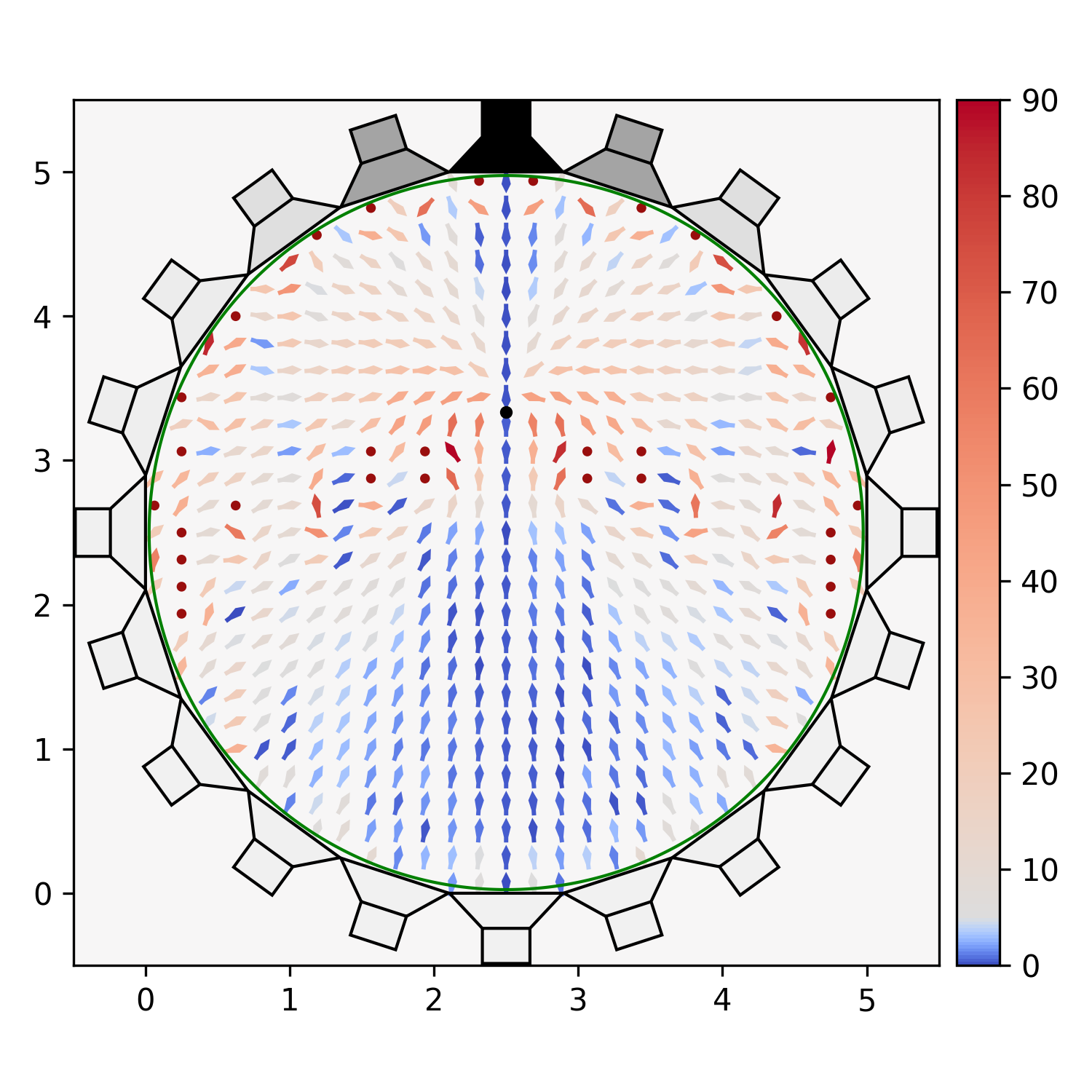}
        \caption{}
        \label{ex:NFCHOA-FS-L}
    \end{subfigure}
    \begin{subfigure}[t]{0.245\textwidth}
        \centering
        \includegraphics[width=\textwidth]{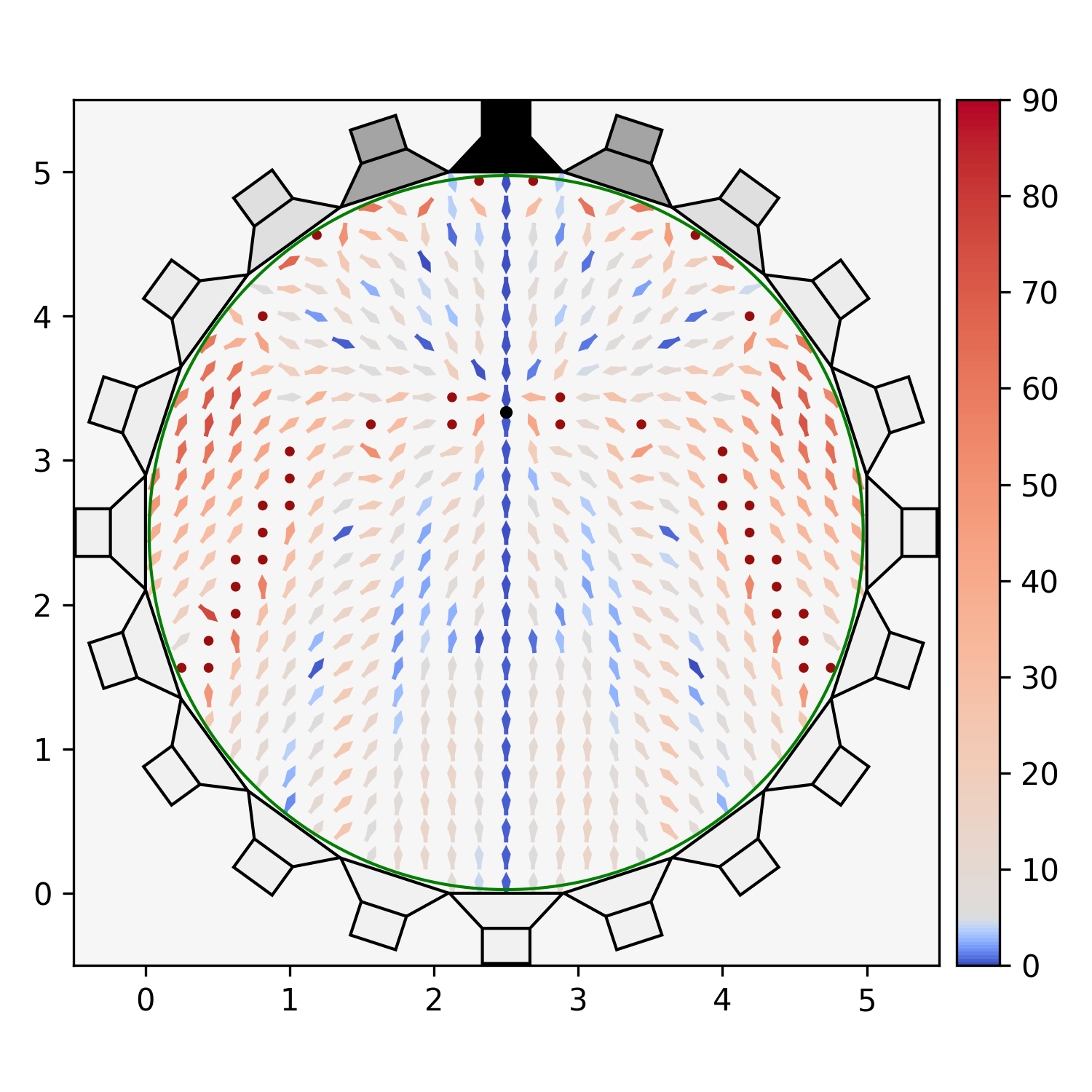}
        \caption{}
        \label{ex:WFS-FS-L}
    \end{subfigure}%
    \begin{subfigure}[t]{0.245\textwidth}
        \centering
        \includegraphics[width=\textwidth]{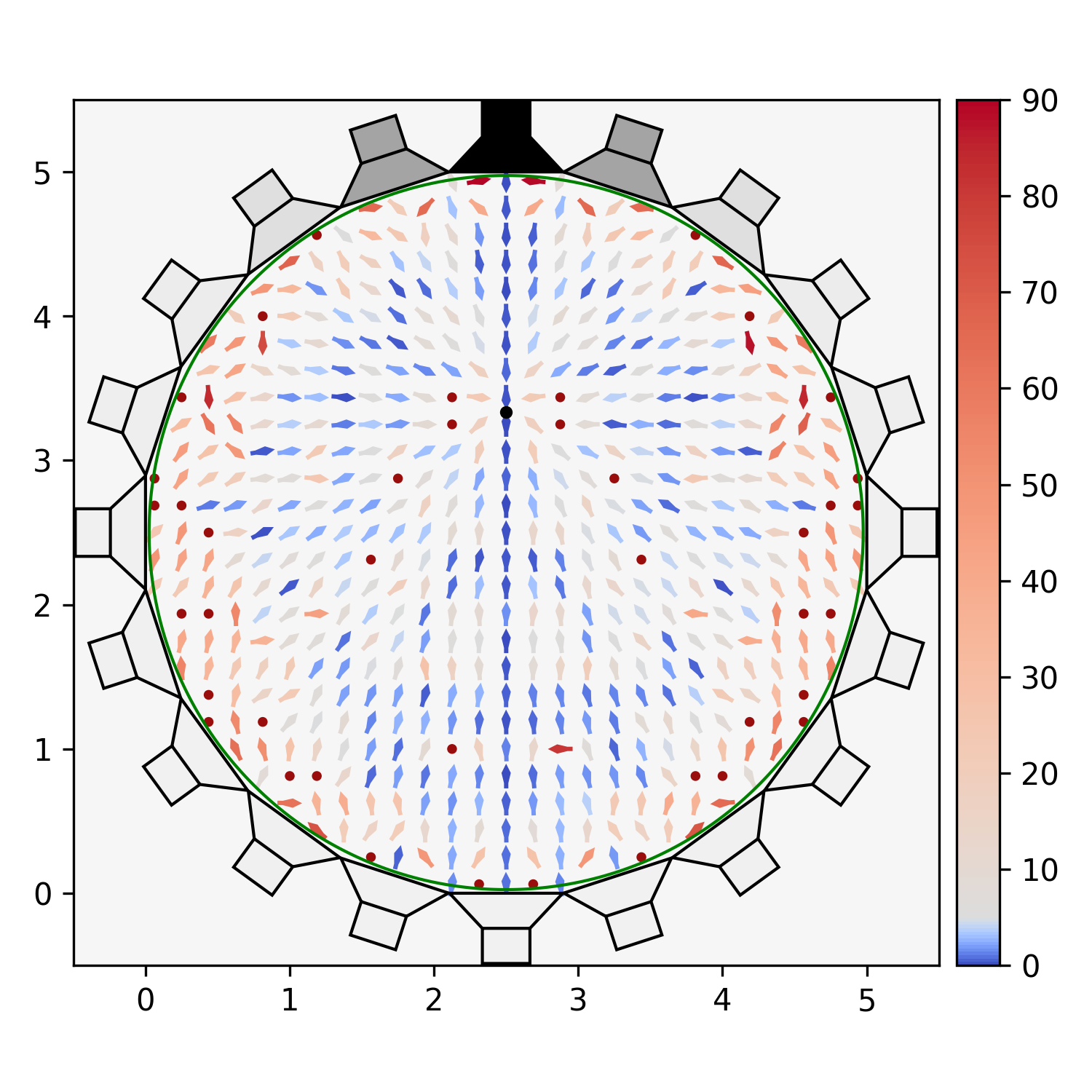}
        \caption{}
        \label{ex:L2-FS-L}
    \end{subfigure}\\
    \begin{subfigure}[t]{0.245\textwidth}
        \centering
        \includegraphics[width=\textwidth]{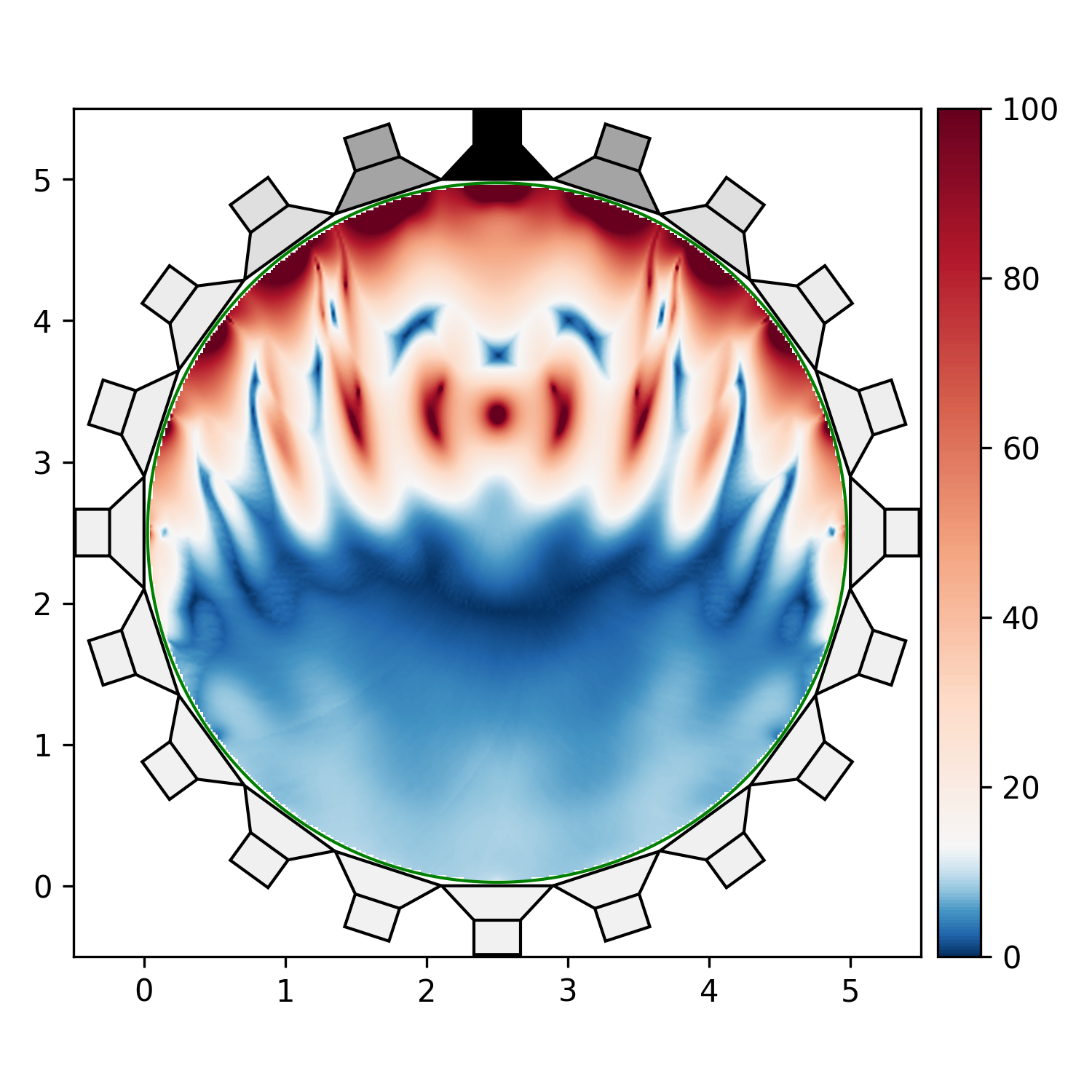}
        \caption{}
        \label{ex:SWEET-FS-C}
    \end{subfigure}%
    \begin{subfigure}[t]{0.245\textwidth}
        \centering
        \includegraphics[width=\textwidth]{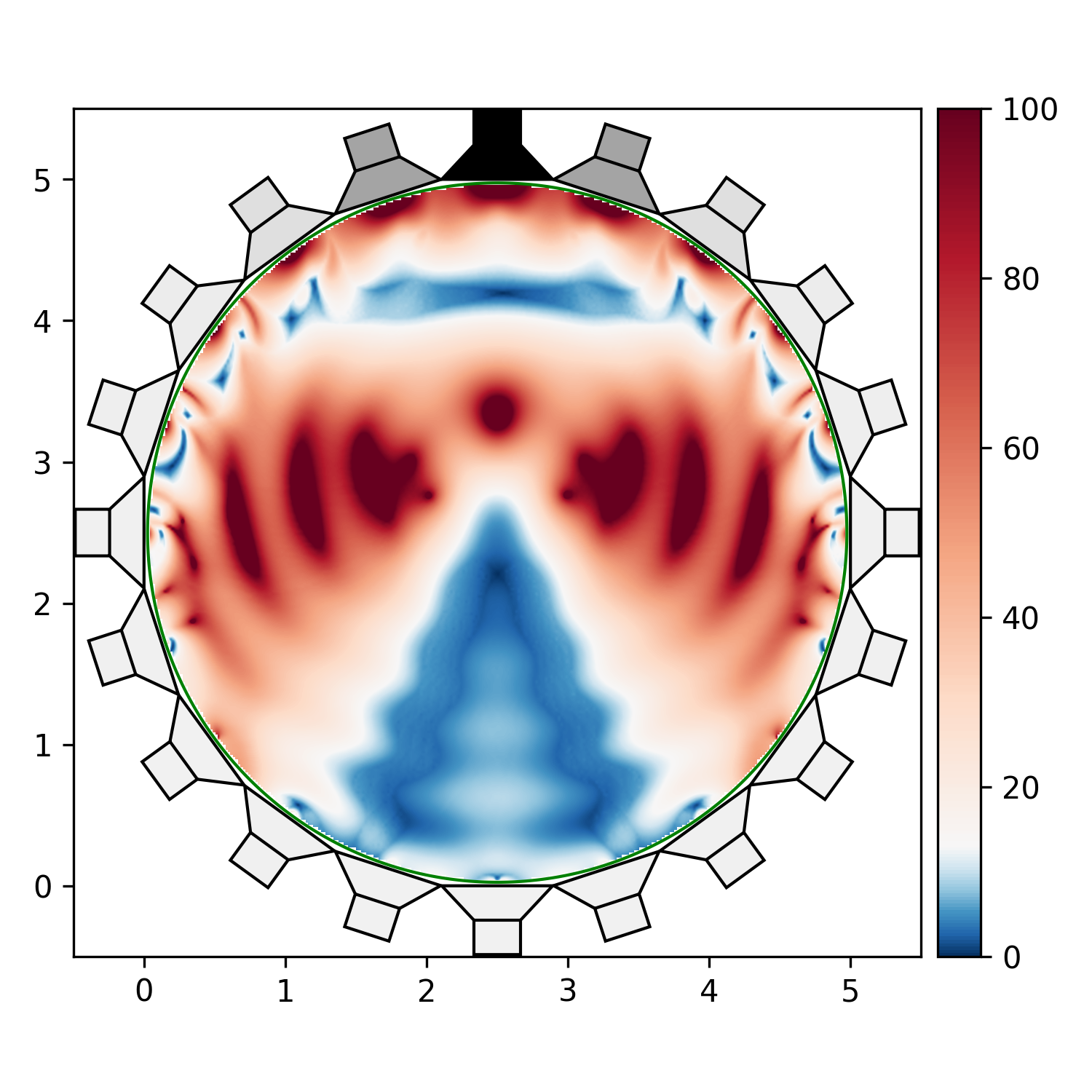}
        \caption{}
        \label{ex:NFCHOA-FS-C}
    \end{subfigure}
    \begin{subfigure}[t]{0.245\textwidth}
        \centering
        \includegraphics[width=\textwidth]{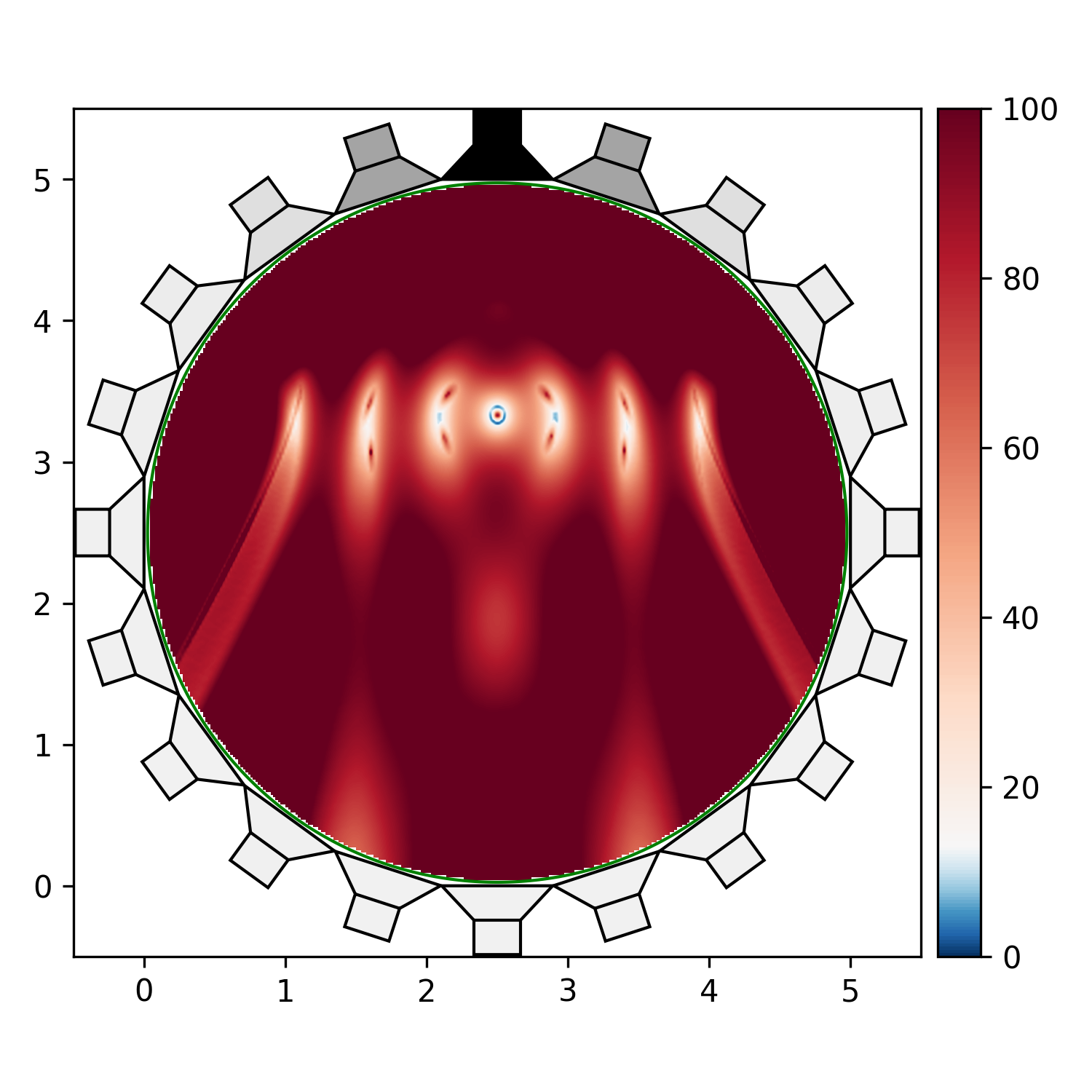}
        \caption{}
        \label{ex:WFS-FS-C}
    \end{subfigure}%
    \begin{subfigure}[t]{0.245\textwidth}
        \centering
        \includegraphics[width=\textwidth]{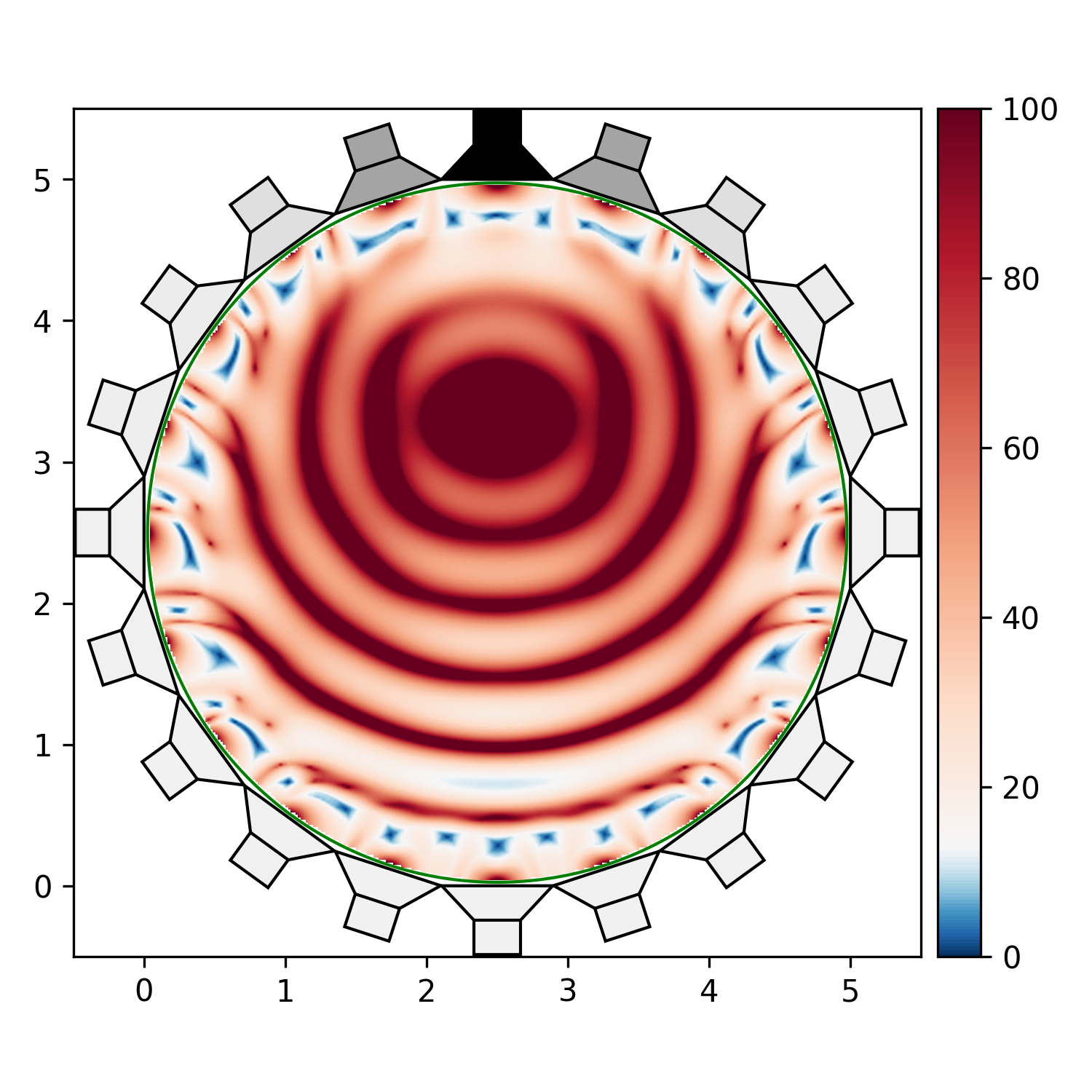}
        \caption{}
        \label{ex:L2-FS-C}
    \end{subfigure}%
    \caption{Focus-source instance. {\em Rows:} Near-field \(\wh{u}(\fopt)\) (real part). Dietz's near-field azimuth localization, where the direction of the arrows represent the perceived localization whereas the color represents the deviation in degrees of the perceived localization from the desired one. McKenzie's near-field coloration (sones). {\em Columns:} SWEET-ReLU, NFC-HOA, WFS, and \(L^2\)-PMM.}
    \label{ex:FS-example}
\end{figure*}

For the experiments we compare the performance of our method with the state-of-the-art methods WFS, NFC-HOA and \(L^2\)-PMM in terms of its azimuth localization and coloration performance, as they are the main features of the auditory scene for spatial sound. We use Dietz's model to measure binaural azimuth localization~\cite{dietz2011auditory} and McKenzie's model for binaural coloration~\cite{mckenzie2022predicting}. The setup for the numerical experiments consists of an equispaced arrangement of 20 loudspeakers lying on a circle of radius 2.5 m and at \(\pi/4\approx 0.785\) m from each other. The region of interest \(\Omega\) is a concentric circle of radius 2.4975 m (Fig.~\ref{ex:instances}). The speed of sound is \(c_s = 343\) m/s. Two instances of this setup were evaluated: the {\em near-field instance}, where the source outside the arrangement at 5 m of its center with \(\wh{c}_0(\fopt)=68\) dB, and the {\em focus-source instance}, where the source is inside the arrangement at 0.82 m of its center (Fig.~\ref{ex:FS-instance}) with \(\wh{c}_0(\fopt)=60\) dB. In both cases, \(\fopt=343\) Hz. To construct the perceptual maps \(D\) and \(L\) we assume that the listeners are looking at the virtual sound source, which implies that \(\Theta = \Theta(x)=\{\text{ang}(x_0 - x)\}\).

The SWEET-ReLU algorithm and the \(L^2\)-PMM method were implemented in Python 3.8 using the CVXPY package, version 1.1.15~\cite{diamond2016cvxpy, agrawal2018rewriting} 
and MOSEK, version 9.3.6~\cite{mosek}. The simulations of 2.5D NFC-HOA and 2.5D WFS were done with the Sound Field Synthesis Toolbox (SFST), version 3.2~\cite{wierstorf2012sound}, unless the focus-source 2.5D NFC-HOA simulations, which were done following the {\em angular weighting approach}~\cite{ahrens2009spatial}. The HRTFs used to simulate \(\bar{u}\) and \(\bar{u}_0\) were constructed as the circulant Fourier transform of the elements of the TU-Berlin HRIR free data base~\cite{wierstorf2011free}. Dietz's and McKenzie's models were implemented using Matlab 2022a with the Auditory Modelling Toolbox (AMT), version 1.1~\cite{majdak2021amt}.

For the implementation of the HRTFs, the 3 meters radial distance HRIRs of the data set were radially extrapolated using delay and attenuation, according to the map 
\[
    d\mapsto \frac{3}{d}\text{HRIR}\left(t-\frac{d-3}{c_s}\right),
\]
where \(d\) is the desired radial distance. It should be noticed that for short distances, e.g. less than 1 m, ILDs vary significantly with distance~\cite{wierstorf2011free}. Hence, the experiments might be enhanced by using a complete HRTF data set. For the implementation of Dietz's model, since we treated (pseudo) sinusoidal signals, the interaural phase differences (IPDs) and ILDs of the reproduced signals were extracted manually following the convention of~\cite{dietz2011auditory} as
\begin{align*}
    \text{ILD}(\bar{u}(x,\theta)) &= \frac{20}{c_s}\log_{10}\left(\frac{|u^r(x,\theta)|}{|u^\ell(x,\theta)|}\right),\\
    \qquad \text{IPD}(\bar{u}(x,\theta))&=\frac{\text{arg}(u^\ell(x,\theta)u^{r}(x,\theta)^*)}{-2\pi i\fopt}.
\end{align*}
Then, the {\em unwrapped} ITDs were obtained from the ILDs and IPDs using the \texttt{dietz2011\_unwrapitd.mat} function. The estimated azimuth localization is obtained by plugging the ITDs into the \texttt{itd2angle.mat} function, which uses the \texttt{itd2angle\_lookuptable.mat} table. The latter table is constructed with the same HRTFs that we consider for the simulation of \(\bar{u}\) and \(\bar{u}_0\). For the implementation of McKenzie's model, the binaural signals were transformed to time-domain using a sampling frequency of 44100 Hz and a number of samples of 256. As proposed in~\cite[Chapter 5.6.2]{ahrens2012analytic}, the implementation of the angular weighting for NFC-HOA in focus-source instances was defined, for the \(n\)-th mode, as 
\[
    w_n(\fopt) = \begin{cases}
        \frac{1}{2}\left(\cos\left(\frac{n}{\lfloor\frac{\omega}{c} d\rfloor}\pi\right)+1\right) & \text{for }n\leq \frac{\omega}{c}d\\
        0 & \text{elsewhere},
    \end{cases}
\]
where \(d\) is the distance from the source to the center of the room, and \(\omega=2\pi\fopt\). For the implementation of the SWEET method, we have chosen \(\eps_i\) adaptively with percentile \(p=99\). For SWEET and \(L^2\)-PMM a uniform discretization of 2348 points was used for \(\Omega\) at a distance of at most 0.09 m, achieving more than 30 points per wavelength.

To compare the performance of the methods, we measure the size of the {\em localization sweet spot} (LSS) and {\em coloration sweet spot} (CSS). We former is the region where the perceived azimuth localization measured by Dietz's model deviates no more than 5 degrees from the desired one, whereas the latter is the region where the coloration measured by McKenzie's model is lower or equal than 13 sones. As discussed in~\cite{mckenzie2022predicting}, a coloration lower than 13 sones is strongly correlated with empirical MUSHRA tests with more than 80 out of 100 points. It should be noted that, since we analyse (pseudo) sinusoidal signals of frequency 343 Hz, in these experiments the CSS and LSS are constituted by the points where the interaural amplitude or phase (respectively) of the binaural signal are correctly reconstructed.

The CSS and LSS generated by each method for the near-field and focus-source instances are shown in Fig.~\ref{ex:NF-example} and~\ref{ex:FS-example} respectively, and their size is shown in Table~\ref{ex:sweet-table}. The blue zones of Figs~\ref{ex:SWEET-NF-L}-l, Figs~\ref{ex:SWEET-FS-L}-l represent the sweet spots of each case. The estimated localization in Figs~\ref{ex:SWEET-NF-L}-h, Figs~\ref{ex:SWEET-FS-L}-h is shown for deviations of 0 to 90 degrees from the desired one. Greater deviations are represented by a dot with no direction.
\vspace{-8pt}
\begin{table}[!htbp]
    \centering
    \begin{tabular}{||c c c c c||}
         \hline
         & SWEET & NFC-HOA & WFS & \(L^2\)-PMM\\
         \hline
         \hline
         NF CSS & 70.8\% & 51  \%  & 55.6\%   & 3.3\%\\
         NF LSS & 68.4\% & 42.9\%  & 47.9\% & 52.3\%\\
         FS CSS & 58.7\% & 27.8\%  & 0.1\%    & 5.2\%\\
         FS LSS & 54  \% & 42.9\%  & 16.2\% & 40.1\%\\
         FS LSS (DH) & 50.8\% & 37.8\% & 12.5\% & 28.9\% \\
         \hline
    \end{tabular}
    \vspace{8pt}
    \caption{CSS and LSS over \(\Omega\) fractions in Near-field (NF) and Focus-Source (FS)] instances.
    [(DH) disregards the convergent halfspace].
    }
    \label{ex:sweet-table}
\end{table}

For the near-field instance, both the LSS and CSS generated by our method are more than 20 and 10 points (respectively) larger than that generated by any other method. The LSS and CSS generated by NFC-HOA (Figs.~\ref{ex:NFCHOA-NF-C}, \ref{ex:NFCHOA-NF-L}) are centered, whereas that generated by WFS (Figs.~\ref{ex:WFS-NF-C}, \ref{ex:WFS-NF-L}) are localized farther away from the source. This is correlated with their degradation of the sound field, which is consistent with the analysis in~\cite{daniel2003further}. Moreover, their LSS are consistent with the empirical results exposed in~\cite{wierstorf2017assessing}: the perceived localization for NFC-HOA degrades away from the center, whereas for WFS the perceived localization is fairly good over almost all the listening region. However, we believe that the perceived localization for WFS and NFC-HOA behaves slightly worse here than in~\cite{wierstorf2017assessing} because we analyze sinusoidal signals instead of Gaussian white noise, which has uniform spectral content. In contrast, the LSS and CSS generated by our method (Figs.~\ref{ex:SWEET-NF-C}, \ref{ex:SWEET-NF-L}) behave like those generated by WFS, but almost encompasses those generated by NFC-HOA. The LSS of \(L^2\)-PMM (Fig.~\ref{ex:L2-NF-L}) is larger than that of WFS and NFC-HOA, but its CSS (Fig.~\ref{ex:L2-NF-C}) is almost negligible. This is consistent with the sound wave \(u\) produced with \(L^2\)-PMM (Fig.~\ref{ex:L2-NF-u}) as the spatial phase of the signals is fairly well reconstructed, whereas its amplitude is too small. 

In the focus-source instance, due to reasons of causality, theoretically any method can achieve the correct reproduction of the direction of propagation of \(u_0\) in one half-space defined by \(\set{x_i}_{i=0}^{\Ns}\), where the sound field diverges from the focus-source position~\cite{ahrens2008focusing}. In the other half-space the reproduced wave field converges towards the location of the virtual source. As shown in Figs.~\ref{ex:SWEET-FS-L}-h, the LSS of all the methods comprise a portion of the converging part of the sound field. This is possible because the interaural phase of the binaural sinusoidal signals is correctly recreated at those points. However, we believe that in more complex scenarios, involving multi-frequency signals and allowing the listener to turn her head, e.g., considering a larger \(\Theta\), it would not be possible to recreate correctly and consistently the localization illusion. As shown in~\cite[Chapter 5.6]{ahrens2012analytic} {\em ``for listeners located in the converging part of the sound field, the perception is unpredictable since the interaural cues are either contradictory or change in a contradictory way when the listener moves the head.''} Table~\ref{ex:sweet-table} contains the fraction of the LSS over \(\Omega\) for the focus-source instance considering all the points of the converging half-space as incorrectly reconstructed, denoted by the label (DH).

For the focus-source instance, the LSS and CSS generated by our method (Fig.~\ref{ex:SWEET-FS-L}, \ref{ex:SWEET-FS-C}) is approximately 10 and 30 (respectively) points larger than those generated by other methods. The LSS and CSS generated by NFC-HOA (Fig.~\ref{ex:NFCHOA-FS-L}, \ref{ex:NFCHOA-FS-C}) is concentrated in a limited region at the diverging part and around a vertical line that passes through \(x_0\). The LSS generated by WFS (Fig.~\ref{ex:WFS-FS-L}) is almost contained in the same vertical line, although the error in the localization reconstruction is below 15 degrees in a larger area. The CSS generated by WFS (Fig.~\ref{ex:WFS-FS-C}) is almost empty as the resulting \(u\) has a large amplitude. This suggest that a focus-source formulation for WFS needs a factor for amplitude normalization. The LSS of \(L^2\)-PMM (Fig.~\ref{ex:L2-FS-L}) has almost the same size as that of NFC-HOA, but its CSS (Fig.~\ref{ex:L2-NF-C}) is almost negligible. This is consistent with the sound wave \(u\) produced with \(L^2\)-PMM (Fig.~\ref{ex:L2-FS-u}) as the spatial phase of the signals is fairly well reconstructed, whereas its amplitude is too small. The LSS and CSS generated by our method (Fig.~\ref{ex:SWEET-FS-L}, \ref{ex:SWEET-FS-C}) comprises almost all the divergent part. This highlights one of the advantages of the greedy approach of SWEET-ReLU: it is capable to detect the direction of \(u_0\) over \(\Omega\) in its first iterations, to then prioritize the part of \(\Omega\) where a good fit to \(\bar{u}_0\) can be obtained. This is a possible explanation for the amplitude mismatch of \(L^2\)-PMM both in the near-field (Fig.~\ref{ex:L2-NF-u}) and the focus-source (Fig.~\ref{ex:L2-FS-u}) instances: just minimizing the square of the spatial errors strongly penalizes the spatial points where the amplitude is very large and difficult to reconstruct, i.e., near the loudspeakers. Then, \(L^2\)-PMM finds a solution where the amplitude error at those points is not too large, leading to a small overall amplitude.

\section{Discussion}
\label{sec:discussion}

Our results show the SWEET-ReLU method yields state-of-the-art results in standard numerical experiments with our proof-of-concept implementation. We believe the performance in these experiments is representative of what we would observe when using more complex pyscho-acoustic models for the perceptual dissimilarity and the loudness discomfort. A key component of our method is the perceptual dissimilarity \(D\). Although its form in our proof-of-concept implementation is quite flexible, it does not account for spatialization and other binaural effects. Finding a model to account for these effects such that \(D\) satisfies~\eqref{eq:dissimilarityIsConvex} is the subject of future research. It should be noted that, as it is shown in~\cite{rumsey2005relative}, the overall quality of a spatial sound system can be explained to 70\% by coloration or timbral fidelity, which can be partially characterized by monoaural effects, and 30\% by spatial fidelity, which needs to be characterized by binaural effects. 

Furthermore, our proof-of-concept experiments show that even though the perceptual model we use is an extension of a monoaural model using a worst-case approach, it still is able to perform better than state-of-the-art methods in terms of localization and coloration. This suggests our implementation with this model is able to capture correctly some of the spatial properties of the auditory scene, even though these properties are not explicitly in the model. This might be explained because the coloration and localization is strongly dependant of the amplitude and phase of the binaural signals, which are controlled by our implementation of a binaural extension of an amplitude and phase-sensitive monaural model. 

Although our implementation assumes the loudspeakers and the sources are monopoles, we believe our method can be readily implemented in real settings with non-trivial sound sources. For instance, reverberation, different radiation patterns for the loudspeakers, and other time-invariant effects can be incorporated by modifying the Green function \(G_k\) in~\eqref{eq:synthethizedSignalFourier} and the transfer functions in~\eqref{eq:synthethizedEarSignalFourier} accordingly. 

Finally, although we have not fully developed a theory for the convergence of SWEET-ReLU, our numerical experiments show that the method converges to reasonable results in practice. Furthermore, our proof-of-concept implementation avoids any potential issues arising from the discretization of the models, either due to numerical computation of the Green function or transfer functions, or to the discretization of the integral that defines the weighted area. Further analysis about this point will be the subject of future work.

\section{Conclusion}
\label{sec:conclusion}

In this work, we introduced a theoretical framework for spatial audio perception that allows the definition of a perceptual sweet spot, that is, the region where the spatial auditory illusion is achieved when approximating one sound wave by another. Furthermore, we developed a method that finds an approximating sound wave that maximizes this sweet spot while guaranteeing no loudness discomfort over a spatial region of interest. 
We provided a theoretical analysis of the method, and an efficient algorithm, the SWEET-ReLU algorithm, for its numerical implementation. In a proof-of-concept implementation using monopoles emitting (pseudo) sinusoidal signals, our method successfully captures some of the spatial properties of the auditory scene, such as localization and coloration, even though these properties are not explicitly in the model. We believe our method is a first step towards a novel approach for spatial sound with loudspeakers, bridging a gap between methods based on perceptual principles, and sound field synthesis methods.

\section*{Acknowledgment}

C.~A.~SL. was partially funded by ANID -- FONDECYT -- 1211643, ANID -- Millennium Science Initiative Program -- NCN17\_059 and ANID -- Millennium Science Initiative Program -- NCN17\_129.

\bibliographystyle{unsrt}
\bibliography{references}

\begin{thebibliography}{100}

\bibitem{nicol2020creating}
Rozenn Nicol.
\newblock Creating auditory illusions with spatial-audio technologies.
\newblock In {\em The Technology of Binaural Understanding}, pages 581--622.
  Springer, 2020.

\bibitem{wierstorf2013binaural}
Hagen Wierstorf, Alexander Raake, and Sascha Spors.
\newblock Binaural assessment of multichannel reproduction.
\newblock In {\em The technology of binaural listening}, pages 255--278.
  Springer, 2013.

\bibitem{blauert1997spatial}
Jens Blauert.
\newblock {\em Spatial hearing: the psychophysics of human sound localization}.
\newblock MIT press, 1997.

\bibitem{spors2013}
Sascha Spors, Hagen Wierstorf, Alexander Raake, Frank Melchior, Matthias Frank,
  and Franz Zotter.
\newblock Spatial sound with loudspeakers and its perception: A review of the
  current state.
\newblock {\em Proceedings of the IEEE}, 101(9):1920--1938, 2013.

\bibitem{leakey1959some}
DM~Leakey.
\newblock Some measurements on the effects of interchannel intensity and time
  differences in two channel sound systems.
\newblock {\em The Journal of the Acoustical Society of America},
  31(7):977--986, 1959.

\bibitem{wierstorf2014perceptual}
Hagen Wierstorf.
\newblock {\em Perceptual assessment of sound field synthesis}.
\newblock PhD thesis, 2014.

\bibitem{frank2017exploring}
Matthias Frank and Franz Zotter.
\newblock Exploring the perceptual sweet area in ambisonics.
\newblock {\em Journal of the Audio Engineering Society}, may 2017.

\bibitem{rumsey2018}
Francis Rumsey.
\newblock Surround sound.
\newblock In {\em In Immersive Sound: The Art and Science of Binaural and
  Multi-Channel Audio}, page Chap. 6. Focal Press, 2017.

\bibitem{pulkki1997virtual}
Ville Pulkki.
\newblock Virtual sound source positioning using vector base amplitude panning.
\newblock {\em Journal of the Audio Engineering sSciety}, 45(6):456--466, 1997.

\bibitem{pulkki2001coloration}
Ville Pulkki.
\newblock Coloration of amplitude-panned virtual sources.
\newblock In {\em Audio Engineering Society Convention 110}. Audio Engineering
  Society, 2001.

\bibitem{huygens}
Christiaan Huygens.
\newblock {\em Traité de la lumière}.
\newblock Pierre Vander Aa Marchand Libraire, 1690.

\bibitem{daniel2000representation}
J{\'e}r{\^o}me Daniel.
\newblock {\em Repr{\'e}sentation de champs acoustiques, application {\`a} la
  transmission et {\`a} la reproduction de sc{\`e}nes sonores complexes dans un
  contexte multim{\'e}dia}.
\newblock PhD thesis, University of Paris VI, 2000.

\bibitem{gerzon1973periphony}
Michael~A Gerzon.
\newblock Periphony: With-height sound reproduction.
\newblock {\em Journal of the Audio Engineering Society}, 21(1):2--10, 1973.

\bibitem{daniel2003further}
J{\'e}r{\^o}me Daniel, Sebastien Moreau, and Rozenn Nicol.
\newblock Further investigations of high-order ambisonics and wavefield
  synthesis for holophonic sound imaging.
\newblock In {\em Audio Engineering Society Convention 114}. Audio Engineering
  Society, 2003.

\bibitem{Ward2001}
Darren~B Ward and Thushara~D Abhayapala.
\newblock Reproduction of a plane-wave sound field using an array of
  loudspeakers.
\newblock {\em IEEE/ACM Transactions on Audio Speech and Language Processing},
  9, 2001.

\bibitem{Ueno2019}
Natsuki Ueno, Shoichi Koyama, and Hiroshi Saruwatari.
\newblock Three-dimensional sound field reproduction based on weighted
  mode-matching method.
\newblock {\em IEEE/ACM Transactions on Audio Speech and Language Processing},
  27:1852--1867, 12 2019.

\bibitem{Zuo2020}
Huanyu Zuo, Thushara~D. Abhayapala, and Prasanga~N. Samarasinghe.
\newblock Particle velocity assisted three dimensional sound field reproduction
  using a modal-domain approach.
\newblock {\em IEEE/ACM Transactions on Audio Speech and Language Processing},
  28:2119--2133, 2020.

\bibitem{zuo20213d}
Huanyu Zuo, Thushara~D Abhayapala, and Prasanga~N Samarasinghe.
\newblock 3d multizone soundfield reproduction in a reverberant environment
  using intensity matching method.
\newblock In {\em ICASSP 2021-2021 IEEE International Conference on Acoustics,
  Speech and Signal Processing (ICASSP)}, pages 416--420. IEEE, 2021.

\bibitem{kirkeby1993reproduction}
Ole Kirkeby and Philip~A Nelson.
\newblock Reproduction of plane wave sound fields.
\newblock {\em The Journal of the Acoustical Society of America},
  94(5):2992--3000, 1993.

\bibitem{Gauthier2005}
Philippe-Aubert Gauthier, Alain Berry, and Wieslaw Woszczyk.
\newblock Sound-field reproduction in-room using optimal control techniques:
  Simulations in the frequency domain.
\newblock {\em The Journal of the Acoustical Society of America}, 117:662--678,
  2 2005.

\bibitem{Gauthier2017}
P.-A. Gauthier, P.~Lecomte, and A.~Berry.
\newblock Source sparsity control of sound field reproduction using the
  elastic-net and the lasso minimizers.
\newblock {\em The Journal of the Acoustical Society of America},
  141:2315--2326, 4 2017.

\bibitem{Jia2018}
Maoshen Jia, Jiaming Zhang, Yuxuan Wu, and Jing Wang.
\newblock Sound field reproduction via the alternating direction method of
  multipliers based lasso plus regularized least-square.
\newblock {\em IEEE Access}, 6:54550--54563, 9 2018.

\bibitem{Feng2018}
Qipeng Feng, Feiran Yang, and Jun Yang.
\newblock Time-domain sound field reproduction using the group lasso.
\newblock {\em The Journal of the Acoustical Society of America},
  143:EL55--EL60, 2 2018.

\bibitem{Radmanesh2013}
Nasim Radmanesh and Ian~S. Burnett.
\newblock Generation of isolated wideband sound fields using a combined
  two-stage lasso-ls algorithm.
\newblock {\em IEEE Transactions on Audio, Speech and Language Processing},
  21:378--387, 2013.

\bibitem{Lilis2010}
Georgios~N. Lilis, Daniele Angelosante, and Georgios~B. Giannakis.
\newblock Sound field reproduction using the lasso.
\newblock {\em IEEE Transactions on Audio, Speech and Language Processing},
  18:1902--1912, 2010.

\bibitem{ajdler2006plenacoustic}
Thibaut Ajdler, Luciano Sbaiz, and Martin Vetterli.
\newblock The plenacoustic function and its sampling.
\newblock {\em IEEE transactions on Signal Processing}, 54(10):3790--3804,
  2006.

\bibitem{kolundzija2009sound}
Mihailo Kolundzija, Christof Faller, and Martin Vetterli.
\newblock Sound field reconstruction: An improved approach for wave field
  synthesis.
\newblock In {\em Audio Engineering Society Convention 126}. Audio Engineering
  Society, 2009.

\bibitem{buerger2015multizone}
Michael Buerger, Roland Maas, Heinrich~W L{\"o}llmann, and Walter Kellermann.
\newblock Multizone sound field synthesis based on the joint optimization of
  the sound pressure and particle velocity vector on closed contours.
\newblock In {\em 2015 IEEE Workshop on Applications of Signal Processing to
  Audio and Acoustics (WASPAA)}, pages 1--5. IEEE, 2015.

\bibitem{buerger2018broadband}
Michael Buerger, Christian Hofmann, and Walter Kellermann.
\newblock Broadband multizone sound rendering by jointly optimizing the sound
  pressure and particle velocity.
\newblock {\em The Journal of the Acoustical Society of America},
  143(3):1477--1490, 2018.

\bibitem{Shin2016}
Mincheol Shin, Philip~A. Nelson, Filippo~M. Fazi, and Jeongil Seo.
\newblock Velocity controlled sound field reproduction by non-uniformly spaced
  loudspeakers.
\newblock {\em Journal of Sound and Vibration}, 370:444--464, 5 2016.

\bibitem{berkhout1993acoustic}
Augustinus~J Berkhout, Diemer de~Vries, and Peter Vogel.
\newblock Acoustic control by wave field synthesis.
\newblock {\em The Journal of the Acoustical Society of America},
  93(5):2764--2778, 1993.

\bibitem{stuart1996application}
EW~Stuart.
\newblock Application of curved arrays in wave field synthesis.
\newblock In {\em Audio Engineering Society Convention 100}. Audio Engineering
  Society, 1996.

\bibitem{spors2008theory}
Sascha Spors, Rudolf Rabenstein, and Jens Ahrens.
\newblock The theory of wave field synthesis revisited.
\newblock In {\em In 124th Convention of the AES}. Citeseer, 2008.

\bibitem{wierstorf2017assessing}
Hagen Wierstorf, Alexander Raake, and Sascha Spors.
\newblock Assessing localization accuracy in sound field synthesis.
\newblock {\em The Journal of the Acoustical Society of America},
  141(2):1111--1119, 2017.

\bibitem{wierstorf2014coloration}
Hagen Wierstorf, Christoph Hohnerlein, Sascha Spors, and Alexander Raake.
\newblock Coloration in wave field synthesis.
\newblock In {\em Audio Engineering Society Conference: 55th International
  Conference: Spatial Audio}. Audio Engineering Society, 2014.

\bibitem{spors2008comparison}
Sascha Spors and Jens Ahrens.
\newblock A comparison of wave field synthesis and higher-order ambisonics with
  respect to physical properties and spatial sampling.
\newblock In {\em Audio Engineering Society Convention 125}. Audio Engineering
  Society, 2008.

\bibitem{fazi2009analogies}
Filippo~M Fazi, Philip~A Nelson, and Roland Potthast.
\newblock Analogies and differences between three methods for sound field
  reproduction.
\newblock {\em Relation}, 10:3, 2009.

\bibitem{Franck2017}
Andreas Franck, Wenwu Wang, and Filippo~Maria Fazi.
\newblock Sparse 1-optimal multiloudspeaker panning and its relation to vector
  base amplitude panning.
\newblock {\em IEEE/ACM TRANSACTIONS ON AUDIO, SPEECH, AND LANGUAGE
  PROCESSING}, 25, 2017.

\bibitem{Firtha2018}
Gergely Firtha, Péter Fiala, Frank Schultz, and Sascha Spors.
\newblock On the general relation of wave field synthesis and spectral division
  method for linear arrays.
\newblock {\em IEEE/ACM Transactions on Audio Speech and Language Processing},
  26:2393--2403, 12 2018.

\bibitem{fazi2007theoretical}
Filippo~Maria Fazi and Philip~A. Nelson.
\newblock A theoretical study of sound field reconstruction techniques.
\newblock In {\em 19th International Congress on Acoustics}, September 2007.

\bibitem{johnston2000perceptual}
James~D Johnston and Yin Hay~Vicky Lam.
\newblock Perceptual soundfield reconstruction.
\newblock In {\em Audio Engineering Society Convention 109}. Audio Engineering
  Society, 2000.

\bibitem{desena2013}
Enzo~De Sena, Huseyin Hacihabiboglu, and Zoran Cvetkovic.
\newblock Analysis and design of multichannel systems for perceptual sound
  field reconstruction.
\newblock {\em IEEE Transactions on Audio, Speech and Language Processing},
  21:1653--1665, 2013.

\bibitem{Ziemer2017}
Tim Ziemer and Rolf Bader.
\newblock Psychoacoustic sound field synthesis for musical instrument radiation
  characteristics.
\newblock {\em AES: Journal of the Audio Engineering Society}, 65:482--496, 6
  2017.

\bibitem{litovsky1999precedence}
Ruth~Y Litovsky, H~Steven Colburn, William~A Yost, and Sandra~J Guzman.
\newblock The precedence effect.
\newblock {\em The Journal of the Acoustical Society of America},
  106(4):1633--1654, 1999.

\bibitem{Lee2020}
Taewoong Lee, Jesper~Kjaer Nielsen, and Mads~Graesboll Christensen.
\newblock Signal-adaptive and perceptually optimized sound zones with variable
  span trade-off filters.
\newblock {\em IEEE/ACM Transactions on Audio Speech and Language Processing},
  28:2412--2426, 2020.

\bibitem{Evans2010}
Lawrence~C. Evans.
\newblock {\em {Partial Differential Equations}}.
\newblock American Mathematical Society, 2nd edition, 2010.

\bibitem{betlehem2005theory}
Terence Betlehem and Thushara~D Abhayapala.
\newblock Theory and design of sound field reproduction in reverberant rooms.
\newblock {\em The Journal of the Acoustical Society of America},
  117(4):2100--2111, 2005.

\bibitem{williams1999fourier}
Earl~G Williams.
\newblock {\em Fourier acoustics: sound radiation and nearfield acoustical
  holography}.
\newblock Academic press, 1999.

\bibitem{Krantz2002}
Steven~G. Krantz and Harold~R. Parks.
\newblock {\em {A Primer of Real Analytic Functions}}.
\newblock Birkh{\"{a}}user Boston, Boston, MA, 2002.

\bibitem{deutsch1983auditory}
Diana Deutsch.
\newblock Auditory illusions, handedness, and the spatial environment.
\newblock {\em Journal of the Audio Engineering Society}, 31(9):606--620,
  september 1983.

\bibitem{francombe2015elicitation}
Jon Francombe, Tim Brookes, and Russell Mason.
\newblock Elicitation of the differences between real and reproduced audio.
\newblock {\em Journal of the Audio Engineering Society}, may 2015.

\bibitem{bregman1994auditory}
Albert~S Bregman.
\newblock {\em Auditory scene analysis: The perceptual organization of sound}.
\newblock MIT press, 1994.

\bibitem{wang2006computational}
DeLiang Wang and Guy~J Brown.
\newblock {\em Computational auditory scene analysis: Principles, algorithms,
  and applications}.
\newblock Wiley-IEEE press, 2006.

\bibitem{blauert1999models}
J~Blauert.
\newblock Models of binaural hearing: architectural considerations.
\newblock In {\em Proc. 18 th DANAVOX Symposium 1999}, pages 189--206, 1999.

\bibitem{raake2020binaural}
Alexander Raake and Hagen Wierstorf.
\newblock Binaural evaluation of sound quality and quality of experience.
\newblock In {\em The Technology of Binaural Understanding}, pages 393--434.
  Springer, 2020.

\bibitem{raake2013comprehensive}
Alexander Raake and Jens Blauert.
\newblock Comprehensive modeling of the formation process of sound-quality.
\newblock In {\em 2013 Fifth International Workshop on Quality of Multimedia
  Experience (QoMEX)}, pages 76--81. IEEE, 2013.

\bibitem{letowski1989sound}
Tomasz Letowski.
\newblock Sound quality assessment: concepts and criteria.
\newblock In {\em Audio Engineering Society Convention 87}. Audio Engineering
  Society, 1989.

\bibitem{francombe2017evaluation}
Jon Francombe, Tim Brookes, and Russell Mason.
\newblock Evaluation of spatial audio reproduction methods (part 1):
  Elicitation of perceptual differences.
\newblock {\em Journal of the Audio Engineering Society}, 65(3):198--211, march
  2017.

\bibitem{reardon2018evaluation}
Gregory Reardon, Andrea Genovese, Gabriel Zalles, Patrick Flanagan, and
  Agnieszka Roginska.
\newblock Evaluation of binaural renderers: Multidimensional sound quality
  assessment.
\newblock {\em Journal of the Audio Engineering Society}, august 2018.

\bibitem{dietz2021computational}
Mathias Dietz and Go~Ashida.
\newblock Computational models of binaural processing.
\newblock In {\em Binaural Hearing}, pages 281--315. Springer, 2021.

\bibitem{raleigh1907our}
SJW Raleigh~Lord.
\newblock On our perception of sound direction.
\newblock {\em Phil Mag}, 13:314--232, 1907.

\bibitem{wightman1992dominant}
Frederic~L Wightman and Doris~J Kistler.
\newblock The dominant role of low-frequency interaural time differences in
  sound localization.
\newblock {\em The Journal of the Acoustical Society of America},
  91(3):1648--1661, 1992.

\bibitem{dietz2011auditory}
Mathias Dietz, Stephan~D Ewert, and Volker Hohmann.
\newblock Auditory model based direction estimation of concurrent speakers from
  binaural signals.
\newblock {\em Speech Communication}, 53(5):592--605, 2011.

\bibitem{brughera2013human}
Andrew Brughera, Larisa Dunai, and William~M Hartmann.
\newblock Human interaural time difference thresholds for sine tones: The
  high-frequency limit.
\newblock {\em The Journal of the Acoustical Society of America},
  133(5):2839--2855, 2013.

\bibitem{cherry1956human}
E~Colin Cherry and Bruce Mc~A Sayers.
\newblock “human ‘cross-correlator’”—a technique for measuring
  certain parameters of speech perception.
\newblock {\em The Journal of the Acoustical Society of America},
  28(5):889--895, 1956.

\bibitem{jeffress1948place}
Lloyd~A Jeffress.
\newblock A place theory of sound localization.
\newblock {\em Journal of comparative and physiological psychology}, 41(1):35,
  1948.

\bibitem{lindemann1986extension}
Werner Lindemann.
\newblock Extension of a binaural cross-correlation model by contralateral
  inhibition. i. simulation of lateralization for stationary signals.
\newblock {\em The Journal of the Acoustical Society of America},
  80(6):1608--1622, 1986.

\bibitem{blauert1978some}
Jens Blauert and W~Cobben.
\newblock Some consideration of binaural cross correlation analysis.
\newblock {\em Acta Acustica united with Acustica}, 39(2):96--104, 1978.

\bibitem{braasch2002localization}
Jonas Braasch.
\newblock Localization in the presence of a distracter and reverberation in the
  frontal horizontal plane: Ii. model algorithms.
\newblock {\em Acta Acustica united with Acustica}, 88(6):956--969, 2002.

\bibitem{nix2006sound}
Johannes Nix and Volker Hohmann.
\newblock Sound source localization in real sound fields based on empirical
  statistics of interaural parameters.
\newblock {\em J. Acoust. Soc. Am}, 119(1):463--479, 2006.

\bibitem{pulkki2009functional}
Ville Pulkki and Toni Hirvonen.
\newblock Functional count-comparison model for binaural decoding.
\newblock {\em Acta Acustica united with Acustica}, 95(5):883--900, 2009.

\bibitem{mcalpine2003sound}
David McAlpine and Benedikt Grothe.
\newblock Sound localization and delay lines--do mammals fit the model?
\newblock {\em Trends in neurosciences}, 26(7):347--350, 2003.

\bibitem{culling2021binaural}
John~F Culling and Mathieu Lavandier.
\newblock Binaural unmasking and spatial release from masking.
\newblock In {\em Binaural Hearing}, pages 209--241. Springer, 2021.

\bibitem{bruggen2001sound}
M~Br{\"u}ggen.
\newblock {\em Sound coloration due to reflections and its auditory and
  instrumental compensation}.
\newblock PhD thesis, PhD thesis, Ruhr-Universit{\"a}t Bochum, 2001.

\bibitem{breebaart2001binaural}
Jeroen Breebaart, Steven Van De~Par, and Armin Kohlrausch.
\newblock Binaural processing model based on contralateral inhibition. i. model
  structure.
\newblock {\em The Journal of the Acoustical Society of America},
  110(2):1074--1088, 2001.

\bibitem{park2008model}
Munhum Park, Philip~A Nelson, and Kyeongok Kang.
\newblock A model of sound localisation applied to the evaluation of systems
  for stereophony.
\newblock {\em Acta Acustica united with Acustica}, 94(6):825--839, 2008.

\bibitem{pulkki1999analyzing}
Ville Pulkki, Matti Karjalainen, and Jyri Huopaniemi.
\newblock Analyzing virtual sound source attributes using a binaural auditory
  model.
\newblock {\em Journal of the Audio Engineering Society}, 47(4):203--217, 1999.

\bibitem{mckenzie2022predicting}
Thomas McKenzie, Cal Armstrong, Lauren Ward, Damian~T Murphy, and Gavin
  Kearney.
\newblock Predicting the colouration between binaural signals.
\newblock {\em Applied Sciences}, 12(5):2441, 2022.

\bibitem{thiede2000peaq}
Thilo Thiede, William~C Treurniet, Roland Bitto, Christian Schmidmer, Thomas
  Sporer, John~G Beerends, and Catherine Colomes.
\newblock Peaq-the itu standard for objective measurement of perceived audio
  quality.
\newblock {\em Journal of the Audio Engineering Society}, 48(1/2):3--29, 2000.

\bibitem{zakarauskas1993computational}
Pierre Zakarauskas and Max~S Cynader.
\newblock A computational theory of spectral cue localization.
\newblock {\em The Journal of the Acoustical Society of America},
  94(3):1323--1331, 1993.

\bibitem{par2005}
Steven van~de Par, Armin Kohlrausch, Richard Heusdens, Jesper Jensen, and
  S{\o}ren~Holdt Jensen.
\newblock A perceptual model for sinusoidal audio coding based on spectral
  integration.
\newblock {\em EURASIP Journal on Advances in Signal Processing},
  2005(9):1--13, 2005.

\bibitem{painter2000perceptual}
Ted Painter and Andreas Spanias.
\newblock Perceptual coding of digital audio.
\newblock {\em Proceedings of the IEEE}, 88(4):451--515, 2000.

\bibitem{Jepsen2008}
Morten~L. Jepsen, Stephan~D. Ewert, and Torsten Dau.
\newblock A computational model of human auditory signal processing and
  perception.
\newblock {\em The Journal of the Acoustical Society of America}, 124:422--438,
  7 2008.

\bibitem{Plasberg2007}
Jan~H. Plasberg and W.~Bastiaan Kleijn.
\newblock The sensitivity matrix: Using advanced auditory models in speech and
  audio processing.
\newblock {\em IEEE Transactions on Audio, Speech and Language Processing},
  15:310--319, 1 2007.

\bibitem{Taal2012}
Cees~H. Taal, Richard~C. Hendriks, and Richard Heusdens.
\newblock A low-complexity spectro-temporal distortion measure for audio
  processing applications.
\newblock {\em IEEE Transactions on Audio, Speech and Language Processing},
  20:1553--1564, 2012.

\bibitem{knobel2006nivel}
Keila Alessandra~Baraldi Knobel and Tanit~Ganz Sanchez.
\newblock N{\'\i}vel de desconforto para sensa{\c{c}}{\~a}o de intensidade em
  indiv{\'\i}duos com audi{\c{c}}{\~a}o normal.
\newblock {\em Pr{\'o}-Fono Revista de Atualiza{\c{c}}{\~a}o Cient{\'\i}fica},
  18(1):31--40, 2006.

\bibitem{sherlock2005estimates}
LaGuinn~P Sherlock and Craig Formby.
\newblock Estimates of loudness, loudness discomfort, and the auditory dynamic
  range: normative estimates, comparison of procedures, and test-retest
  reliability.
\newblock {\em Journal of the American Academy of Audiology}, 16(02):085--100,
  2005.

\bibitem{terhardt1979calculating}
Ernst Terhardt.
\newblock Calculating virtual pitch.
\newblock {\em Hearing research}, 1(2):155--182, 1979.

\bibitem{glasberg1990derivation}
Brian~R Glasberg and Brian~CJ Moore.
\newblock Derivation of auditory filter shapes from notched-noise data.
\newblock {\em Hearing research}, 47(1-2):103--138, 1990.

\bibitem{quarteroni2010numerical}
Alfio Quarteroni, Riccardo Sacco, and Fausto Saleri.
\newblock {\em Numerical mathematics}, volume~37.
\newblock Springer Science \& Business Media, 2010.

\bibitem{diamond2016cvxpy}
Steven Diamond and Stephen Boyd.
\newblock {CVXPY}: {A} {P}ython-embedded modeling language for convex
  optimization.
\newblock {\em Journal of Machine Learning Research}, 17(83):1--5, 2016.

\bibitem{agrawal2018rewriting}
Akshay Agrawal, Robin Verschueren, Steven Diamond, and Stephen Boyd.
\newblock A rewriting system for convex optimization problems.
\newblock {\em Journal of Control and Decision}, 5(1):42--60, 2018.

\bibitem{mosek}
MOSEK ApS.
\newblock {\em The MOSEK optimization toolbox for Python manual. Version
  9.2.44}, 2019.

\bibitem{wierstorf2012sound}
Hagen Wierstorf and Sascha Spors.
\newblock Sound field synthesis toolbox.
\newblock In {\em Audio Engineering Society Convention 132}. Audio Engineering
  Society, 2012.

\bibitem{ahrens2009spatial}
Jens Ahrens and Sascha Spors.
\newblock Spatial encoding and decoding of focused virtual sound sources.
\newblock In {\em Ambisonics Symposium}, pages 25--27, 2009.

\bibitem{wierstorf2011free}
Hagen Wierstorf, Matthias Geier, and Sascha Spors.
\newblock A free database of head related impulse response measurements in the
  horizontal plane with multiple distances.
\newblock In {\em Audio Engineering Society Convention 130}. Audio Engineering
  Society, 2011.

\bibitem{majdak2021amt}
Piotr Majdak, Clara Hollomey, and Robert Baumgartner.
\newblock Amt 1.0: The toolbox for reproducible research in auditory modeling.
\newblock {\em submitted to Acta Acustica}, 2021.

\bibitem{ahrens2012analytic}
Jens Ahrens.
\newblock {\em Analytic methods of sound field synthesis}.
\newblock Springer Science \& Business Media, 2012.

\bibitem{ahrens2008focusing}
Jens Ahrens and Sascha Spors.
\newblock Focusing of virtual sound sources in higher order ambisonics.
\newblock In {\em Audio Engineering Society Convention 124}. Audio Engineering
  Society, 2008.

\bibitem{rumsey2005relative}
Francis Rumsey, S{\l}awomir Zieli{\'n}ski, Rafael Kassier, and S{\o}ren Bech.
\newblock On the relative importance of spatial and timbral fidelities in
  judgments of degraded multichannel audio quality.
\newblock {\em The Journal of the Acoustical Society of America},
  118(2):968--976, 2005.

\bibitem{Cohn2013}
Donald~L. Cohn.
\newblock {\em {Measure Theory}}.
\newblock Birkh{\"{a}}user Advanced Texts Basler Lehrb{\"{u}}cher. Springer New
  York, New York, NY, 2nd edition, 2013.

\bibitem{Bauschke2011}
Heinz~H. Bauschke and Patrick~L. Combettes.
\newblock {\em {Convex Analysis and Monotone Operator Theory in Hilbert
  Spaces}}.
\newblock CMS Books in Mathematics. Springer New York, New York, NY, 2011.

\bibitem{Tao1997}
Pham~Dinh Tao and Le~Thi~Hoai An.
\newblock {Convex Analysis Approach to D.C. Programming: Theory, Algorithms and
  Applications}.
\newblock {\em Acta Mathematica Vietnamica}, 22(1):289--355, 1997.

\bibitem{Horst1999}
R.~Horst and N.~V. Thoai.
\newblock {DC Programming: Overview}.
\newblock {\em Journal of Optimization Theory and Applications}, 103(1):1--43,
  oct 1999.

\bibitem{Lipp2016}
Thomas Lipp and Stephen Boyd.
\newblock {Variations and extension of the convex–concave procedure}.
\newblock {\em Optimization and Engineering}, 17(2):263--287, jun 2016.

\bibitem{Barbu2012}
Viorel Barbu and Teodor Precupanu.
\newblock {\em {Convexity and Optimization in Banach Spaces}}.
\newblock Springer Monographs in Mathematics. Springer Netherlands, Dordrecht,
  4th edition, 2012.

\bibitem{Sogge2017}
Christopher~D. Sogge.
\newblock {\em {Fourier Integrals in Classical Analysis}}.
\newblock Cambridge University Press, Cambridge, 2nd edition, 2017.

\bibitem{Rudin1986}
Walter Rudin.
\newblock {\em {Real and Complex Analysis}}.
\newblock McGraw-Hill Education, 3rd edition, 1986.

\bibitem{Aubin1990}
Jean-Pierre Aubin and Hélène Frankowska.
\newblock {\em Set-Valued Analysis}.
\newblock Birkhauser, 1990.

\end{thebibliography}

\newpage
\appendix
\section{Appendix: The SWEET method}
\label{sec:analysisOfSWEET}

\subsection{Preliminaries}
\label{sec:maxSweetSpot}

\def\mF{\mathcal{F}}

We let \(\LE(\R)\) be the space of (equivalence classes of) {\em complex-valued functions} that are modulus-square integrable with respect to the Lebesgue measure on \(\R\). Let
\[
    X_E := \set{(\uleft,\uright):\uside\in \LE(\R),\, s\in \set{\ell,r}}
\]
be the space of {\em pairs of signals}. When endowed with
\[
    \nrm{\vu}_{X_E}^2 = \sum\nolimits_{s\in\set{\ell,r}}  \int_{\R} |\uside(t)|^2\, dt
\]
it becomes a complete metric space. Define the space
\[
    W := \set{\vu:\Omega\times\Th\to X_E: \mbox{\(\vu\) continuous and bounded}}
\]
of {\em spatial distributions of pairs of signals}; these are not equivalence classes. If \(\vu\in W\) then \(\vu_{(x,\th)}\in X_E\) for every \((x,\th)\in \Omega\times\Th\). When endowed with the norm
\[
    \nrmW{\vu} := \sup_{(x,\th)\in \Omega\times\Th}\, \nrm{\vu_{(x,\th)}}_{X_E}
\]
the space \(W\) is complete. The Fourier transform is an isometry in \(\LE(\R)\). We define \(\mF:W\to W\) as \(\mF\vu_{(x,\theta)} = (\wh{u}^{\ell}_{(x,\theta)}, \wh{u}^{r}_{(x,\theta)})\). It can be verified \(\mF\) is an isometry in \(W\). Let \(I_S\subset \R\) and let \(\cmax > 0\). The set of {\em admissible audio signals} driving each loudspeaker is
\[
    S_D := \set{c\in L^2(\R;\C):\wh{c}|_{I_S^c} = 0,\nrm{\wh{c}}_{\LE} \leq \cmax}
\]
where \(|\) denotes restriction; they are bandlimited to \(I_S\) and have norm bounded by \(\cmax\). The sound waves in~\eqref{eq:synthethizedEarSignalFourier} belong to
\[
    W_S := \mF^{-1}\left(\Lset{\left(\sum_{k=1}^{\Ns} \wh{c}_k H^r_k,\,\sum_{k=1}^{\Ns} \wh{c}_k H^\ell_k\right): c_k\in S_D}\right).
\]
To quantify the area of the sweet spot in~\eqref{eq:defSweetSpot} we use a finite Borel measure \(\mu\) on \(\Omega\)~\cite[Section~1.2]{Cohn2013} and we suppose that the resulting measure space is complete~\cite[Section~1.5]{Cohn2013}; we usually suppose \(\mu\) is absolutely continuous with respect to the Lebesgue measure or atomic. If \(\vu\) is the pair of signals generated by the array, then \(\mu(\S(\vu))\) is the weighted area of the sweet spot \(\S(\vu)\). We consider the space \(\LInfO\) of (equivalence classes of) real-valued Borel measurable functions that are bounded \(\mu\)-a.e.~\cite[Section~3.3]{Cohn2013}. Frequently used operations, such as the sum, supremum or infimum of functions, the integral, and inequalities, are well-defined for such equivalence classes of measurable functions.

\begin{assumption}
    \label{assmpt:general}
    (i) \(\Omega\) and \(\Th\) are compact. (ii) \(I_S\) is compact. (iii) The functions \(H_k^{\ell}, H^r_k\) are continuous and bounded on \(I_S\times \Omega\times\Th\). (iv) The dissimilarity map \(D:W\times W \to C^0(\Omega\times \Th)\) is continuous and convex on its first argument. (v) The discomfort map \(L:W \to C^0(\Omega\times \Th)\) is continuous and convex. (vi) \(\vu_0\in W\). 
\end{assumption}

The model proposed in Section~\ref{sec:sweetSpot} is well-defined. 

\begin{proposition}
    \label{prop:wellPosedA}
    Under Assumption~\ref{assmpt:general} the following assertions are true: (i) The set \(W_S\) is convex and compact in \(W\). (ii) The map \(\TD:W_S \to \LInfO\) is continuous, and for every \(\vu\in W_S\) the map \(\vu\to \TD\vu(x)\) in~\eqref{eq:defThresholdMap} is convex for \(\mu\)-a.e. \(x\). (iii) The map \(\mu\circ \S:W \to \R\) is well-defined, that is, its values do not depend on the choice of representative of \(\TD\vu\). (iv) The set \(\P\) in~\eqref{eq:defPainThreshold} is convex and closed in \(W\).
\end{proposition}

We defer the proof to Appendix~\ref{proof:wellPosedA}. Although the feasible set for \((P_0)\) is compact, the objective depends on the properties of the {\em set-valued function} \(u\rightrightarrows \S(u)\). As studying these properties and minimizing \(\mu\circ \S\) is potentially challenging, we propose an approximation to \((\Po)\) that can be analyzed and solved with standard methods.

\subsection{The layer-cake representation}
\label{sec:sweetSpot:layerCake}

We approximate the area of \(\S(u)\) using the {\em layer-cake representation}.

\begin{assumption}
\label{assmpt:phi}
    \(\vphi:\R\to\R\) is absolutely continuous, bounded, non-negative and such that \(\vphi(t) = 0\) for \(t< 0\) and \(\|\vphi\|_{L^1}=1\). 
\end{assumption}

For \(\eps > 0\) let \(\vphi_\eps\) denote \(\vphi_\eps(t) = \vphi(t/\eps)/\eps\) and define
\[
    \Phi_{\eps}(t) = \int_{-\infty}^t \vphi_{\eps}(s)\, ds.
\]
Since \(\Phi_\eps\) is continuous, the composition \(\Phi_\eps\circ v\) is well-defined as an element in \(\LInfO\) for any \(v\in \LInfO\). Similarly,
\[
    \Ae(v) := \int_{\Omega} \Phi_\eps(v(x))\, d\mu(x).
\]
is well-defined for any \(v\in \LInfO\). By Assumption~\ref{assmpt:phi} \(\Phi_\eps\) is non-decreasing, whence \(v_1 \leq v_2\) implies \(\Ae(v_1) \leq \Ae(v_2)\), i.e., \(\Ae\) is non-decreasing.

\begin{proposition}
    \label{prop:layerCake}
     Under Assumption~\ref{assmpt:phi}, for every \(v\in\LInfO\) and any representative \(v'\) of \(v\) we have
    \[
    \lim\nolimits_{\eps\downarrow 0}\, \Ae(v) = \mu(\set{x\in\Omega: v'(x) > 0}).
    \]
\end{proposition}

\begin{proof}[Proof of Proposition~\ref{prop:layerCake}]
    Let \(t \geq 0\) and let \(\set{\eps_n}_{n\in\N}\) be non-negative and monotone decreasing to zero. Define \(V_{t,n} :=\set{x\in\Omega: v'(x) \geq \eps_n t}\) and note that 
    \(V_{t,n} \subseteq V_{t,n+1}\). Define \(V := \bigcup_{n > 0} V_{t,n} = \set{x\in \Omega: v'(x) > 0}\) and \(h_n(t) = \vphi(t) \mu(V_{t,n})\). Note 
    \(h_n\) is measurable for every \(n\) as \(t\mapsto\mu(V_{t,n})\) is monotone. Then \(h_n(t) \uparrow \vphi(t) \mu(V)\) as \(n\to\infty\) by continuity from 
    below~\cite[Proposition~1.2.5]{Cohn2013}. Since \(\Omega\) is bounded, \(v\in L_\mu^1(\Omega)\) and \(v'\) is absolutely integrable. By Fubini's theorem,
    \begin{align*}
        A_{\eps_n}(v) &=\int_\Omega\int_{-\infty}^{v'(x)}\vphi_{\eps_n}(t) dt d\mu(x)\\
        &=\int_{\R}\vphi_{\eps_n}(t)\int_\Omega \chi_{\set{x\in\Omega:v'(x) \geq t}}(t,x)\ d\mu(x) dt\\
        &=\int_{\R}\vphi_{\eps_n}(t)\mu(\set{x\in\Omega: v'(x) \geq t}) dt\\
        &=\int_0^{\infty}\vphi(s)\mu(V_{s,n}) ds \xrightarrow{\eps\downarrow 0}\mu(\set{x\in\Omega: v'(x) > 0})
    \end{align*}
    where we used the change of variables \(s = t/\eps_n\) and the monotone convergence theorem~\cite[Theorem~2.4.1]{Cohn2013}. As \(\set{\eps_n}_{n\in\N}\) is arbitrary, the claim follows.
\end{proof}

For \(\vu\in W_S\) and \(\eps>0\) sufficiently small we have by Propositions~\ref{prop:wellPosedA} and~\ref{prop:layerCake} that
\begin{align*}
    \Ae(\TD\vu) &= \int_{\Omega} \Phi_\eps(\TD\vu(x))\, d\mu(x)\\ 
    &\approx \mu(\set{x\in \Omega: \TD\vu'(x) > 0}) = \mu(\Omega) - \mu(\S(\vu))
\end{align*}
where \(\TD\vu'\) is any representative of \(\TD\vu\). Consequently, \(\mu(\S(\vu)) \approx \mu(\Omega) - \Ae(\TD\vu)\) and for every fixed \(\vu\in W_S\) we can approximate \(\mu(\S(\vu))\) by \(\Ae(\TD\vu)\).

\subsection{The variational problem}
\label{sec:sweetSpot:vatiationalProblem}

We propose to solve the surrogate problem
\begin{equation}
\label{opt:sweetSpotApproximate}
    (\Pe) \,\,\left\{\,\,
    \begin{aligned}
        & \underset{\vu\in W_S\cap\P}{\text{minimize}}
        & & \Ae(\TD\vu)
    \end{aligned}\right.
\end{equation}

\begin{proposition}
    \label{prop:AepsContinuous}
    Suppose Assumption~\ref{assmpt:phi} holds. The function \(\Ae: \LInfO\to \R\) is continuous, and \((\Pe)\) has at least one minimizer.
\end{proposition}

\begin{proof}[Proof of Proposition~\ref{prop:AepsContinuous}]
    Let \(\delta > 0\), let \(v_0,v\in \LInfO\) be such that \(\nrm{v-v_0}_{\LInf} < \delta/2\), and let \(v'\) and \(v_0'\) be representatives. There exists \(\Omega^*\subset\Omega\) with \(\mu(\Omega\setminus\Omega^*) = 0\) such that \(|v'(x) - v'_0(x)| < \delta/2\) for \(x\in\Omega^*\). Since 
    \(\vphi_\eps\) 
    is non-negative and bounded on the interval \([-\nrm{v_0}_{\LInf} - \delta/2,\, \nrm{v_0}_{\LInf}  +\delta/2]\), for any \(x\in \Omega^*\) we have that
    \[
    |\Phi_{\eps}(v'(x)) - \Phi_{\eps}(v_0'(x))| \leq \int_{v_0'(x)-\delta/2}^{v_0'(x) + \delta/2}\vphi_\eps(t)\,dt \leq c_{\vphi_\eps}\delta
    \]
    where \(c_{\vphi_\eps} > 0\) depends only on \(\vphi_\eps\). As the bound is independent of the choice of  \(v',v_0'\), \(|\Ae(v) - \Ae(v_0)| \leq c_{\vphi_\eps}\mu(\Omega) 
    \delta\) whence \(\Ae\) is continuous. The existence of solutions follows from the compactness of \(W_S\cap \P\).
\end{proof}

\subsection{DC Formulation}
\label{sec:sweetSpot:DCC}

To solve \((\Pe)\) we first rewrite it equivalently as
\begin{equation}
    \label{opt:sweetSpotApproximateDCP}
    (\tPe) \,\, \left\{\,\,
    \begin{aligned}
        & \underset{\genfrac{}{}{0pt}{}{\vu\in W_S\cap \P}{v\in \LInfO}}{\text{minimize}}
        & & \Ae(v)\\
        & \text{subject to}
        & & \TD\vu \leq v.
    \end{aligned}\right.
\end{equation}
We interpret the auxiliary variable \(v\) as an overestimate of the auditory illusion map over \(\Omega\).

\begin{proposition}
    \label{prop:sweetSpotProblemEquivalence}
    Under Assumptions~\ref{assmpt:general} and~\ref{assmpt:phi}, if \(\vu^\star\) is an optimal solution to \((\Pe)\) then \((\vu^{\star}, \TD\vu^{\star})\) is an optimal solution to \((\tPe)\). In particular, \((\tPe)\) has a solution.
\end{proposition}

\begin{proof}[Proof of Proposition~\ref{prop:sweetSpotProblemEquivalence}]
    Let \(p_\eps\), \(\tilde{p}_\eps\) be the optimal values for \((\Pe)\) and \((\tPe)\) respectively, where \(\tilde{p}_\eps\) is finite as \((\tPe)\) is feasible and \(\Ae \geq 0\). On one hand, if \(\vu_\eps\) is an optimal solution to \((\Pe)\), which exists by Proposition~\ref{prop:AepsContinuous}, then \((\vu_\eps, \TD\vu_\eps)\) is feasible for \((\tPe)\). Hence, \(\tilde{p}_\eps \leq  p_\eps\). On the other, if \((\vu, v)\) feasible for \((\tPe)\) then \(\vu\) is feasible for \((\Pe)\). Since \(\Ae\) is monotone, \(p_\eps \leq \Ae(\vu) \leq \Ae(v) \) whence \(p_\eps \leq \tilde{p}_\eps\). We conclude \(\tilde{p}_{\eps} = p_\eps\) and \((\vu_\eps, \TD\vu_\eps)\) is an optimal solution to \((\tPe)\).
\end{proof}

Under suitable assumptions, the objective function in \((\tPe)\) is the {\em difference of convex functions}.

\begin{assumption}
\label{assmpt:phiDCC}
    In addition to Assumptions~\ref{assmpt:general} and~\ref{assmpt:phi}, there exists \(\vphi^+:\R\to\R\) absolutely continuous, non-decreasing and such that \(\vphi^+(t) = 0\) for \(t < 0\) and that \(\vphi^+-\vphi\) is a non-decreasing function.
\end{assumption}

We let \(\vphip(x) = \vphi^+(x/\eps)/\eps\) and we define
\[
\Phip(t)=\int_{-\infty}^t\vphip(s) ds.
\]
Similarly, let \(\vphim = \vphi - \vphip\) and \(\Phim(t) = \Phip(t) - \Phi_\eps(t)\). By construction, \(\Phi_\eps = \Phip - \Phim\).

\begin{proposition}
    \label{prop:dccDecomposition}
    Under Assumption~\ref{assmpt:phiDCC}, \(\Phip\) and \(\Phim\) are convex.
\end{proposition}

\begin{proof}[Proof of Proposition~\ref{prop:dccDecomposition}]
    By construction, both \(\Phip\) and \(\Phim\) have derivatives \(\vphip\) and \(\vphim\) respectively, which are both non-decreasing, hence 
    monotone~\cite[Proposition~17.10]{Bauschke2011}.
\end{proof}

Since \(\Phip,\Phim\) are convex, they are continuous. We decompose \(\Ae\) as \(\Ae = \Ap - \Am\) where
\[
    \Ap(v) := \int_\Omega \Phip(v(x))\, d\mu(x)
\]
and \(\Am\) is defined similarly. It is apparent that \(\Ap,\Am:\LInfO\to \R\) are convex and also continuous. We conclude \((\tPe)\) is a 
Difference-of-Convex (DC) program~\cite{Tao1997,Horst1999}. The Convex-Concave Procedure (CCCP)~\cite{Lipp2016} is an efficient method to attempt to find a solution to this class of optimization problems. The CCCP is an iterative method that uses an affine majorant for the concave part, e.g., using subgradients, to majorize the objective function in~\eqref{opt:sweetSpotApproximateDCP} by a convex function.

Let \(v_0\in \LInfO\) and let \(\LInfO^*\) be the topological dual of \(\LInfO\). We say \(g\in \LInfO^*\) is a {\em subgradient}  of \(\Am\) at \(v_0\) if
\[
\forall\, v\in \LInfO:\,\, \Am(v) \geq \Am(v_0) + g(v - v_0).
\]
The {\em subdifferential} \(\partial \Am(v_0)\) is the collection of all subgradients at \(v_0\). Since \(\Am:\LInfO\to \R\) is continuous and convex, it has a subdifferential at every \(v_0\)~\cite[Proposition~2.36]{Barbu2012}. We can use the convex majorizer
\[
\Ae(v) = \Ap(v) - \Am(v) \leq \Ap(v) -\Am(v_0) - g_{v_0}(v-v_0)
\]
where \(g_{v_0}\in\LInfO^*\) and solve
\begin{equation*}
    (\tP_{\eps,v_0}) \,\,\left\{\,\,
    \begin{aligned}
        & \underset{\genfrac{}{}{0pt}{}{u\in W_S\cap\P}{v\in \LInfO}}{\text{minimize}}
        & & \Ap(v) - \Am(v_0) - g_{v_0}(v - v_0)\\
        & \text{subject to}
        & & Tu \leq v.
    \end{aligned}\right.
\end{equation*}

\begin{proposition}
    \label{prop:convexMajorizer}
    Under Assumption~\ref{assmpt:phiDCC}, for \(\eps > 0\) and \(v_0\in \LInfO\) \((\tP_{\eps,v_0})\) has at least one optimal solution.
\end{proposition}

We defer the proof to Appendix~\ref{proof:convexMajorizer}. In our method we make an explicit choice of a subgradient. 

\begin{proposition}
    \label{prop:subgradient}
    Under Assumption~\ref{assmpt:phiDCC},  for \(v_0\in \LInfO\)  the linear application
    \begin{equation}
        \label{eq:subgradient}
        g_{v_0}(v) := \int_\Omega \vphim(v_0(x)) v(x)\, d\mu(x)
    \end{equation}
    is well-defined for \(v\in\LInfO\), continuous and is a subgradient for \(\Am\) at \(v_0\).
\end{proposition}

\begin{proof}[Proof of Proposition~\ref{prop:subgradient}]
    Since \(v_0\in\LInfO\) and \(\vphim\) is non-decreasing, \(\vphim\circ v_0\in \LInfO\). Hence, \(g_{v_0}\) is well-defined and \(g_{v_0}\in\LInfO^*\). By the definition of \(\Phim\) and Assumption~\ref{assmpt:phi}, \(\Phim(t) \geq \Phim(t_0) + 
    \vphim(t_0)(t - t_0)\) for all \(t,t_0\in \R\). This implies
    \begin{equation*}
    \int_{\Omega}\Phim(v(x))\, d\mu(x) \geq \int_{\Omega}\Phim(v_0(x))\, d\mu(x)  + \int_{\Omega} \vphim(v_0(x))(v(x) - v_0(x))\, d\mu(x).
    \end{equation*}
\end{proof}

Given \((\vu_0, v_0)\) the CCCP constructs a sequence \(\set{(\vu_k, v_k)}_{k\in\N}\) where \((\vu_{k+1}, v_{k+1})\) is the optimal solution to \((\tP_{\eps,v_{k}})\). To our knowledge, the best theoretical guarantees for finite-dimensional problems show that this sequence converges to a stationary point of \((\tPe)\)~\cite[Theorem~3]{Tao1997}, whereas we are not aware of similar guarantees for infinite-dimensional problems. However, \(\set{\vu_k}_{k\in\N}\) is a sequence in \(W_S\)  and we can extract a subsequence \(\set{\vu_{k(\ell)}}_{\ell\in\N}\) with limit \(\vu^\opt\). If \(\set{\tilde{p}^\star_{k}}_{k\in\ell}\) is the sequence of optimal values to each \((\tP_{\eps,v_k})\) then 
\begin{equation*}
    \Ae(\TD\vu^{\opt}) = \lim\inf_{\ell\to\infty}\Ae(\TD\vu_{k(\ell)})
    \leq \lim\inf_{\ell\to\infty}\Ae(v_{k(\ell)})  \leq \lim\inf_{k\to\infty}\, \tilde{p}^\star_{k(\ell)}
\end{equation*}
by the continuity of \(\TD,\Ae\) and the fact that \(\Ae\) is non-decreasing. The optimal values are a conservative estimate of \(\Ae(\TD\vu^{\opt})\). Our numerical results show the solutions found this way performs well in practice.

\subsection{SWEET-ReLU}
\label{sec:sweetReLUInstance}

When \(\vphi\) is the indicator function of \([0, 1]\) the function \(\Phi\) becomes the difference of two \textit{Rectified Linear Units} (ReLUs). In this case, \(\vphi_\eps = \eps^{-1}\chi_{[0, \eps]}\). The decomposition \(\Phi_\eps = \Phip - \Phim\) becomes
\[
\Phip(x) = x_+/\eps \quad\mbox{and}\quad \Phim(x) =  (x - \eps)_+/\eps
\]
and the subgradient~\eqref{eq:subgradient} becomes
\[
g_{v_0}(v) = \frac{1}{\eps}\int_{\set{x\in\Omega: v_0(x) > \eps}} v(x)\, d\mu(x).
\]
Let \(\Omega_{\eps,v_0}:= \set{x\in\Omega: v_0(x) \leq \eps}\). Since both \(\Am(v_0)\) and \(g_{v_0}(v_0)\) in \((\tP_{\eps,v_0})\) are constant, it suffices to compute
\begin{align*}
    \Ap(v) - g_{v_0}(v) &= \frac{1}{\eps}\int_{\Omega} v(x)_+\, d\mu(x) -\frac{1}{\eps} \int_{\Omega_{\eps,v_0}^c} v(x)\, d\mu(x) \\
    &= \frac{1}{\eps} \int_{\Omega_{\eps,v_0}} v(x)_+\,d\mu(x)+ \frac{1}{\eps} \int_{\Omega_{\eps,v_0}^c} (-v(x))_+\,d\mu(x)
\end{align*}
where we used the fact that \(t_+ - t = (-t)_+\). The second term is non-negative, and becomes positive only when \(v\) takes negative values. As \(Tu \leq v\) in  \((\tP_{\eps,v_0})\) we can choose \(v\) arbitrarily large on \(\Omega_{\eps,v_0}^{c}\) to decrease the objective value and to neglect the second integral. Then, only the first term contributes to the objective in \((\tP_{\eps,v_0})\) and we obtain
\[
    (\tP_{\eps,v_0}) \,\,\left\{\,\,
    \begin{aligned}
        & \underset{\genfrac{}{}{0pt}{}{\vu\in W_S\cap\P}{v\in \LInfO}}{\text{minimize}}
        & & \int_{\Omega_{\eps, v_0}} v(x)_+\, d\mu(x)\\
        & \text{subject to}
        & & \TD\vu \leq v,\, 0 \leq v|_{\Omega_{\eps, v_0}^{c}}.
    \end{aligned}\right.
\]
As the positive-part function is monotone, we can eliminate \(v\) to obtain
\[
    (\tP_{\eps,v_0})\,\,\left\{\,\,
    \begin{aligned}
        & \underset{\vu\in W_S\cap\P}{\text{minimize}}
        & & \int_{\Omega_{\eps, v_0}} (\TD\vu(x))_+\, d\mu(x)
    \end{aligned}\right.
\]
which is precisely the problem \((P^{\text{SReLU}}_1)\) when \(v_0 \equiv 0\). 
It depends on \(v_0\) only through \(\Omega_{\eps, v_0}\). To construct an optimal solution \((\vu_{k}, v_{k})\) from the optimal solution \(\vu_k\) we proceed as follows: by choosing \(v_{k}|_{\Omega_{\eps,v_{k-1}}} = \TD\vu_{k}|_{\Omega_{\eps,v_{k-1}}}\) and \(v_{k}|_{\Omega_{\eps,v_{k-1}}^c}= \max\set{\eps, \TD\vu_{k}|_{\Omega_{\eps,v_{k-1}}^c}}\) we obtain
\begin{equation*}
    \Omega_{\eps,v_{k}} = \set{x\in\Omega: v_{k}'(x) \leq \eps}
    = \set{x\in \Omega_{\eps,v_{k-1}}: \TD\vu_{k}'(x) \leq \eps} 
    = \Omega_{\eps,v_{k-1}}\cap \set{x\in \Omega: \TD\vu_{k}'(x) \leq \eps},
\end{equation*}
yielding the method as presented in Section~\ref{sec:sweetReLU}.

\subsection{A class of monoaural dissimilarity maps}
\label{sec:filtersToDissimilarity}

We introduce a class of  dissimilarity metrics based on time-variant filters satisfying Assumption~\ref{assmpt:general}.

\begin{lemma}
\label{lem:kernels}
    Let \(K:\Omega\times \Th \to \LE(\R^2)\) be continuous with
    \[
    \sup_{(x,\th)\in \Omega\times\Th} \sup_{t\in\R}\int_{\R} (|K_{(x,\th)}(t,t')| + |K_{(x,\th)}(t',t)|)\, dt'
    \]
    finite. Define for \((x,\th)\in\Omega\times\Th\), \(w\in\LE(\R)\)
    \[
        A^K_{(x,\th)} w(t) := \int K_{(x,\th)}(t,t') w(t')\, dt'.
    \]
    Then \(A^K_{(x,\th)}\) is linear, \(A^K_{(x,\th)}w\in\LE(\R)\) and \((x,\th,w) \mapsto A^K_{(x,\th)} w\) is continuous.
\end{lemma}

\begin{proof}
    The linearity follows from the integral representation. From Young's inequality for integral operators~\cite[Theorem~0.3.1]{Sogge2017}
    \[
    \int_{\R} \left|\int_{\R} K_{(x,\th)}(t,t') w(t')\,dt'\right|^2\, dt \leq C_K(x,\th)^2\nrm{w}_{\LE}^2
    \]
    whence \((x,\th)\mapsto A_{(x,\th)}w\) is bounded. By the Cauchy-Schwarz inequality
    \begin{align*}
        \nrm{A_{(x,\th)}w' - A_{(y,\phi)} w}_{\LE}^2 &
        \leq 2\int_{\R} \left|\int_{\R} (K_{(x,\th)}(t,t') -  K_{(y,\phi)}(t,t'))w'(t')\,dt'\right|^2\, dt\\
        &\quad +2\int_{\R} \left|\int_{\R} K_{(y,\phi)}(t,t')(w'(t') - w(t'))\,dt'\right|^2\, dt\\
        &\leq 2\nrm{w}_{\LE}^2 \iint_{\R^2} | K_{(x,\th)}(t,t') -  K_{(y,\phi)}(t,t')|^2 dt dt'
        + 2C_K(x,\th) \nrm{w' - w}_{\LE}^2,
    \end{align*}
    where continuity follows from hypothesis. 
\end{proof}

\begin{proposition}
    \label{prop:filtersToDissimilarity}
    Let \(\vuo\in W\) and let \(\set{B_{k}}_{k=1}^{\Nb}\) be as in~\eqref{eq:hearingBandGeneral} where \(\set{K_{B_k}}_{k=1}^{\Nb}\) satisfy the hypotheses of Lemma~\ref{lem:kernels}. Let \(\Psi:\Rp^{\Nb}\to \R\) be convex and monotone increasing on each one of its arguments. For \(s\in\set{\ell,r}\) 
    \begin{equation*}
        D_{(\vu,\vuo)}^{s}(x,\th) = \Psi(B_1(\uside_{(x,\theta)} - \uside_{0,(x,\theta)}),
        \ldots, B_{\Nb}(\uside_{(x,\theta)} - \uside_{0,(x,\theta)}))
    \end{equation*}
    is a dissimilarity satisfying Assumption~\ref{assmpt:general}. In particular, so is
    \[
        D_{(\vu,\vuo)}(x,\th) = \max\set{D_{(\vu,\vuo)}^\ell(x,\th), D_{(\vu,\vuo)}^r(x,\th)}.
    \]
\end{proposition}

\begin{proof}
    For simplicity, we prove the result for \(\vuo = 0\). It can be verified that \(B_k \uside_{(x,\theta)} = \nrm{A^{K_{B_k}}_{(x,\theta)}\uside_{(x,\theta)}}_{\LE}^2\). By Lemma~\ref{lem:kernels}, a function of the form
    \begin{equation*}
        D^s_{(\vu,0)}(x,\th)
        =\Psi(\nrm{A^{K_{B_1}}_{(x,\th)} \uside_{(x,\th)}}_{\LE}, \ldots, \nrm{A^{K_{B_{\Nb}}}_{(x,\th)} \uside_{(x,\th)}}_{\LE})
    \end{equation*}
    is continuous on \(\Omega\times\Th\). Since \(A^{K_{B_k}}_{(x,\theta)}\) is linear on \(\uside_{(x,\th)}\) and the norm is convex, the convexity of \(D^s\) follows from the assumptions on \(\Psi\).
\end{proof}

\subsection{Discussion}

Although Proposition~\ref{prop:layerCake} implies \(\Ae\) converges pointwise to \(\mu\circ \S\), this is not sufficient to ensure a global minimizer for 
\((\Pe)\) converges to a global minimizer for \((\Po)\). A future line of work consists on leveraging \(\Gamma\)-convergence to answer this question. Related to it is the choice of \(\vphi\). Although we have not studied extensively the effect of this choice, we believe it affects the quality of the approximation to the area of the sweet spot. This is another interesting future line of work.

In addition, we would like to point out that the advantage of solving a sequence of approximate problems for decreasing values of \(\eps\) over solving a single approximate problem for \(\eps\ll 1\) lies in the fact that each approximate problem is non-convex. Then, to find an approximate solution to it, the DCCC algorithm needs an initial point. This, coupled with the observation that for a sufficiently large \(\eps\) the problem becomes equivalent to a convex problem, solving a sequence of approximate problems is a systematic way to find an initial point for the approximate problems defined for \(\eps\ll 1\).

Finally, the method allows for several choices of \(\mu\). Therefore, the results presented apply both for the continuous case, e.g., when \(\mu\) is the Lebesgue measure, and the discrete case, e.g., when \(\mu\) is discrete. 

\subsection{Proof of Proposition~\ref{prop:wellPosedA}}
\label{proof:wellPosedA}

\noindent{\em Proof of~(i).} It is apparent that \(W_S\) is convex. Since \(\Omega\times \Th\) is a separable metric space, if \(W_S\) is bounded and equicontinuous, by Arzel\`a-Ascoli's theorem~\cite[Theorem~11.28]{Rudin1986} it follows it is compact. By Assumption~\ref{assmpt:general},
\begin{align*}
    \nrm{u^{s}_{(x,\th)}}_{\LE}^2 &\leq \Ns\sum_{k} \int_{I_S} |\wh{c}_k(f)|^2 |H^s_k(f,x,\th)|^2\, df \\
    &\leq \Ns\cmax^2\sum_{k}\sup_{(f,x,\th)\in I_S\times \Omega\times\Th}|H_k^s(f,x,\th)|^2 
\end{align*}
and \((x,\th)\mapsto u_{(x,\th)}\) is uniformly bounded. Thus, \(W_S\) is bounded. Let \(\eps > 0\). Since \(\wh{H}^s_k\) is continuous on the compact set \(I_S\times \Omega\times \Th\), there is \(\delta > 0\) such that for any \(|x - y|, |\th-\phi| < \delta\) and \(f\in I_S\) we have that \(|\wh{H}_k^s(f, x, \th) - \wh{H}_k^s(f,y, \phi)| < \eps / 2\Ns^2 \cmax^2\). Then, 
\begin{equation*}
    \nrm{u^s_{(x,\th)} - u^s_{(y,\phi)}}_{\LE}^2
    \leq \Ns \sum_{k}\int_{I_S} |\wh{c}_k(f)|^2 |H^s_k(f,x,\th) - H^s_k(f,y,\phi)|^2\, df <\frac{1}{2}\eps
\end{equation*}
whence \((x,\th)\mapsto \vu_{(x,\th)}\) is continuous. Since\(\delta\) is independent of \(\vu\), we conclude \(W_S\) is equicontinuous.

\noindent{\em Proof of~(ii).} Let \(\vu,\vuo \in W\). The function \(D_{(\vu,\vuo)}\) is continuous on the compact set \(\Omega\times\Th\) and thus bounded. Hence, \(\TD\vu\) is bounded and we can associate to it its equivalence class in \(\LInfO\). Let \(\eps > 0\). There exists \(\delta >0\) such that \(\nrm{\vleft -\uleft}_{\LE},\nrm{\vright- \uright}_{\LE} < \delta\) implies \(|D_{(\vv,\vuo)}(x,\th) - D_{(\vu,\vuo)}(x,\th)| < \eps/2\). This implies \(|\TD\vv(x) - \TD\vu(x)| < \eps\) whence \(\TD\) is continuous. By Assumption~\ref{assmpt:general}, the map \(D\) is convex on its first argument. The conclusion follows from
\begin{equation*}
    \sup_{\th\in\Th}\, D_{(\lambda\vu_1+(1-\lambda)\vu_2,\vuo)}(x,\th)
    \leq \lambda \sup_{\th\in\Th}\,D_{(\vu_1,\vuo)}(x,\th) + (1-\lambda)\sup_{\th\in\Th}\, D_{(\vu_2,\vuo)}(x,\th).
\end{equation*}

\noindent{\em Proof of~(iii).} For \(\TD u\in \LInfO\) we can choose representatives \(\TD u', \TD u''\) of \(\TD u\). Hence, the set \(\set{x\in\Omega:\, \TD u'(x) = \TD u''(x)}\) has \(\mu\)-measure zero, from where the conclusion follows.

\noindent{\em Proof of~(iv).} From the same arguments used for (ii) the map \(\TL:W\to \LInfO\) is continuous. As \(\set{v\in\LInfO: \mbox{\(v(x) > 0\) \(\mu\)-a.e.}}\) is open, then \(\P^c = \set{\vu\in W: \mbox{\(\TL \vu(x) > 0\) \(\mu\)-a.e.}}\) is open, whence \(\P\) is closed.

\subsection{Proof of Proposition~\ref{prop:convexMajorizer}}
\label{proof:convexMajorizer}

\def\vvopt{v'{}^{\star}}
\def\vopt{v^{\star}}

Let \(v'_0\) be a representative and define
\[
    f(x,\alpha) = \Phip(\alpha) - \vphip(v'_0(x)) \alpha.
\]
This is a Carath\'eodory map~\cite[Definition~8.2.7]{Aubin1990}. By Theorem~8.2.11 in~\cite{Aubin1990} there exists \(\vvopt\) measurable such that
\[
   \Phip(\vvopt(x)) - \vphip(v'_0(x)) \vvopt(x) = \inf\set{f(x, \alpha): \alpha\in\R}.
\]
By Assumption~\ref{assmpt:phiDCC},  \(\vphip(v'_0(x)) \geq 0\) and \(0 \leq \vvopt(x) \leq v_0'(x)\) whence \(\vvopt\) is bounded \(\mu\)-a.e. Let \(\vopt\in\LInfO\) be its equivalence class. Let \(\set{(\vu_k,v_k)}_{n\in\N}\) be a minimizing sequence. As \(W_S\cap\P\) is compact, without loss of generality we may assume \(\set{\vu_k}_{k\in\N}\) has a limit \(\vu_\infty\). Let \(\TD\vu'_k\) be a representative and let \(w_k' := \max(\vvopt, \TD u'_k)\). Then \(w'_k\) is \(\mu\)-a.e. bounded. Let \(w_k\in\LInfO\) denote its equivalence class. By construction,
\[
    f(x,v'_k(x)) \geq f(x,\TD \vu_k'(x)) \geq f(x,\TD\vvopt(x))
\]
for any representative \(v'_k\).  Therefore
\begin{equation*}
    \lim\inf_{k\to\infty} (\Ap(v_k) - g_{v_0}(v_k))
    \geq \lim\inf_{k\to\infty} (\Ap(w_k) - g_{v_0}(w_k))
\end{equation*}
whence \(\set{(\vu_k,w_k)}_{k\in\N}\) is also minimizing. By continuity of \(\TD\) we conclude \(w_k \to \max\set{\vopt, \TD\vu_\infty}\) whence \((\tP_{\eps,v_0})\) has a solution.

\end{document}